\documentclass[12pt]{article}
\usepackage{amsmath}
\usepackage{amsthm}
\usepackage{lipsum}
\usepackage{graphicx}
\usepackage{enumitem}
\usepackage{bbm}
\usepackage{xcolor}
\usepackage{array,multirow,booktabs}
\usepackage{graphics}
\usepackage{xr-hyper}
\usepackage{caption}
\usepackage{subcaption}
\usepackage{filecontents}
\usepackage{natbib}
\bibpunct[]{(}{)}{;}{a}{}{}

\newcommand{\blind}{0}

\makeatletter
\newcommand*{\addFileDependency}[1]{% argument=file name and extension
\typeout{(#1)}% latexmk will find this if $recorder=0
\@addtofilelist{#1}
\IfFileExists{#1}{}{\typeout{No file #1.}}
}\makeatother
\newcommand*{\myexternaldocument}[1]{%
\externaldocument{#1}%
\addFileDependency{#1.tex}%
\addFileDependency{#1.aux}%
}
\myexternaldocument{supplement}

\newtheorem{theorem}{Theorem}
\newtheorem{proposition}{Proposition}

\newtheorem{remark}{Remark}

\newtheorem{condition}{Condition}

\usepackage{notation-main-binary}

\begin{document}

\def\spacingset#1{\renewcommand{\baselinestretch}%
{#1}\small\normalsize} \spacingset{1}

%%%%%%%%%%%%%%%%%%%%%%%%%%%%%%%%%%%%%%%%%%%%%%%%%%%%%%%%%%%%%%%%%%%%%%%%%%%%%%

\if0\blind
{
  \title{\bf Robust Design-Based Estimation and Inference for Stratified Randomized Trials with Varying Cluster Sizes}
  \author{Xinhe Wang \\
    and \\
    Ben B. Hansen\thanks{
    Research supported by the Institute of Education Sciences, U.S. Department of Education, through Grant R305D210029 to the University of Michigan. This report represents views of the authors, not the Institute or the U.S. Department of Education.} \\
    Department of Statistics, University of Michigan}
  \date{}
  \maketitle
} \fi
\if1\blind
{
  \bigskip
  \bigskip
  \bigskip
  \begin{center}
    {\LARGE\bf Robust Design-Based Estimation and Inference for Stratified Randomized Trials with Varying Cluster Sizes}
\end{center}
  \medskip
} \fi

\bigskip
\begin{abstract}
Clustered randomized controlled trials are often stratified or pair-matched to improve covariate balance and efficiency. Sample average treatment effects (SATEs) are commonly estimated by averaging stratum-level treatment-control mean contrasts---an approach that is natural and widely used. We show that, in stratified clustered trials with heterogeneous cluster sizes, such estimators need not be consistent for the SATE. They can converge to the wrong limit even under correct randomization and without model misspecification. The source is a covariance between cluster sizes and treatment effects: stratumwise averaging mis-weights clusters in a way that produces bias of constant order, regardless of sample size.
We study the H\'ajek (ratio) estimator as a robust alternative. By aggregating outcomes within treatment groups before taking their difference, it remains consistent in clustered trials that grow by increasing strata sizes or the number of strata. Despite that, its use in design-based analyses of clustered trials has been limited by the lack of variance estimators. We develop a design-based variance estimator that applies to any number of strata of any size, and show that it is asymptotically conservative, a property that holds even when some strata contain only a single treated or control unit. We also present tests improving the coverage of Wald tests when the number of clusters is moderate. The framework extends naturally to covariate-adjusted estimators via a variance orthogonality property. 
\end{abstract}

\noindent%
{\it Keywords:} generalized score test, heterogeneous treatment effect, Huber-White standard error, model-assisted inference, Neyman-Rubin causal model
\vfill

\newpage
\spacingset{1} % DON'T change the spacing!

\section{Introduction}
\label{sec:mov}

In clustered randomized controlled trials (cRCTs), fine stratification and pair matching are commonly recommended designs to improve efficiency \citep{raudenbush1997statistical, imai2009essential, imbens2011experimental}. However, when cluster sizes differ within strata or pairs, commonly used estimators can fail to identify average treatment effects.

Estimators that average stratum-specific treatment-control contrasts using stratum weights --- including the widely used fixed effects estimator, as well as the weighting by stratum size estimator of \citet{imai2009essential} --- are not generally consistent. In paired or finely stratified cRCTs, they fail to converge on the average treatment effect
whenever cluster sizes vary within strata and correlate with treatment effects. Such conditions are natural in practice: within-stratum alignment of cluster sizes may be incompatible with operational constraints or competing investigator priorities \citep[e.g.,][discussed in Section~\ref{secappl}]{Giles2012}; large and small clusters may differ systematically in baseline outcomes or intervention responsiveness \citep[e.g.,][]{song2019effect, bugni2025inference}.

Previous work either underestimates the problem or mis-diagnoses it. For example, \citet{imai2009essential} recommend pair matching on cluster size to mitigate bias, and \citet{de2024level} develop variance estimators for the fixed effects estimator that address overrejection. While valuable, these do not resolve the more fundamental issue that the point estimator can be inconsistent if cluster sizes vary within strata.

To see what kinds of estimator maintain consistency in cRCTs, consider the source of the bias: stratum-averaging estimators treat strata as the primary units of aggregation, but cluster weights are not properly represented in small strata. The H\'ajek (or ratio) estimator \citep{hajek1971comment} avoids the pitfall by aggregating cluster outcomes within treatment conditions first and then taking the difference, preserving the role of cluster weights. It is also closely connected to weighted least squares and to classical survey-sampling estimators \citep{cochran1977, lohr2021sampling}, yet its adoption in design-based analyses of stratified clustered trials has been limited by the lack of variance estimators that accommodate arbitrary stratum sizes---a gap that is especially acute when strata are small and within-stratum sample variances cannot be computed \citep{fogarty2018finely, pashley2021insights}.

In this paper, we characterize the inconsistency of common stratum-averaging estimators, clarify the conditions under which it arises, and quantify the asymptotic bias. We develop a design-based inferential framework for the H\'ajek estimator, including its consistency under both standard asymptotic regimes for stratified experiments (growing stratum sizes and growing numbers of strata), a novel variance estimator that accommodates arbitrary stratum numbers and sizes and is asymptotically conservative, a testing procedure for small to moderate samples, and a straightforward extension to covariate-adjusted estimation. Our results provide practical guidance for the analysis of finely stratified cRCTs by clarifying when stratum-averaging estimators can fail and when ratio-based estimators are preferable.

Although motivated by clustered trials, our framework applies more broadly to stratified RCTs with unequal unit weights.  This includes survey experiments, with subjects carrying sampling weights \citep[e.g.,][]{holbrook2010social,hanmer2014experiments}, and internet experiments with feature-rich non-probability samples \citep[e.g.,][]{weinbergKapelner22}, where calibration weights may foster generalizability \citep{tiptonOlsen22}. Rich baseline information recommends paired randomization \citep{kapelner2025pairwise}, the design with which Section~\ref{sec:inconsistency} demonstrates the inconsistency of stratum-averaging estimators.

\section{Background, notation, and framework}
\label{sec2}

\subsection{Related work}

Prior work has examined efficiency gains from blocking or matching in randomized experiments \citep{raudenbush1997statistical, hansen2008covariate, imbens2011experimental, wang2023rerandomization}. In cRCTs, the Horvitz-Thompson estimator \citep[HT,][]{horvitz1952} is unbiased for the sample average treatment effect (SATE) but can have high variance \citep{middleton2015unbiased, mann2024general}, while work on estimator choice has shown that unequal cluster sizes materially affect both estimands and estimator performance \citep{su2021model}. Our paper sharpens this picture: once fine stratification is combined with heterogeneous cluster sizes, stratum-averaging estimators can incur $O(1)$ asymptotic bias for the SATE.

A related strand studies inference for stratified experiments. \citet{fogarty2018finely} develops randomization-based inference for non-clustered finely stratified designs. \citet{pashley2021insights} study variance estimator for designs containing fine strata but require the number of fine strata to be large. \citet{Liu2020} develop covariate adjustment methods for the stratum-averaged differences-in-means estimator in stratified RCTs. More recently, \citet{baiLiuShaikhTabord-Meehan24} study SATE estimation in paired cRCTs under a super-population framework. \citet{de2024level} recommend clustering standard errors at the pair level for the H\'ajek estimator in paired cRCTs. 
Among the closest related papers, \citet{schochet2022design} analyze ratio and weighted least-squares estimators in blocked cRCTs and derive variance estimators under a fixed number of large strata, while \citet{mann2024general} develop improved HT-based estimators for paired cRCTs. We study the H\'ajek effect estimator under the finite-population framework and develop a design-based variance estimator that applies to any number of strata of any sizes.

\subsection{Notation and Assumptions}

Consider an RCT of $n$ allocation units, persons or clusters of people according as the trial has person- or cluster level random assignment. Let $\wps>0$ be a \textit{priori} weight associated with unit $\s$, such as the cluster size in a cRCT.
Suppose all units are arranged into $\B\geq 1$ strata. Stratum $\p$ contains $\np$ units, with a simple random sample of $\npt$ assigned to treatment and the remaining $\np-\npt=\npc$ to control. 
Let $\Zps$ denote the random treatment assignment of unit $\s$, with $\Zps=1$ indicating treatment and $\Zps=0$ indicating control. For $z\in\{0,1\}$, define $\psk=\npk/\np\in(0,1)$ as the probability assignment to treatment or control, where $\s\in\p$. (Section~\ref{sec:covadj} presents a method adjusting for cluster- and person-level covariates, with additional notation.)
The potential outcome of unit $\s$ under treatment $z$ is $\ypsk$ for $z=1,0$. Either $\ypst$ or $\ypsc$ is observed, depending on $\zps$, the realization of treatment assignment $\Zps$. In cRCTs, $\ypsk$ is typically the cluster mean. Let $\tau_i = \ypst - \ypsc$ be the treatment effect of unit $\s$.

We make the no interference assumption \citep{cox1958planning,rubin1990comment}: a unit's observed outcome depends only on its own treatment assignment. For a cRCT, this excludes interference across clusters. Under this assumption, each unit's observed outcome can be expressed as $y_{iz_i}=\zps \ypst + (1-\zps)\ypsc$, without dependence on other units' treatment assignments. We also assume complete randomization within strata: treatment groups are selected uniformly at random from the $\binom{\np}{\npt}$ possible allocations, with treatment assignments mutually independent across strata.
The parameter of main interest is the SATE:
\begin{equation} 
\label{eq:rho_k-def}
\tau  = \frac{\ssn\wps(\ypstt - \ypsc)}{\ssn\wps}= \frac{\ssn\wps \ypstt}{\ssn\wps} - \frac{\ssn\wps \ypsc}{\ssn\wps} = \rho_1 - \rho_0.
\end{equation} 
Here, $\rho_1$ and $\rho_0$ denote weighted means of potential outcomes and are also parameters of interest.
They are unknown as we observe only one of the potential outcomes of each unit. 

\section{Inconsistency of common stratum-averaging estimators}
\label{sec:inconsistency}
Following \citet[][Ch~5]{fleiss2013smr}, %,rosenbaum2017observation
we use the term direct adjustment for estimators of the form:
\begin{equation}
\label{eq:sadm}
\sum_{b=1}^B \tilde w_b \hat\tau_b = \sum_{b=1}^B \tilde w_b \left\{\frac{\sum_{i\in b} Z_i w_i y_{iz_i}}{\sum_{i\in b} Z_i w_i} - \frac{\sum_{i\in b} (1-Z_i) w_i y_{iz_i}}{\sum_{i\in b} (1-Z_i) w_i} \right\},
\end{equation}
where $\tilde w_{\p}\ge0$ are stratum weights satisfying $\spb \tilde w_{\p} = 1$, and $\hat\tau_b $ is the difference between weighted mean outcomes under treatment and control in stratum $b$. When the weights $\tilde w_b$ are proportional to the total cluster size in stratum $b$, \eqref{eq:sadm} coincides with the estimator recommended by Imai, King, and Nall (2009). When the weights are proportional to the harmonic mean of the treatment and control group sizes within each stratum, \eqref{eq:sadm} coincides with the treatment coefficient from a linear regression of outcomes on a treatment indicator and stratum fixed effects. For brevity, we refer to these as the IKN and fixed effects (FE) estimators, respectively.

These estimators are natural and appealing. In non-clustered stratified experiments, stratum-level difference-in-means estimators are unbiased for stratumwise treatment effects, and their weighted average is unbiased for the SATE under appropriate weighting. Under complete randomization, these estimators coincide with the common difference-in-means and are unbiased from the design-based perspective. Basically, such estimators are well behaved in non-clustered or coarsely stratified experiments \citep{imbens2015causal, ding2024first}. It is therefore unsurprising that they have become default choices in practice.

Despite their familiarity and wide use, these estimators can be inconsistent for the SATE in cRCTs with many small strata and heterogeneous cluster sizes. %To illustrate, consider a simple paired cRCT. In one pair, suppose one cluster contains 200 students and the other 50. Under treatment, the larger cluster has outcome 1 and the smaller cluster has outcome 5, whereas under control both cluster outcomes are 0. The pair-specific SATE is therefore $\{200(1)+50(5)\}/(200+50)=1.8$. By contrast, the IKN estimator equals 1 with probability 0.5 and 5 with probability 0.5, so its expectation is $0.5(1)+0.5(5)=3$, having bias $3-1.8=1.2$ for this pair. This contributes to the total bias and does not vanish as the number of pairs grows.
The problem is structural: it is not caused by model misspecification, but by how these estimators aggregate information across strata, and it is not resolved by changing the estimand.
Even when the target is
\[
\sum_{b=1}^B \tilde w_b \, \tilde\tau_b,
\]
where $\tilde\tau_b$ denotes the stratum mean treatment effect, \eqref{eq:sadm} need not be consistent as the number of strata grows while stratum sizes remain small.
\begin{proposition}
\label{prop:wasdom}
    In a paired cRCT, denote $\bwp$ and $\bar\tau_\p$ as the simple averages of cluster weights and treatment effects within pair $\p$. The IKN effect estimator targeting the SATE $\spb \tilde{w}_{\p}^\text{(s)} \tilde\tau_\p$, where $\tilde{w}_{\p}^\text{(s)} \propto \sump\wps$ satisfy $\spb\tilde{w}_{\p}^\text{(s)} =1$, has bias
    \begin{equation}
        \frac{1}{\ssn \wps} \spb \sump (\wps-\bwp)(\tau_i-\bar\tau_\p).
        \label{eq:wasdom-bias}
    \end{equation}
    The FE estimator targeting $\spb \tilde{w}_{\p}^\text{(f)} \tilde\tau_\p$, where $\tilde{w}_{\p}^\text{(f)}$ is proportional to the harmonic mean of the two cluster sizes in pair $\p$ with $\spb \tilde{w}_{\p}^\text{(f)} =1$, has bias
    \begin{equation}
        \spb \tilde{w}_{\p}^\text{(f)} \cdot \frac{1}{\sump \wps} \sump (\wps-\bwp)(\tau_i-\bar\tau_\p).
        \label{eq:wasdom-bias2}
    \end{equation}
\end{proposition}

Biases \eqref{eq:wasdom-bias} and \eqref{eq:wasdom-bias2} are driven by the within-stratum covariance between cluster sizes $\wps$ and treatment effects $\tau_i$: when this covariance is nonzero, $\hat\tau_b$ systematically mis-weights clusters, producing a bias that does not vanish as the number of strata grows.
\begin{remark}
\label{rmk:inconsistency-cond}
    The biases \eqref{eq:wasdom-bias} and \eqref{eq:wasdom-bias2} remain bounded away from zero as the number of pairs tends to infinity if either of the following hold:
    \begin{enumerate}
        \item Cluster sizes $\wps$ are uniformly bounded while $\lim\inf_{\B\to\infty}\sum_{{\p=1,  i\in\p}}^\B \cov(\wps, \tau_i)/\B \ge c$ for some constant $c>0$.
        \item For $\p=1,\ldots, \B$ and $i \in \p$, the pairs $(\wps, \tau_i)$ are drawn i.i.d.\ from a distribution under which the correlation between $w$ and $\tau$ is nonzero, the coefficient of variation of $w$ is nonzero, and the standard deviation of $\tau$ is nonzero. 
    \end{enumerate}
\end{remark}

These conditions hold whenever cluster sizes vary substantially and larger clusters respond differently to treatment, which is a natural scenario in practice \citep[e.g.,][]{hayes2017cluster}, making the risk of bias diagnosable by comparing IKN and H\'ajek estimates.

These results accord with \citet{de2024level}, who note inflated test sizes for fixed effects (FE) estimators in paired cRCTs, a finding at odds with their claim that suitable standard errors alleviate the size problem. We offer a more fundamental diagnosis: whatever the challenges of variance estimation, the problem is with the point estimator itself. When cluster sizes vary and correlate with treatment effects, the FE estimator is inconsistent for the SATE regardless of how the standard error is computed; no improvement in variance estimation can correct an $O(1)$ asymptotic bias. Matching cluster sizes within pairs, as recommended by \citet{imai2009essential}, mitigates the problem when feasible, but does not resolve it when perfect size balance within strata is unavailable.

The H\'ajek estimator does not estimate $\tau_{\p}$ en route to estimating $\tau$. In sampling theory, H\'ajek estimators are known to have a bias decreasing at rate \(O(n^{-1})\) in simple random samples \citep{cochran1977, lohr2021sampling}. In cRCTs, the bias of the H\'ajek effect estimator decreases at \(O(n^{-1})\), as Section~\ref{sec3-1} shows, even under Remark \ref{rmk:inconsistency-cond}'s conditions, when stratum-averaging estimators \eqref{eq:sadm} can incur \(O(1)\) biases. The inconsistency of \eqref{eq:sadm} is therefore not a pathological edge case: it arises generically whenever cluster sizes vary and correlate with treatment effects within strata, a condition natural and empirically common in cRCTs. %We examine the magnitude of this bias relative to the H\'ajek estimator in simulations in Section~\ref{sec:OSNAPsimuATE}.

\section{Design-based inference for the H\'ajek effect estimator}
\label{sec3}

\subsection{Linearization and asymptotic normality}
\label{sec3-1}

The H\'ajek effect estimator estimates $\rho_1$ and $\rho_0$ separately, incorporating inverse assignment probabilities:
\begin{equation*}
  \hat\tau = \hat\rho_1 - \hat\rho_0; \quad\hat{\rho}_{z} = \left. \ssn\frac{\Indk}{\psk}\wps \ypsk \middle/ {\ssn \frac{\Indk}{\psk}\wps} \right.,\ z=1,0,
\end{equation*}
where treatment indicator $\Indk$ equals 1 if $\Zps=z$ and 0 otherwise. Seen from another perspective, $(\hat\rho_1,\hat\rho_0)$ is the minimizer of L2 loss as inflated by reciprocal probabilities of assignment, $L_{2(n)}(r_{1}, r_{2}) = \ssn \wps [ {\Indt} (\ypst-r_1)^2 /\pst + \Indc (\ypsc - r_0)^2 / \psc].$
We may obtain $\hat\tau$ from weighted least squares (WLS) regression of $\yps$ on treatment indicator $\zps$, including an intercept term and weights $\wps/\pi_{\s\zps}$.  

\label{sec:m-est}

We cast $(\hat{\rho}_{1}, \hat{\rho}_{0})$ as an M-estimator \citep{huber67,stefanski2002calculus} to obtain an asymptotic linearization that accommodates arbitrary stratum sizes and forms the basis for variance estimation in Section~\ref{sec3-var}. 
Fix $\s=1,\ldots,n$ and let $\p$ be the stratum containing unit $\s$.
For $z=1,0$, define a unit-level weighted and centered outcome $\gamsz (r_z) = \wps(\ypsk - r_z)$, functions of $r_z\in\mathbb{R}$, and
\begin{equation}
  \label{eq:1}
 \psi_i(r_1, r_0) = \left( \begin{array}{c}
     \Indt \gamst (r_1) /\ppt   \\
     \Indc \gamsc (r_0)/ \ppc  
\end{array} \right) =
    \frac{\gamma_{iZ_{i}}(r_{Z_{i}})}{\pi_{bZ_{i}}}
\left(
\begin{array}{c}
    \Indt\\ \Indc
\end{array}
\right),
\end{equation}
a random function of parameters $\theta = (r_1, r_0)^T$. Unlike the usual M-estimation framework, outcomes $\ypsk$ are fixed while assignments $Z_{i}$ are random.  For a finite population of size $n$, the parameter $\theta_n = (\rho_1,  \rho_0)^T$ defined in \eqref{eq:rho_k-def} solves estimating equation $\E [ \Phi_n(\rho_1,  \rho_0)] = 0$, where $\Phi_n(r_1,r_0)  = (\ssn\wps)^{-1}\ssn \psi_i(r_1,r_0)$; $\theta_{n}$ also minimizes L2 risk, $R_{2(n)}(r_{1}, r_{0})=\E[L_{2(n)}(r_{1}, r_{0})]=\ssn \wps (\ypst-r_1)^2 + \ssn\wps (\ypsc - r_0)^2$.
As an M-estimator, $(\hat\rho_1, \hat\rho_0)^T$ solves $\Phi_n(r_1,r_0) =  0$. %, where $W = \ssn \wps$ is the sum of unit weights.
Although it satisfies $\hat{\theta}_n = \theta_n + (-\nabla \Phi_n)^{-1} \Phi_n(\theta)$ exactly, $\Phi_n(r_1,r_0)$ being linear in $(r_1,r_0)$, its linearization \(\tilde\theta_n := \theta_n + ( -\E[\nabla \Phi_n] )^{-1} \Phi_n(\theta) \approx \hat\theta_n\) is an approximation only.

Consider RCT designs such as fine stratification, pair-matching, and others where many small strata are present. In the remainder of Section \ref{sec3-1}, assume the number of strata $\B$ grows to infinity. Now we derive the asymptotic covariance of $(\hat\rho_1, \hat\rho_0)$.

\begin{condition}
\label{cond:wlln-pi}
There exists a constant $c>0$ that is independent of the sample size such that for all $\p=1,\ldots,\B$ and $ z=1,0$, $\ppk\in (c,1-c)$ as $\B\to\infty$.
\end{condition}
\begin{condition}
\label{cond:wlln-w}
Kish's effective sample size $(\ssn \wps)^2 / (\ssn \wps^2) \to \infty$ as $\B\to\infty$. % and (b) for $ k=0,\ldots,K$, $\ssn \wps / \facB =\mu > 0$ and $ \mu \to \mu_{\infty}>0$ as $\B\to\infty$. Here, $\mu$ should also have a subscript $n$ but we omit it for simplicity.
\end{condition}
%\sum_{\substack{b=1 \\ \ppk>0}}^B

Conditions \ref{cond:wlln-pi} and \ref{cond:wlln-w} guarantee that the gradient of $\Phi_n$ with respect to $r_z$ has a vanishing variance as $\B$ increases. Condition \ref{cond:wlln-pi} is also known as the strong overlap condition. Condition \ref{cond:wlln-w} rules out dominance by a few large clusters in the RCT. 
A sufficient condition of Condition \ref{cond:wlln-w} is: there exist constants $M\ge m> 0$ such that $m\le \wps \le M$, $\forall \s$. We can also permit that $\wps=0$ for some $\s$, which may arise in estimation of subgroup effects. This requires $\wps \le M$ for all $\s$ and a positive proportion of the weights exceeding $m>0$ as the sample size grows.

To characterize the asymptotic distribution of $\Phi_n(\rho_1,\rho_0)$, write $\gamsz= \gamsz(\rho_z)$ and $\bgampz=\bgampz(\rho_z)$ and define stratum-wise finite population variances and covariances of $\gamsz$'s as
\begin{equation}
    S_{\p,z}^2 = \frac{1}{\np-1} \sump \left( \gamsz - \bgampz \right)^2,\ z=1,0; \ \  S_{\p,01} = \frac{1}{\np-1} \sump  \left( \gamst - \bgampt \right) \left( \gamsc - \bgampc \right).
    \label{eq:var-gamma}
  \end{equation}
Analogous definitions give $S_{\p,1}^2(r_1)$, $S_{\p,0}^2(r_0)$ and $S_{\p,01}(r_1,r_0)$, with $\gamsz(r_z)$ and $\bgampz(r_z)$ in place of $\gamsz$ and $\bgampz$. 
Let $W=\ssn\wps$ be the total unit weight. Define 
\begin{equation}
\tsigma_z^2 = \frac{1}{W} \spb \np \left( \frac{\np}{\npk} - 1 \right) S_{\p,z}^2,\ z=1,0\ \ \text{and} \ \ \tsigma_{01} = -\frac{1}{W} \spb {\np} S_{\p,01}.
    \label{eq:sigmas}
\end{equation}
Then the covariance matrix of $\Phi_n(\rho_1,\rho_0)$ is
\begin{equation}
\cov\left(  \Phi_n(\rho_1,\rho_0) \right) = \frac{1}{W} \Sigma_{n},\  \text{ where }\  \Sigma_n= \left( \begin{array}{cc}
        \sigma_1^2 & \sigma_{01}  \\
        \sigma_{01} & \sigma_0^2 
        \end{array} \right).
    \label{eq:covXb}
\end{equation}
%as shown in the supplementary materials.
\iffalse
\begin{align*}
    \sigma_z^2 = \frac{1}{\facB} \spb & \Bigg[  \left(\frac{1}{\ppk} - 1\right)\sump \wps^2(\ypsk - \rho_z)^2 \\
    & +  \left(\frac{\delbkk}{\ppk^2}-1\right) \sump \sumssp \wps \wpst(\ypsk - \rho_z)(\ypstk - \rho_z) \Bigg]   ,\ \text{ } z=1,0
\end{align*}
and 
\begin{align*}
 \sigma_{01} = \frac{1}{\facB} \spb &\Bigg[ (-1) \sump \wps^2(\ypsc - \rho_0)(\ypstt - \rho_1)\\
&+  \left(\frac{\delbct}{\ppc\ppt}-1\right) \sump \sumssp \wps \wpst(\ypsc - \rho_0)(\ypsstt - \rho_1) \Bigg]  .
\end{align*}
\fi
We assume standard regularity conditions ensuring that $\Sigma_n$ is finite and positive definite, and that a central limit theorem under Conditions \ref{cond:lf1}--\ref{cond:lf2} specified in the Supplementary Material (Section \ref{sec:sm-thm1-conditions}) applies to $\Phi_n(\rho_1,\rho_0)$. 
\begin{theorem}
\label{thm1}
Under Conditions \ref{cond:wlln-pi}--\ref{cond:wlln-w} and \ref{cond:lf1}--\ref{cond:lf2}, 
$ (\ssn \wps)^{1/2} (\hat\theta_n - \theta_n)  \simd \mathcal{N} (0, \Sigma_n)$ and $(\ssn \wps)^{1/2}(\hat\tau - \tau)\simd \mathcal{N}(0, \sigma_1^2-2\sigma_{01}+\sigma_0^2) $ as $\B\to\infty$. 
\end{theorem}

The notation $X_n\simd \mathcal{N}(0,V_n)$ as $n\to\infty$, where $X_n\in\mathbb{R}^d$ is a sequence of random vectors and $V_n\in\mathbb{R}^{d\times d}$ is a sequence of positive definite matrices, indicates that $V_n^{-1/2}X_n \xrightarrow{d}\mathcal{N}(0,I_d)$ where $I_d$ is a $d$ by $d$ identity matrix. %This suggests that $\mathcal{N}(0,V_n) $ serves as an approximation of the asymptotic distribution of $X_n$.
Theorem \ref{thm1} establishes the consistency, asymptotic normality and large-sample variance of $\hat\tau$ as the number of strata grows.

\subsection{Variance estimation for arbitrary stratum sizes}
\label{sec3-var}

Consider $ \gamsz(r) = \wps(\ypsk - r) $ as an adjusted potential outcome function. The mean of $\gamsz(r)$'s for units \textit{assigned to} treatment $z$ in stratum $\p$ is
$\bgamobsz(r) = \npk^{-1}\sum\{ \gamsz(r): i \in \p, \Zps=z \} = \sump {\Indk}/{\psk} \cdot \wps (\ypsk - r) $.
By Theorem \ref{thm1}, a large-sample approximation of $ \var(\hat\tau) $
\begin{align}
    \frac{1}{W}(\sigma_1^2 - 2\sigma_{01}+\sigma_0^2) = \frac{1}{W^2}\spb \np^2 \var \left(\bgamobst(\rho_1) - \bgamobsc(\rho_0) \right) .
    \label{eq:varestsec3}
\end{align}
Thus, we can estimate $\var( \hat\tau)$ by aggregating estimators of $\var \left(\bgamobst(\rho_1) - \bgamobsc(\rho_0) \right)$ by stratum. We present two such estimators, applicable for strata of varying sizes. 

The first estimator, a well-established result dating back to \citet{splawa1990application}, requires at least two units per treatment arm:
\begin{equation}
    \nu_{\p}^{(l)} (r_1,r_0)  =  \frac{s_{\p 1}^2(r_1)}{\npt} + \frac{s_{\p 0}^2(r_0)}{\npc},
    \label{eq:est1}
\end{equation}
where $ s_{\p z}^2 (r_z)  = \sum_{i\in\p,\Zps=z} \{\gamsz(r_{z}) - \bgamobsz(r_{z})\}^2 /(\npk-1)$ is the sample variance of $\gamsz$'s for units assigned to treatment $z$ in stratum $\p$: $\{\gamsz: \s\in\p, \Zps=z \}$, and $(l)$ stands for \textit{large}-stratum variance estimator.

The second estimator is our novel estimator that applies to strata of any size. When strata are small, as in paired or finely stratified designs where $\npt=1$ or $\npc=1$, the stratum-wise within-arm sample variances required by \eqref{eq:est1} are not defined. We propose an estimator that remains valid for strata of any size, including singletons.
\begin{proposition}
\label{lem:ab2}
Consider a completely randomized trial with $n$ units, where $n_1$ units are randomly assigned to treatment and $n_0$ to control. For each subject $\s$, let $\ai$ and $\bi$ be the potential outcomes under treatment and control. Denote $\baraobs$ and $\barbobs$ as the observed mean outcomes for treated and control units, respectively. Let $\bara = \sum_\s \ai / n$, $\barb = \sum_\s \bi /n$, $\sigma_1^2 = \sum_\s (\ai -\bara)^2 / n$, $\sigma_0^2 = \sum_\s (\bi -\barb)^2 / n$, and $\sigma_{10} = \sum_\s (\ai-\bara)(\bi-\barb) / n$. Then the following is an upward-biased (conservative) estimator of $\var(\baraobs - \barbobs)$:
\begin{equation*}\label{eq:novel-var-estimator}
    \frac{1}{n_1 n_0} \sum_{i=1}^{n_1} \sum_{j=1}^{n_0} (y_{i|1}-   y_{i|0})^2 - \frac{1}{n_1} \sum_{i=1}^{n_1} (y_{i|1} - \baraobs)^2 - \frac{1}{n_0} \sum_{j=1}^{n_0} (y_{i|0} - \barbobs)^2 ,
\end{equation*}
which is nonnegative when $n_1=1$ or $n_0=1$, having a nonnegative bias of $(\bara - \barb)^2$.
\end{proposition}
The intuition behind this proposed estimator is straightforward. When within-arm sample variances cannot be computed, we instead use the average squared difference between all treatment-control outcome pairs, and correct by subtracting within-arm variations. The resulting estimator overestimates the variance by the squared treatment effect but remains well-defined even for singleton strata.
By Proposition \ref{lem:ab2}, a conservative and nonnegative estimator of $\var \left(\bgamobst(r_{1}) - \bgamobsc(r_{0}) \right)$ for small stratum with $\npt=1$ or $\npc=1$ is
\begin{equation}
    \nu_{\p}^{(s)} (r_1,r_0) = (\npt \npc)^{-1} \sum_{\substack{\s\in\p \\ \Zps=1}}\sum_{\substack{\s'\in\p \\ \Zpst=0}} \{\gamst(r_{1})-\gamssc(r_{0})\}^2 - \sum_{z=1,0}\npk^{-1} \sum_{\substack{\s\in\p \\\Zps=z}} \{\gamsz(r_{\zps}) - \bgamobsz(r_{\zps})\}^2 .
    \label{eq:est3}
\end{equation}

\label{sec:varlarge}
%\subsubsection{An estimator for small strata}
\label{sec:varsmall}

\begin{remark}\label{rmk:varbias}
Under constant treatment effects, $\nu_{\p}^{(e_\p)}(\rho_1,\rho_0)$, $e_\p\in\{l,s\}$, 
are unbiased for $\var\left(\bgamobst(\rho_1) - \bgamobsc(\rho_0)\right)$. In general, 
the bias of $\nu^{(l)}$ increases with within-stratum treatment effect heterogeneity, 
while the bias of $\nu^{(s)}$ increases with heterogeneity across strata. Consequently, 
$\nu^{(l)}$ is preferred when heterogeneity is mainly across strata, and $\nu^{(s)}$ 
when strata are small and within-stratum variances are not defined. Formal expressions 
for both biases are given in Section \ref{sec:sm-remark-varbias} of the Supplementary Material.
\end{remark}

Based on the decomposition of $\var(\hat\tau)$ in (\ref{eq:varestsec3}), the variance of $\hat{\tau}$ is estimated by
\begin{equation}
    \hat{v}(\hat\tau) = \frac{1}{W^2} \spb \np^2 \nu^{(e_{\p})}_{\p}(\hat\rho_1, \hat\rho_0),
    \label{eq:case1varest}
\end{equation}
where $e_{\p} \in \{l, s\}$ indicates the stratum-specific estimator, allowing $\nu^{(e_{\p})}_{\p}(\cdot, \cdot)$ to be either \eqref{eq:est1} or \eqref{eq:est3}, adapting to strata of varying sizes. H\'ajek estimators $\hat\rho_1$ and $ \hat\rho_0$ are substituted for unknown parameters $\rho_1$ and $\rho_0$. For a pair-matched, non-clustered RCT, this estimator simplifies to $\B^{-2}\spb (\hat\tau_\p - \hat\tau)^2 , $ where $\hat\tau_\p$ is the difference between treatment and control outcomes of pair $\p$. When $\B$ is large, this estimator approximates the classical variance estimate for \eqref{eq:sadm} in paired designs \citep{imbens2011experimental, fogarty2018finely}: $\{\B (\B-1)\}^{-1} \spb (\hat\tau_\p - \hat\tau)^2 $.

The following Theorem indicates that estimator \eqref{eq:case1varest} has a consistent relationship with the variance of the asymptotic linearization of $\hat\tau$, denoted $\tilde\tau=\tilde\rho_1-\tilde\rho_0$, where the linearization $\tilde\theta_n=(\tilde\rho_1,\tilde\rho_0)^\top$ is defined in Section \ref{sec:m-est}. It assumes mild regularity conditions ensuring that no single stratum dominates the variance sum, and that the small-stratum estimator $\nu^{(s)}$ is only applied for strata with bounded average outcomes. Without loss of generality, assume that $\rho_1$ and $\rho_0$ converge to finite values as $\B \to \infty$, which can be achieved by appropriately scaling the potential outcomes.

\begin{theorem}
\label{prop:cons1}
Suppose that $\nu^{(s)}_\p$ is only applied for small strata with bounded $\np$ and weighted outcome means $\sump \wps \ypsk / (\sump \wps),$ $ z=1,0$. Under regularity conditions specified in the Supplementary Materials (Section \ref{sec:sm-cons1-conditions}), the estimator is asymptotically conservative in the sense that there exists a sequence $\{\xi_\B, \ \B\ge 1\}\subseteq \mathbb{R}$ such that $(\ssn\wps) \hat{v}(\hat\tau) / \xi_\B \xrightarrow{p} 1$ as $\B\to\infty$ and $\liminf_{\B\to\infty}\big\{\xi_\B/\var\big( (\ssn\wps)^{-1/2}\tilde\tau \big) \big\} \ge 1,$ where equality holds when the unit treatment effect, $\ypst-\ypsc$, is constant for all units. 
\end{theorem}

\subsection{Score-based testing for improved coverage}
\label{sec:hypothesis}
Given a variance estimator, a natural approach is a Wald-type confidence interval, $\hat\tau \pm z_{1-\alpha/2} \sqrt{\hat\nu(\hat\tau)}$. However, our simulations show that this interval undercovers when the number of clusters is moderate. The issue arises because the variance estimator $\hat\nu(\hat\tau)$ is evaluated at the estimated values $(\hat\rho_1,\hat\rho_0)$, which may differ substantially from the values implied by the null. Even with an $n-2$ degrees-of-freedom adjustment---reflecting that the linear model producing $\hat\tau$ includes an intercept and a binary treatment indicator---the Wald interval continues to undercover in small samples.
 
To mitigate this problem, we propose a score-type test that evaluates the variance estimator under null hypotheses of the form $H$: $\big( \sum_{i:\zps=z} \wps \big)^{-1} \sum_{i:\zps=z} \wps \tau_i = \tau_0$ for $z = 1,0$.
This null fixes the SATE at $\tau_0$ while also restricting the weighted average of cluster-level effects within the realized treatment and control groups to $\tau_0$. In the spirit of Rosenbaum's \citeyearpar{rosenbaum2001effects} hypotheses on attributable effects, this \emph{equal attribution hypothesis} imposes no restrictions on means over groups that could have been selected for treatment but were not; it is strictly weaker than the constant-effect hypothesis $H:\tau_i=\tau_0$ for all $i$. It is the weakest hypothesis under which %$(\rho_{1(H)}, \rho_{0(H)})$
$\theta_{H}$ can be identified from observed data alone, allowing us to employ the variant of \eqref{eq:case1varest} replacing estimators $\hat\rho_1$ and $\hat\rho_0$ with fixed constants $\rho_{1(H)}$ and $\rho_{0(H)}$.
By Theorem~\ref{prop:cons1}'s proof,
this modified \eqref{eq:case1varest} is nonnegatively biased for $\var_{H}(\tilde\tau)$, at samples of any size.  

Define the test statistic: $t=(\ssn\wps)^{-1}\big| \spb [\bgamobst (\rho_{1(H)}) -  \bgamobsc (\rho_{0(H)})] \big| / \big\{ \hat{v}_{(H)}(\hat\tau) \big\}^{1/2}$, where $\rho_{1(H)}$ and $\rho_{0(H)}$ are null imputations of $\rho_1$ and $\rho_0$, obtained by shifting observed treatment condition means and taking weighted averages:
\begin{equation}
    \rho_{1(H)} = \frac{\sum_{i:\zps=0}\wps}{\ssn\wps} \cdot \left( \frac{\sum_{i:\zps=0}\wps\ypsc}{\sum_{i:\zps=0}\wps} + \tau_0 \right) + \frac{\sum_{i:\zps=1}\wps }{\ssn\wps}\cdot \frac{\sum_{i:\zps=1}\wps\ypst}{\sum_{i:\zps=1}\wps},
    \label{eq:rho1H}
\end{equation}
\begin{equation}
    \rho_{0(H)} = \frac{\sum_{i:\zps=0}\wps}{\ssn\wps} \cdot \frac{\sum_{i:\zps=0}\wps\ypsc}{\sum_{i:\zps=0}\wps} + \frac{\sum_{i:\zps=1}\wps }{\ssn\wps}\cdot \left(\frac{\sum_{i:\zps=1}\wps\ypst}{\sum_{i:\zps=1}\wps} - \tau_0 \right),
    \label{eq:rho0H}
\end{equation}
and $\hat{v}_{(H)}(\hat\tau)= (\ssn\wps)^{-2}\spb \nu_\p^{(e_\p)}(\rho_{1(H)},\rho_{0(H)})$ is the variance estimation of $t$'s numerator. 
Under $H$ and the conditions of Theorems \ref{thm1} and \ref{prop:cons1}, the statistic $t$ is approximately Normal with variance at most 1. The level-$\alpha$ test rejects $H$ when $\lvert t \rvert > z_{1-\alpha/2}$, where $z_{1-\alpha/2}$ is the upper $\alpha/2$ quantile of the standard Normal. A $(1-\alpha)$ confidence interval (CI) for the SATE includes all $\tau_0$ not rejected by this test. This procedure resembles a generalized score test \citep[e.g.,][]{boos2013essential}, differing only by using the ratios ${\sum_{i:\zps=z}\wps\ypsk}/{\sum_{i:\zps=z}\wps} , z=1,0$ instead of the corresponding H\'ajek ($\pi^{-1}$-weighted) estimators in \eqref{eq:rho1H} and \eqref{eq:rho0H}.

\subsection{Complementary regime: a few large strata}
\label{sec:inference-large-strata}
Sections \ref{sec3-1}--\ref{sec:hypothesis} consider cRCTs with many small strata. The Supplementary Material (Section \ref{sec4}) establishes analogous results for the complementary regime of a fixed number of strata with growing sizes: $(\ssn\wps)^{-1/2}(\hat\tau-\tau)$ converges to the same limiting normal distribution as in Theorem~\ref{thm1}, and the variance estimator formed by aggregating \eqref{eq:est1} remains asymptotically conservative, as formalized in a theorem parallel to Theorem~\ref{prop:cons1}.  

\subsection{Covariate adjustment}
\label{sec:covadj}
\label{sec:cov-adj-theory}
Covariate adjustment can improve efficiency without sacrificing the design-based properties established in Sections~\ref{sec3-1}--\ref{sec:inference-large-strata}. We show that, with centered covariates, the adjusted estimator continues to target the SATE, and the variance estimators of Section~\ref{sec3-var} apply with only one modification: replacing each $\gamsz$ with a covariate-residualized counterpart $\tgamsz$.

Let $\xsj\in\mathbb{R}^p$ be individual- or cluster-level covariates for individual $j$ in cluster $i$, centered so that $\ssn\sumij \wsj\xsj = 0$. Each individual carries weight $\wsj\ge0$ with $\sumij\wsj = \wps$, where $\mi$ is the number of individuals in cluster $i$, so that individual weights aggregated to cluster weights. Let individual-level outcomes be $\ysjz$ with $\sumij \wsj \ysjz / \wps = \ypsk$. 
The covariate-adjusted H\'ajek effect estimator $\hat\tau_{\mathrm{adj}} = \hat{\rho}_{1,\mathrm{adj}} - \hat{\rho}_{0,\mathrm{adj}}$ is obtained from WLS regression of $y_{\s j\zps}$ on a treatment indicator and $\xsj$, with weights $\wsj/\pi_{b\zps}$:
\begin{equation}
    y_{\s j\zps} \approx {\rho}_{\zps,\mathrm{adj}}  + \xsj^T {\beta}.
    \label{eq:covadj-wls}
\end{equation}

The analysis of $\hat\tau_{\mathrm{adj}}$ mirrors Section~\ref{sec:m-est}. Define residualized weighted outcomes $\tgamsz (\beta,r_z) = \wsj(\ysjz - \xsj^T \beta - r_z) $ in parallel with $\gamsz (r_z)=\wps(\ypsk-r_z)$, and define estimating functions
\begin{equation*}
  \tilde{\psi}_{\s j}(\beta,r_1,r_0)    =
   \frac{\tilde{\gamma}_{ijZ_{i}}(\beta, r_{Z_{i}})}{\pi_{bZ_{i}}}
 \left(
   \begin{array}{c}
     \xsj \\ \Indt\\ \Indc
   \end{array}
\right) ,\quad
   \tilde \Phi_n(\beta, r_1,r_0) =\ssn\sumij \tilde{\psi}_{\s j},
\end{equation*}
so that $(\hat\beta,\hrhotadj, \hrhocadj)$ solves $\tilde{\Phi}_n(\beta, r_1,r_0) =0$.
Barring collinearities, there exist unique parameters $(\tilde\beta, \rho_{1,\mathrm{adj}}, \rho_{0,\mathrm{adj}})$ satisfying $\E\tilde{\Phi}_n(\tilde\beta, \rho_{1,\mathrm{adj}}, \rho_{0,\mathrm{adj}}) =0$. We first show that $\hat\tau_{\mathrm{adj}}$ estimates the same target as $\hat\tau$. Since $\sum_{j=1}^{m_{i}}\tilde{\gamma}_{ij z}(\beta, r_{z}) -  \gamsz(r_{z})=\sum_{j=1}^{m_{i}} w_{ij}x_{ij}^{T}\beta$ for each $i$ and $z=0,1$, covariate centering entails
$\ssn\sumij \E[\pi_{bz}^{-1}\{\tilde{\gamma}_{ijz}(\beta, r_{z}) -
 \gamsz(r_{z})\} \Indk] =0$ for any $\beta$ and $r_{z}$. Therefore, $\E\{\tilde \Phi_n(\beta, r_1,r_0)\}$ and $\E\{\Phi_n(r_1,r_0)\}$ agree in their first $r_1$ and $r_0$ components, which gives
$(\rho_{1,\mathrm{adj}}, \rho_{0,\mathrm{adj}}) = (\rho_{1}, \rho_{0})$: the estimator $\hat\tau_{\mathrm{adj}}$ targets the SATE.

Covariate centering also makes the ``bread'' matrix $\E\{\nabla_{\beta,r_{1}, r_{0}} \tilde \Phi_n\}$ block diagonal between $\beta$ and $(r_1,r_0)$, since the off-diagonal block equals $\ssn \sumij w_{ij}x_{ij}^\top =0$ by centering. Consequently, the asymptotic behavior of $(\hat\rho_{1,\mathrm{adj}}, \hat\rho_{0,\mathrm{adj}})$ is not influenced by the estimation of $\beta$: in the linearization of $\hat\theta_n$, the $(r_1,r_0)$ components depend only on the $(r_1,r_0)$ components of $\tilde\Phi_n$, not on its $\beta$ component.
Define $\tilde{S}_{\p, z}^{2}= \tilde S_{\p,z}^2(\tilde\beta,\rho_z) $ and $\tilde{S}_{\p,01} = \tilde S_{\p,01}(\tilde\beta,\rho_1,\rho_0)$, in parallel with \eqref{eq:var-gamma}, by substituting $\tgamsz=\tgamsz(\tilde\beta,\rho_z)$ for $\gamsz$.
Define $\tilde{\sigma}_{1}^{2}$, $\tilde{\sigma}_{0}^{2}$ and $\tilde{\sigma}_{01}$ from $\tilde{S}_{\p,1}^{2}$, $\tilde{S}_{\p,0}^{2}$ and $\tilde{S}_{\p,01}$ as \eqref{eq:sigmas} defined $\tsigma_1^2$, $\tsigma_0^2$ and $\tsigma_{01}$ in terms of ${S}_{\p,1}^{2}$, ${S}_{\p,0}^{2}$ and ${S}_{\p,01}$. 
We obtain under conditions similar to those required for Theorems~\ref{thm1} and \ref{thm2}, $(\ssn\wps)^{-1/2}(\hat\tau_{\mathrm{adj}} - \tau) \simd \mathcal{N}(0,\tilde \sigma_1^2 - 2\tilde \sigma_{01} + \tilde\sigma_0^2)$.
The variance estimation and hypothesis testing methods of Section~\ref{sec:hypothesis} apply directly to $\hat\tau_{\mathrm{adj}}$, replacing $\gamsz(\hat\rho_z)$ with $\sumij \tilde{\gamma} (\hat\beta,\hat\rho_{z,\mathrm{adj}})$ in \eqref{eq:est1} and \eqref{eq:est3}. This substitution is the only modification required.
%Two further remarks on this extension. First, it does not accomodate the FE estiamtor. In fine stratification or matched-pair designs where $B \propto n$, including block fixed effects would require $p \propto n$ parameters, violating the fixed-$p$ assumption.
%Our extension of H\'ajek estimation to adjust for covariates does not generally accommodate the FE estimator, unfortunately. (Our $p$ is fixed, but $B \propto n$ in small-block designs such as matched pairs, so block fixed effects would require $p \propto n$.)

The centering of covariates is optional if interest is only in $\hat{\tau}_{\mathrm{adj}}$ and its variance. By the Frisch-Waugh-Lovell theorem, $\hat\tau_{\mathrm{adj}}$, $\hat\beta$, and the weighted residuals are invariant to centering. However, centering is recommended when combining covariate adjustment with the method of Section~\ref{sec:hypothesis}, which fixes $\rho_{1(H)}$ and $\rho_{0(H)}$ rather than estimating them jointly with $\beta$.
With centered covariates, the $\mathrm{ATE}_{\mathrm{interact}}$ estimator of \citet{lin2013} adapts to the cRCT setting by adding $\Indt \cdot \xsj^T {\gamma}$ or $\Indc \cdot \xsj^T {\gamma}$ to the right-hand side of \eqref{eq:covadj-wls}, giving $\hat{\tau}_{\mathrm{interact}}=\hat\rho_{1, \mathrm{interact}}-\hat\rho_{0, \mathrm{interact}}$.  Again \eqref{eq:case1varest} provides variance estimates,
this time with substitution of
$\sumij\tgamsz\{(\hat\beta, \hat\gamma, \hat\rho_{1}, \hat\rho_{0})_{\mathrm{interact}}\}$  for $\gamsz(\hrhozadj)$ in \eqref{eq:est1} and \eqref{eq:est3}. 

\section{Simulation Studies} 
\label{secsimu}

\subsection{Performance of SATE estimators under size-effect correlation}
\label{sec:ATEsimu}

%The H\'ajek estimator has been evaluated for SATE estimation in some settings \citep[e.g.,][]{liu2016inverse, KhanUgander2023}.
Now we compare the H\'ajek effect estimator (HA) with three alternatives: the IKN effect estimator (IKN), the fixed effects estimator (FE), and the Horvitz-Thompson estimator (HT). The HT is defined as 
$ (\ssn \wps)^{-1}\spb\left\{  \ppt^{-1} \sump \Indt \wps\ypst -  \ppc^{-1} \sump \Indc \wps\ypsc \right\}. $
We vary the sample size and generate cluster sizes from an i.i.d.\ Gamma distribution (following \cite{Tong2022accounting}). Consider two cases of stratification: clusters are matched or not matched on size within strata. Potential outcomes are generated from a multivariate normal distribution, referencing \cite{pashley2021insights}:
\begin{equation}
    \label{eq:potential_gen}
    \left(\begin{array}{c}
        \ypsc  \\
        y_{\s1}\\
    \end{array}\right) \sim MV \mathcal{N} \left( \left( \begin{array}{c}
        \Phi^{-1}(1-\frac{i}{1+n})  \\
        \Phi^{-1}(1-\frac{i}{1+n}) + \ti
    \end{array} \right),
    \left(\begin{array}{ccc}
        1 & 0  \\
        0 & 1  
    \end{array}\right) \right),\ i=1,\ldots,n,
\end{equation}
where $\Phi^{-1}$ is the quantile function of standard normal distribution. We set $\ti$ to either a constant or correlated with cluster size. The simulation employs a $2\times2$ factorial design, varying (1) cluster size variation within strata and (2) size-effect correlation. After generating cluster sizes and outcomes, we fix them across 10,000 treatment allocations, each randomizing half of each stratum's clusters to treatment. All four estimators are computed in each iteration. Figure~\ref{fig:hajek-ab} shows selected empirical distributions; see the Supplementary Materials (SM) for full details and results. %Table \ref{tab:hajek} in the supplementary material presents the biases, standard deviations (SD), and root mean squared errors (rMSEs) of these estimators in different cases.
\begin{figure}[!ht]
    \centering
    \includegraphics[width=\textwidth]{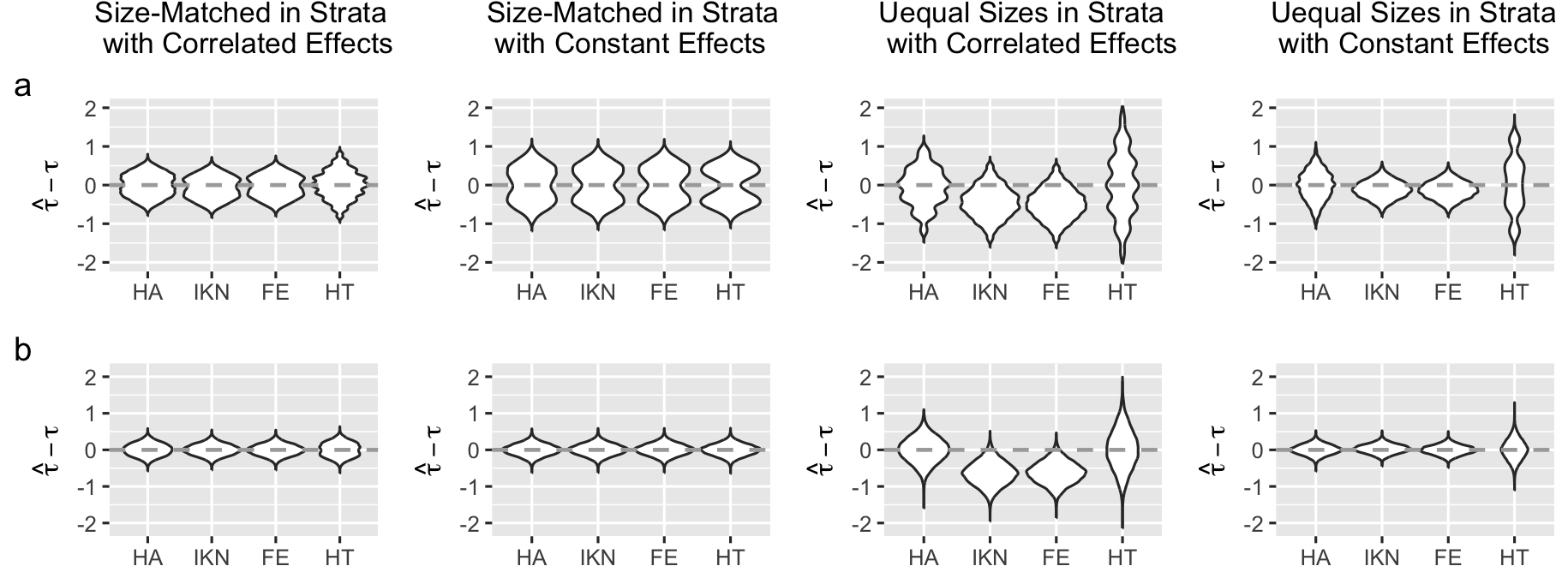}
    \caption{Empirical distributions of the H\'ajek effect estimator (HA), the IKN effect estimator (IKN), the fixed effects estimator (FE), and the Horvitz-Thompson estimator (HT) centered at the SATE. Each column represents a design, described at the top. Each row represents a sample size: a, 10 strata of 2 clusters; b, 50 strata of 2 clusters.}
    \label{fig:hajek-ab}
\end{figure}

No estimator dominates uniformly in MSE. HA has the clear advantage in paired designs with cluster size variation within pairs and correlation between treatment effects and cluster size. It performs well in other cases with a small stratum size. HT is unbiased but carries the largest MSE in most cases  \citep[as noted in][]{middleton2015unbiased, mann2024general}, improving only when strata are large and clusters are size-matched. IKN and FE behave similarly to each other and to HA when cluster sizes are matched or stratum size exceeds 10, but incur larger biases when sizes are unmatched, particularly with heterogeneous effects and small strata. The one scenario where alternatives to HA have an edge, constant effects with within-pair size variation, is also where that edge is smallest.

\subsection{Variance estimation and confidence interval coverage}
\label{secsimu-small}

We consider two settings without cluster size matching: many small strata ($\B = 10, 50$; $\np = 2, 4$) and two larger strata (size 10 or 50). Potential outcomes follow model \eqref{eq:potential_gen} with $\ti = 5 + f_i(\alpha,\beta)$, where $\alpha$ and $\beta$ control treatment effect heterogeneity across and within strata; when one varies, the other is fixed at 0. The equal attribution hypothesis fails whenever $\alpha \ne 0$ or $\beta \ne 0$.
We study both balanced and unbalanced treatment assignment. In balanced designs, half of the clusters in each stratum are assigned to treatment, and we evaluate $\hat{v}_{\text{large}}(\hat\tau) = \spb \np^2 \nu^{(l)}_{\p}(\hat\rho_1, \hat\rho_0)/\facB^2$ (when $\npt, \npc \geq 2$) and $\hat{v}_{\text{small}}(\hat\tau) = \spb \np^2 \nu^{(s)}_{\p}(\hat\rho_1, \hat\rho_0)/\facB^2$. In unbalanced designs with $\np = 4$, we assign one treated cluster in strata $\p = 1, \ldots, \B/2$ and two treated clusters in strata $\p = \B/2 + 1, \ldots, \B$, allowing evaluation of the ``mixed" estimator
\(
\hat{v}_{\text{mixed}}(\hat\tau) =
[
\sum_{\p\le \B/2} \np^2 \nu^{(s)}_{\p}(\hat\rho_1, \hat\rho_0)
+
\sum_{\p>\B/2} \np^2 \nu^{(l)}_{\p}(\hat\rho_1, \hat\rho_0)
]/\facB^2.
\)
When $\np \ge 10$, one-fifth of each stratum is treated and only $\hat{v}_{\text{large}}(\hat\tau)$ is used.
After randomization, compute the H\'ajek effect estimator, its variance estimators, and 95\% confidence intervals (CIs) by inverting the score-based and Wald-type tests of Section \ref{sec:hypothesis}. As a model-based benchmark, we also compute the HC2 heteroskedasticity-robust variance estimator and the corresponding Wald CI with $n-2$ degrees of freedom, exploiting the equivalence of $\hat\tau$ to the coefficient of $z_i$ in a weighted least squares regression of $y_{iz_i}$ on $z_i$ with weights $\wps/\pi_{[i]z_i}$. We repeat the simulations 10,000 times and evaluate relative bias ($ \{\mean(\hat{v}(\hat\tau)) - \var(\hat\tau)\} /  \var(\hat\tau)$, $\var(\hat\tau)$ being the empirical variance), SD, and empirical CI coverage; see the SM for details.

Figure \ref{fig:var_bias} shows the relative biases of variance estimators for small-strata designs. Overall, the proposed estimators track the variation of the H\'ajek effect estimator well. The estimator $\hat v_{\text{small}}(\hat\tau)$ is more sensitive to treatment effect heterogeneity across strata than $\hat v_{\text{large}}(\hat\tau)$, while both are largely insensitive to within-stratum heterogeneity. HC2 is more conservative, partly because it does not fully capture the precision gains from stratification \citep{bai2024primer}. Differences in the SDs of $\hat{v}_{\text{small}}(\hat\tau)$ and $\hat{v}_{\text{large}}(\hat\tau)$ are small (Figure \ref{fig:var_sd}, SM).
\begin{figure}[!htb]
    \centering
    \includegraphics[clip,width=\textwidth]{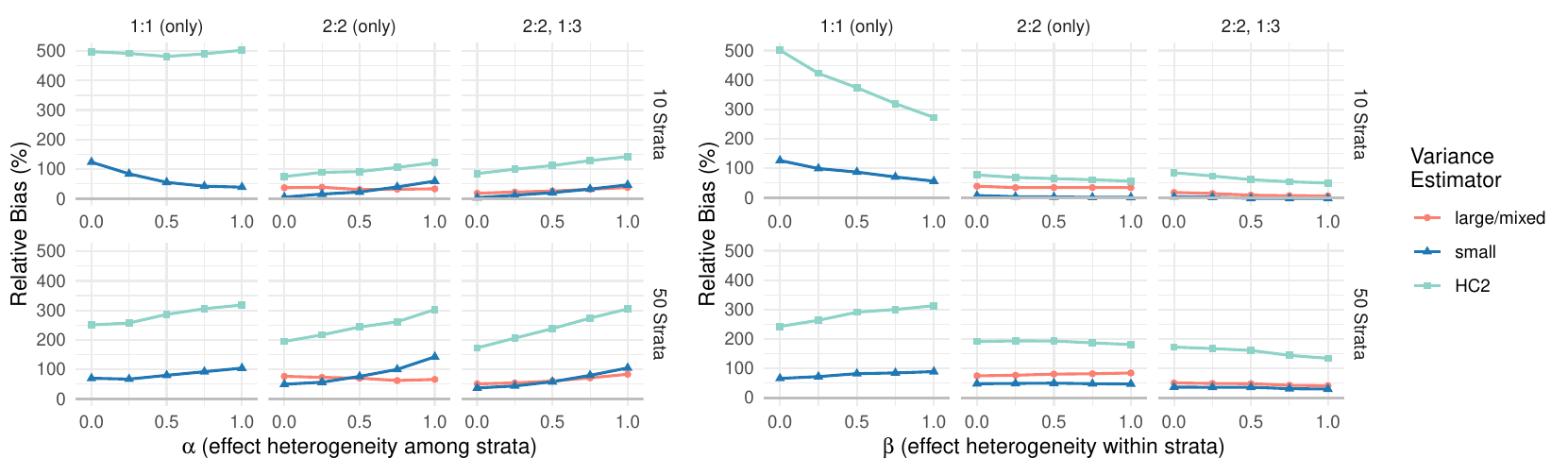}
    \caption{Relative biases of variance estimators plotted against $\alpha$ and $\beta$. Each column represents a stratum design and each row represents a number of strata. }
    \label{fig:var_bias}
\end{figure}

Figure \ref{fig:var_cvg} presents the empirical coverage probabilities of the CIs. For matched-pair designs, CIs based on $\hat v_{\text{small}}(\hat\tau)$ attain the nominal coverage. For balanced designs with $\np=4$, Wald CIs using $t_{n-2}$ quantiles and $\hat v_{\text{large}}(\hat\tau)$ are also close to or above 95\%. In unbalanced designs with a small sample, Wald-type CIs tend to under-cover, whereas score-based CIs maintain adequate coverage without excessive length (Figure \ref{fig:var_cil}, SM). With 50 strata, all methods perform well, consistent with Theorem \ref{prop:cons1}.
\begin{figure}[!htb]
    \centering
    \includegraphics[clip,width=\textwidth]{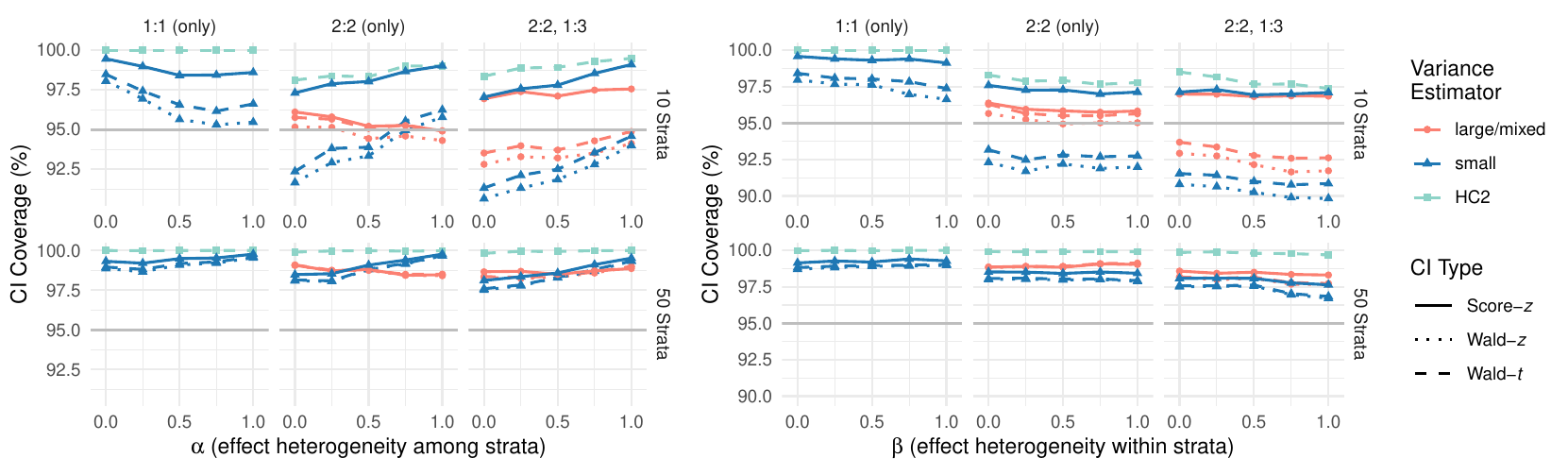}
    \caption{Empirical coverage probabilities of 95\% confidence intervals (CIs) plotted against $\alpha$ and $\beta$ under different designs and sample sizes. The solid lines represent score-based CIs, which consistently achieve empirical coverage near or above 95\%. }
    \label{fig:var_cvg}
\end{figure}

\label{sec:simu-large}

With moderate or large strata, the proposed variance estimator $\hat v_{\text{large}}(\hat\tau)$ remains valid and is less conservative than HC2, though neither uniformly dominates in SD (Figure \ref{fig:c2_sd}, SM). Figure \ref{fig:c2_cvg} shows that all CIs attain about 95\% coverage in balanced or large-stratum designs. In unbalanced, small-strata settings, coverage declines because both $\hat\tau$ and its variance estimates are more variable. Score-based CIs outperform Wald-type CIs based on $\hat v_{\text{large}}(\hat\tau)$: they attain 95\% coverage under homogeneous cluster treatment effects and degrade only mildly when the equal attribution hypothesis fails.
\begin{figure}[!ht]
    \centering
    \includegraphics[width=\textwidth]{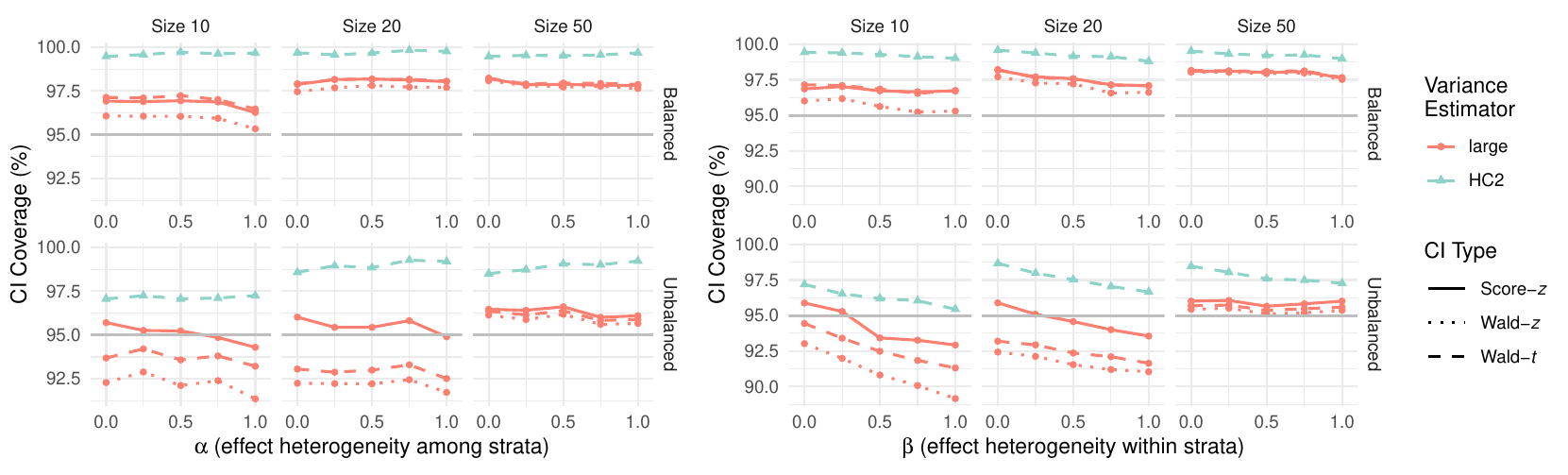}
    \caption{Empirical coverage probabilities of 95\% confidence intervals. Each column represents a stratum size and each row represents a treatment assignment design. The solid lines represent score-based CIs. }
    \label{fig:c2_cvg}
\end{figure}

Based on these results, we recommend using $\nu^{(l)}_\p$ when applicable and $\nu^{(s)}_\p$ otherwise to estimate the variance of $\hat\tau$, denoted $\hat v(\hat\tau)$ without the ``large," ``mixed," or ``small" distinctions (can be implemented using the \texttt{propertee} package \citep{propertee}). Varying the constant in $\ti$ of the outcome model \eqref{eq:potential_gen} does not significantly affect the results, so coverage above 95\% when the SATE is zero suggests approximate type I error control. The score-based CI has better coverage than Wald CIs based on $\hat v(\hat\tau)$, especially in small samples or under unbalanced treatment assignment. Similar behavior is expected for the covariate-adjusted H\'ajek estimator $\hat\tau_{\mathrm{adj}}$ in Section~\ref{sec:cov-adj-theory}; additional simulations appear in Section~\ref{sec:sm-minecraft-eg} of the SM.

\section{Empirical example: a study of child nutrition}
\label{secappl}

\subsection{Analysis of the OSNAP trial} 

We illustrate the proposed method using a pair-matched cRCT from a child nutrition study. Insufficient water consumption can impair cognitive function and overall health. In 2011, a cRCT conducted in Boston evaluated the Out of School Nutrition and Physical Activity (OSNAP) intervention, designed to encourage water intake and improve nutrition practices among children in after-school programs \citep{Giles2012, DVN_YEUCUP_2015}. Twenty after-school sites were paired based on site characteristics, with one site in each pair randomly assigned to treatment and the other to control. We refer to \cite{Giles2012} for more details.

Table~\ref{tab1:outcomes} summarizes the study data. For each site, the table reports the number of enrolled children (cluster size) and the outcome, defined as the change in daily ounces of water served per child from pre- to post-intervention. Cluster sizes vary substantially, ranging from 31 to 320 children. In forming matched pairs, the investigators appear to have prioritized other site characteristics over size, which motivates the issues discussed in Section \ref{sec:inconsistency}.
%The second row from the bottom presents the means of this outcome across treatment and control sites and the mean of paired differences. Being simple, unweighted means, the difference between treatment and control means equals the mean of paired differences. Due to random assignment, this quantity is unbiased for the site average treatment effect (site ATE), the average of 20 differences between the outcomes a site exhibited or would have exhibited under treatment and control.

\begin{table}[t]
\caption{Site sizes and outcomes of the 10 matched pairs in the OSNAP study. \emph{Size} is the number of children enrolled at each site, and \emph{Outcome} is the change in daily ounces of water served per child from pre- to post-intervention. For each pair, the last two columns report the total size of the matched pair and the treated-minus-control outcome difference. }
\label{tab1:outcomes}
\centering
\begin{tabular}{cccccccccc}
\hline
\multirow{2}{*}{Pair} && \multicolumn{2}{c}{Treatment site} && \multicolumn{2}{c}{Control site} && Total & Paired \\
&& Size & Outcome && Size & Outcome && size & difference \\
\cline{1-1} \cline{3-4} \cline{6-7} \cline{9-10}
1  && 110 & 0.00 && 320 & 0.01  && 430 & -0.01 \\
2  && 142 & 0.05 && 68  & -0.02 && 210 & 0.07 \\
3  && 75  & 0.11 && 80  & 0     && 155 & 0.11 \\
4  && 43  & 0.05 && 95  & 0     && 138 & 0.05 \\
5  && 52  & 0.06 && 67  & 0.02  && 119 & 0.04 \\
6  && 55  & 0.08 && 38  & 0     && 93  & 0.08 \\
7  && 39  & 0.09 && 46  & 0.03  && 85  & 0.06 \\
8  && 38  & 0.13 && 39  & 0     && 77  & 0.13 \\
9  && 36  & 0.04 && 40  & -0.10 && 76  & 0.14 \\
10 && 31  & 0.09 && 34  & 0     && 65  & 0.09 \\
\cline{1-1} \cline{3-4} \cline{6-7} \cline{9-10}
\multicolumn{1}{l}{Size-weighted mean} &&& 0.06$^{1}$ &&& 0.00$^{1}$ &&& 0.05$^{2}$ \\
\hline
\end{tabular}
\vspace{2pt}
\raggedright
\footnotesize
$^{1}$ Cluster-size-weighted outcome means under treatment and control. \\
$^{2}$ A pair-total size-weighted mean of paired differences.
\end{table}
The estimand of interest is the SATE for children in the study, defined as the mean difference in potential outcomes across all 1,448 children. Several estimators introduced earlier can be applied to this setting. The Horvitz-Thompson (HT) estimator computes a difference of cluster-size-weighted outcome totals under treatment and control using design weights. The IKN estimator corresponds to the pair-size-weighted mean of within-pair differences (0.05 in Table~\ref{tab1:outcomes}). The H\'ajek estimator instead takes the difference of cluster-size-weighted means under treatment and control, which are 0.06 and 0.00 in the table.

Using the H\'ajek estimator, we estimate the SATE for the change in daily water consumption per child to be 0.058 ounces, with estimated standard error 0.015. A score-based 95\% confidence interval is $(0.032, 0.129)$, while the Wald $t$-interval is $(0.027, 0.090)$. Both intervals exclude zero, indicating a statistically significant positive effect of the intervention on daily water consumption per child at the 5\% level.
\subsection{SATE estimator comparison and variance estimation calibrated to the OSNAP study}
\label{sec:OSNAPsimuATE}

We demonstrate a scenario where the H\'ajek estimator outperforms alternative SATE estimators by simulating a design that mirrors the underlying pair-level fixed-effects assumptions of \citet{Giles2012}. In their study, Giles et al. estimated that the effect of treatment on total site-level water consumption (ignoring variation in the number of children per site) was 3.6 ounces. Following their model, we impute unobserved site total water consumptions assuming a constant effect of 3.6. Then we convert site totals to per-child outcomes by dividing by program size. This imputation induces a correlation (-0.78) between the 20 programs' sizes and their treatment effects. 

For the original 10 pairs, we compute SATE estimators (Section~\ref{sec:ATEsimu}) across all $2^{10}$ possible treatment assignments, yielding exact finite-sample biases, standard deviations, and root mean square errors (columns 1-3, Table~\ref{tab1:biases}). To evaluate performance at scale, we extend the experiment to 1,000 pairs by replicating each original pair 100 times. These synthetic pairs retain the original sizes and potential outcomes but receive independent treatment randomizations, preserving the correlation structure while reflecting a large-scale trial. We conduct 10,000 Monte Carlo simulations for the scaled experiment. In each iteration, treatments are randomized independently across pairs and SATE estimates are computed. Results are in Table~\ref{tab1:biases} (columns 4-6). HA achieves the lowest RMSE in both small and large samples. %The fixed effects is fitted with a treatment indicator and nine pair fixed effects. %The Horvitz-Thompson estimator is expected to reduce bias compared to the H\'ajek effect estimator. It parallels \cite{Giles2012}'s analysis of regressing changes in water consumption per site on treatment and nine fixed effects, as they only differ by a constant time. %The Des Raj difference estimator is an improvement of the Horvitz-Thompson estimator explained by \cite{middleton2015unbiased}. The variant adopted regresses observed outcomes on treatment assignments, cluster sizes, and their interactions, and then takes the coefficient of cluster size. 
\begin{table}[ht]
\begin{center}
\begin{minipage}{\textwidth}
    \caption{Biases, standard deviations (SDs), and root mean squared errors (rMSEs) of the H\'ajek effect estimator (HA), Imai-King-Nall effect estimator (IKN), the fixed effects estimator (FE), and Horvitz-Thompson estimator (HT). Results for 10 pairs are exact over all $2^{10}$ treatment assignments; results for 1000 pairs are based on a synthetic experiment in which each pair is replicated 100 times under independent treatment assignment, with 10,000 Monte Carlo simulations. Each scenario's lowest rMSE is \underline{underlined}.
    }
    \label{tab1:biases}
    \centering
    \begin{tabular}{ccccccccc}
    \hline
     \multirow{2}{*}{Estimator} && \multicolumn{3}{c}{10 Pairs} && \multicolumn{3}{c}{1000 Pairs} \\
    && Bias & SD & rMSE && Bias & SD & rMSE \\ \cline{1-1} \cline{3-5} \cline{7-9}
     HA && 0.001 & 0.009& \underline{0.009} && 0.000 & 0.0008 & \underline{0.0008} \\ 
     IKN && 0.003& 0.012 & 0.013 && 0.003 & 0.0012 & 0.0036 \\
     FE && 0.006 & 0.012 &0.013 &&  0.006 & 0.0012 & 0.0060 \\
     HT && 0.000 & 0.013 & 0.013 && 0.000 & 0.0013 & 0.0013 \\
     \hline
\end{tabular}
\end{minipage}
\end{center}
\end{table}

The bias of HA becomes negligible at scale, consistent with previous conclusions. IKN and FE have larger biases that contribute to their MSEs. When scaled, their biases are retained but SDs decrease by a factor of 10. This causes bias to dominate RMSE. 
While HT is unbiased, HA's bias-variance trade-off is superior in practice. We therefore recommend the H\'ajek estimator for paired cRCTs with heterogeneous cluster sizes.

Besides SATE estimators, we also assess variance estimation using counterfactuals imputed from a linear model; $\hat{v}(\hat{\tau})$ outperforms HC2 in both bias and standard deviation, with all confidence intervals achieving above 95\% coverage. See Section~\ref{sec:SM-var-OSNAP} of the SM for details.

\section{Discussion}
\label{secdisc}

The inconsistency of stratum-averaging estimators in clustered trials stems from the within-stratum covariance between cluster sizes and treatment effects: when nonzero, stratumwise differences-in-means systematically mis-weight clusters, 
producing an $O(1)$ bias. Our simulations confirm that the H\'ajek estimator's advantage is most pronounced in finely stratified designs with size-heterogeneous clusters and heterogeneous treatment effects. The score-type test improves coverage over Wald-type intervals when cluster numbers are moderate, and covariate adjustment fits into the framework with one modification: centering induces block diagonality between the treatment and covariate components of the estimating equations, so the same variance estimator applies with residuals replacing raw outcomes.

We recognize several limitations. We have not studied the use of our covariance adjustment method in high-dimensional regimes. This would call for a different asymptotic analysis, and likely different techniques. %Another limitation is the theoretical guarantee of the proposed hypothesis test with calibrated variance. %Its type I error is controlled when testing Fisher's sharp null hypothesis for large sample. 
%Recent literature has derived increasingly sophisticated degrees of freedom adjustments for cRCTs from various assumptions about intra-cluster correlations of outcomes \citep{McCaffrey2002Bias, imbens2016robust, pustejovsky2017small}; we have used Normal approximation exclusively, without exploring whether or how design-based inference with H\'ajek estimates would be improved by the use of $t$ distributions. %\cite{McCaffrey2002Bias} and \cite{imbens2016robust} propose Satterthwaite-type approximations, primarily in the context of OLS regression. Extending this, \cite{pustejovsky2017small} covers WLS.  %\cite{schochet2022design} suggest a more general adjustment using $n-2\B-v^*$, where $v^*$ accounts for covariates. In simpler settings, such as pair-matched, nonclustered RCTs, the degrees of freedom is often $(n-2)/2$ \citep{imbens2011experimental}. 
Present methods may combine with techniques of \cite{Chung2013}, \cite{cohenFogarty22} and \cite{caugheyDafoeLiMiratrix23} to give finite-sample exact tests under constant effects while maintaining large-sample validity for inference about potentially heterogeneous effects, a topic to be explored in future work.

\if0\blind
{
\bigskip
\begin{center}
{\large\bf ACKNOWLEDGMENTS}
\end{center}
The authors thank Jake Bowers, Colin Fogarty, Johann Gagnon-Bartsch, and Joshua Wasserman for valuable feedback. 
} \fi

\bigskip
\begin{center}
{\large\bf SUPPLEMENTARY MATERIAL}
\end{center}

%\begin{description}
The supplementary materials contains additional regularity conditions, supplementary theoretical results, proofs, and supporting simulation and data example details.
%\end{description}

%\bibliographystyle{abbrvnat}
\bibliographystyle{agsm}

\bibliography{causal}

\makeatletter\@input{xx-sup.tex}\makeatother
\end{document}

\def\spacingset#1{\renewcommand{\baselinestretch}%
{#1}\small\normalsize} \spacingset{1}

%%%%%%%%%%%%%%%%%%%%%%%%%%%%%%%%%%%%%%%%%%%%%%%%%%%%%%%%%%%%%%%%%%%%%%%%%%%%%%

\if1\blind
{
  \title{\bf Supplementary Material for ``Design-based H\'ajek estimation for clustered and stratified experiments"}
  \author{Xinhe Wang\\
    and \\
    Ben B. Hansen \\
    Department of Statistics, University of Michigan}
  \date{}
  \maketitle
} \fi

\if0\blind
{
  \bigskip
  \bigskip
  \bigskip
  \begin{center}
    {\LARGE\bf Supplementary Material for ``Design-based H\'ajek estimation for clustered and stratified experiments"}
\end{center}
  \medskip
} \fi

\spacingset{1} % DON'T change the spacing!

\section{Regularity conditions for Theorems \ref{thm1} and \ref{prop:cons1}}

\subsection{Theorem \ref{thm1}}
\label{sec:sm-thm1-conditions}
\begin{condition}
\label{cond:lf1}
For $z=0$ and $1$, the variance $\sigma^2_z$ is positive and finite.
\end{condition}
Condition \ref{cond:lf1} guarantees that the covariance \eqref{eq:covXb} exists and $\Sigma_{11}$ is positive definite. 
To state the next condition, decompose $\sqrt{W}\Phi_n(\rho_1,\rho_0)$ into additive components $X_\p = (X_{\p0}, X_{\p1})^T$, where 
$ X_{\p z} =  \sump {\Indk} \gamsz(\rho_z)/(\sqrt{W} \ppk).  $
\begin{condition}
\label{cond:lf2}
Let $\nu_\B^2$ be the minimum eigenvalue of $\Sigma_{n}$. For any $\varepsilon>0$, 
    \begin{equation*}
        \begin{split}
        \lim_{\B\to\infty} \frac{1}{\nu_\B^2} \spb  \E\left[ \left\| X_\p - \E X_\p \right\|^2 
        \cdot 
         \I\left(\left\| X_\p - \E X_\p \right\|^2
         > \varepsilon \nu_\B^2 \right)
        \right] = 0 .
        \end{split}
    \end{equation*}
\end{condition}
Condition \ref{cond:lf2} is a Lindeberg-type condition that allows us to utilize the Lindeberg-Feller central limit theorem on $\spb X_\p=\sqrt{W}\Phi_n(\theta_n)$. It guarantees the contribution from any stratum to the covariance $\Sigma_{n}$ is small when sample size is large enough. Since $X_\p$ is a stratum sum of weight-adjusted centered outcome, this condition has a restriction on the distance of $\ypsk$ from the population mean $\rho_z$.

\iffalse
\begin{remark}
\label{rmk:lf}
Condition \ref{cond:lf2} has a set of sufficient conditions:
$\ppk\in (c,1-c)$ if $\ppk>0$, $\forall\  \p=1,\ldots,\B,\ k=0,\ldots,K$ for some $c>0$; 
    $\wps\le M,\ \s\in \p,\ \p=1,\ldots,\B$ for some $ M>0 $; and
    $ \max_{0\le k\le K} \max_{1\le \p\le \B} \max_{\s\in \p} (\ypsk-\rho_k)^2/\facB \to 0 $ as $\B\to\infty$. These conditions suggest that the nonzero propensity scores do not vanish, the weights have a uniform upper bound, and there is no extreme potential outcome.
\end{remark}
\fi

\subsection{Theorem \ref{prop:cons1}}
\label{sec:sm-cons1-conditions}
Let $\bwp=\sump \wps/\np$ be the average unit weight of stratum $\p$. As $\B$ grows to infinity, suppose that $(\spb \np\bwp^2)^2 / \spb(\np\bwp^2)^2\to\infty$ and the following regularity conditions hold:
\begin{enumerate}
    \item The following quantities remain bounded:
    $ \sumpla  \sump \left( \wps- \bwp \right)^4 /\sumpla \np\bwp^4 , $ 
    $$ \frac{1}{\sumpla \np\bwp^4}\sumpla \sump \left\{ \wps (\ypsk - \rho_z) - \frac{1}{\np} \sum_{i'\in\p} \wpst (\ypstk - \rho_z) \right\}^4  ,\  z=1,0,$$
    and $$ \frac{1}{\sumpsm \np^2 \bwp^4} \sumpsm   \sumssp \left\{ \wps(\ypsc - \rho_0) - \wpst(y_{i'1}-\rho_1) \right\}^4 . $$
    \item When there exists $\p$ with $e_\p=l$, 
    \begin{align*}
        \frac{1}{\sumpla \np\bwp^2}\sumpla  \sump \sum_{z=1,0} \left\{ \wps (\ypsk - \rho_z) - \frac{1}{\np} \sum_{i'\in\p} \wpst (\ypstk - \rho_z) \right\}^2 \ge c
    \end{align*}
    for a constant $c>0$. When there exists $\p$ with $e_\p=s $, 
    \begin{align*}
        \frac{1}{\sumpsm \np^2\bwp^2} \sumpsm   \sumssp \left\{ \wps(\ypsc - \rho_0) - \wpst(y_{i'1}-\rho_1) \right\}^2 \ge c'
    \end{align*}
    for a constant $c'>0$.
    \item For each stratum $\p=1,\ldots,\B$, the following quantities have uniform bounds: $$ \frac{\B\np\bwp^2}{\spb\np\bwp^2},\   \frac{1}{\np\bwp^2} \sump \left( \wps- \bwp \right)^2,  \text{ and }  \frac{1}{\np\bwp^2}\sump \left( \wps\ypsk -\frac{1}{\np} \sum_{i'\in\p} \wpst \ypstk \right)^2,\ z=1,0 .$$
\end{enumerate}

\section{Complementary regime: a few large strata}
\label{sec4}

\subsection{Asymptotic normality}
\label{sec4:normality}
In this section, assume that the number of strata $\B$ is fixed and the stratum sizes $\np$ tend to infinity for $\p=1,\ldots,\B$. This design is useful when there are known subgroups in the population that are expected to respond differently to the intervention. 
Recall that 
for $\s=1,\ldots,n$ and $z=0,1$, we define $r_z \mapsto \gamsz (r_z) =
\wps(\ypsk - r_z) $ as a function of $r_z\in\mathbb{R}$.
We write $\gamsz= \gamsz(\rho_z)$ and let $\bgampz =\bgampz(\rho_z)=\sump \gamsz/\np $ be the mean of $\gamsz$'s in a stratum $\p$, evaluated at $r_z$.
The stratum-wise finite population variances and covariances of $\gamsz(r_{z})$ are:
\begin{gather*}
    S_{\p,z}^2(r_z) = \frac{1}{\np-1} \sump \left[ \gamsz(r_{z}) - \bgampz(r_{z}) \right]^2,\ z=0,1; \\  S_{\p,01}(r_1,r_0) = \frac{1}{\np-1} \sump  \left( \gamsc(r_{0}) - \bgampc(r_{0}) \right) \left( \gamst(r_{1}) - \bgampt(r_{1}) \right).
\end{gather*}
We derive the asymptotic normality of $\hat\tau$ under conditions on unit weights, correlations between $\gamsc$'s and $\gamst$'s, limiting values of $\ppk$'s, and extreme values of $\gamsz$.

\begin{condition}
\label{cond:wlln2-pi}
There exists a constant $c>0$ that is independent of the increase of sample size such that for all $\p=1,\ldots,\B$ and $ z=0,1$, $\ppk\in (c,1-c)$ as $\np\to\infty$.
\end{condition}
\begin{condition}
\label{cond:wlln2-w}
The effective sample size $(\ssn\wps)^2 / (\ssn \wps^2) \to \infty$ as $\np\to\infty$ for all $\p$.
\end{condition}

Next, we utilize a central limit theorem for finite populations \citep{Li2017} to derive the asymptotic distribution of $\Phi_n(\theta_n)$.
We apply the theorem to $ \gamsz $ in each large stratum %Since $\spb X_\p$ equals the first two entries of $\sqrt{W} G_n(\theta_n)$, we can 
and combine the results of every stratum to derive the asymptotic distribution of $\Phi_n(\theta_n)$. %, as the sum of ``observed" mean of $\gamsz$'s, $\spb \bgamobsz(\rho_z)$, is proportional to $( G_n(\theta_n))_{1:2}$.
Recall that the mean of unit weights in stratum $\p$ is denoted as $\bwp = \sump \wps/\np$.

\begin{condition}
\label{cond:clt-var}
For $\p=1,\ldots,\B$ and $z=0,1$, the stratumwise variances and covariances $S_{\p,z}^2(\rho_z)/ \bwp$ and $S_{\p,01}(\rho_1, \rho_0)/\bwp $
have finite limits as $\np$ tends to infinity.
\end{condition}

\begin{condition}
\label{cond:clt-prop}
The assignment probabilities $\ppk$ have positive limiting values for all $z$ and $\p$ as $\np$ tends to infinity.
\end{condition}

\begin{condition}
\label{cond:clt-square}
For any $\p$, $\max_{z=0,1}\max_{\s\in \p}  \left[ \wps (\ypsk-\rho_z) - \sump \wps(\ypsk - \rho_z)/\np \right]^2  / (\np\bwp) $ tends to zero
as $\np$ tends to infinity.
\end{condition}

Conditions \ref{cond:clt-var}--\ref{cond:clt-square} are motivated by the conditions of the finite population central limit theorem \citep{Li2017}. Note that by proper scaling of the potential outcomes, the diagonal elements of $S_{\p,z}^2/(\np^2 \bwp)$ can always have limits. Thus, Condition \ref{cond:clt-var} essentially requires that the correlation between $\gamsc$ and $\gamst$ in each stratum divided by $(\ssn \wps)^2$ converges.  Condition \ref{cond:clt-prop} implies Condition \ref{cond:wlln2-pi}, the strong overlap condition. Condition \ref{cond:clt-square} holds with probability one when $\wps(\ypsk-\rho_z)/\sqrt{\bwp}$ is bounded or i.i.d.\ drawn from a super-population with more than two moments. %Given that the number of strata is known and remains constant as the stratum sizes grow, we can apply the theorem to each stratum and derive the following lemma. 

In order to combine the results of all strata, we impose a requirement on the proportion of stratum weights $\np\bwp/(\ssn\wps) $ to ensure that no stratum dominates.

\begin{condition}
\label{cond:wt-proportion}
    The weight proportion ${\np\bwp}/(\ssn\wps) $ has a positive limiting value as $\np\to\infty$ for $\p=1,\ldots,\B$ with $\B$ fixed.
\end{condition}

\begin{theorem}
\label{thm2}
Under Conditions \ref{cond:wlln2-w}--\ref{cond:wt-proportion}, we have $(\ssn\wps)^{1/2} (\hat\theta_n - \theta_n) \simd \mathcal{N}(0, \Sigma_n)$ and $(\ssn\wps)^{1/2}(\hat\tau - \tau) \simd \mathcal{N}(0, \tsigma_0^2 - 2\tsigma_{01} + \tsigma_1^2)$, where $\Sigma_n, \tsigma_0^2, \tsigma_{01}, $ and $\tsigma_1^2$ are defined by \eqref{eq:sigmas} and \eqref{eq:covXb}.
\end{theorem}

Theorem \ref{thm2} establishes the consistency of $\hat\tau$ as the strata sizes grow infinitely. For unstratified cRCTs, our result is consistent with \cite{schochet2022design}'s Theorem 1 on the asymptotic distribution of the ratio estimator for one large stratum.

\subsection{Variance estimation}
\label{sec4:var}
We can rely on the sample variances of $\gamsz$ to construct variance estimators. 
%We have that $ s_{\p z}^2$ (defined in Section \ref{sec:varlarge}) is unbiased for $S_{\p,z}^2$. 
As $2{\tsigma_{01}} \le \spb\np [S^2_{\p,0}(\rho_0) + S^2_{\p,1}(\rho_1)]/(\ssn\wps) $,
%$$ 2{\tsigma_{01}} = -\frac{2}{W} \spb\frac{1}{\np} S_{\p,01}(\rho_1, \rho_0) \le \frac{1}{W}\spb \frac{1}{\np}S^2_{\p,0}(\rho_0) + \frac{1}{W} \spb \frac{1}{\np} S^2_{\p,1}(\rho_1) , $$
the large-sample approximation of $\var(\hat\tau)$ has an upper bound: %$ \left( {\tsigma_0^2}- 2\tsigma_{01} + \tsigma_1^2 \right) /W \le  \spb [  {S^2_{\p,0}(\rho_0)}/{\npc}   + {S^2_{\p,1}(\rho_1) }/{\npt} ]/W^2.$
\begin{align*}
    \var(\hat\tau) \approx \frac{1}{W}\left( {\tsigma_0^2}- 2\tsigma_{01} + \tsigma_1^2 \right) \le \frac{1}{W^2} \spb\np^2 \left(  \frac{S^2_{\p,0}(\rho_0)}{\npc}   + \frac{S^2_{\p,1}(\rho_1) }{\npt}   \right).
\end{align*}
Here $\bgamobsz(r_{z})$ is the mean of $\gamsz(r_{z})$'s for units \textit{assigned to} treatment $z$ in stratum $\p$ evaluated at $r_z$:
$\bgamobsz (r_z) =  \sum_{i\in\p,\zps=z} \gamsz/\npk.$ Let
\begin{equation}
    \nu_b(r_1,r_0) = \frac{s^2_{\p0}(r_0)}{\npc} + \frac{s^2_{\p1}(r_1)}{\npt}= \frac{\sum_{i\in\p,\zps=0} [\gamsc(r_{0}) - \bgamobsc(r_{0})]^2}{(\npc-1)\npc}+ \frac{\sum_{i\in\p,\zps=1} [\gamst(r_{1}) - \bgamobst(r_{1})]^2}{(\npt-1)\npt} 
    \label{eq:nu-large}
\end{equation}
be a function of $r_0$ and $r_1$, where
$ s_{\p z}^2 (r_z)$ %\sumpindk
is the stratum and treatment specific sample variance of $\gamsz(r_{z})$ and is unbiased for $S_{\p,z}^2(r_z)$, $z=0,1$. 
An estimator of the variance of $\hat\tau$ is
\begin{equation}
    \hat{v}(\hat\tau) = \frac{1}{W^2} \spb \np^2 \nu_\p(\hat\rho_1,\hat\rho_0) ,
    \label{eq:case2varest}
\end{equation}
where $\hat\rho_1$ and $\hat\rho_0 $ are plugged in as estimates of the unknown parameters $\rho_1$ and $\rho_0$ in the variance.
To study the properties of this estimator, assume, without loss of generality, that the parameters $\rho_1$ and $\rho_0$ converge to finite values as the sample size grows to infinity. This condition is achievable through proper scaling of the potential outcomes. 
\begin{theorem}
\label{prop-cons2}
Let $\bwp=\sump \wps/\np$. 
Suppose that there exists a constant $c>0$ such that 
\begin{align*}
    \frac{1}{\spb \np\bwp^2}\spb \sump \sum_{z=0,1} \left[ \wps (\ypsk - \rho_z) - \frac{1}{\np} \sum_{i'\in\p} \wpst (\ypstk - \rho_z) \right]^2 \ge c ,
\end{align*}
and for each stratum $\p=1,\ldots,\B$, the following quantities remain bounded as $\np$ grows: $$ \frac{1}{\np\bwp^4} \sump \left( \wps-\bwp \right)^4, \ \text{ and } \ \frac{1}{\np\bwp^4}\sump \left[ \wps\ypsk -\frac{1}{\np} \sum_{i'\in\p} \wpst \ypstk \right]^4,\ z=0,1. $$ Then the variance estimator is asymptotically conservative in the sense that there exists a sequence $\{\xi_n, \ n\ge 1\}\subseteq \mathbb{R}$ such that $(\ssn\wps)\hat{v}(\hat\tau) / \xi_n \xrightarrow{p} 1$ as $\np\to\infty$ for all $\p$ and $\liminf_{\np\to\infty, \forall \p}[\xi_n/\var\left( (\ssn\wps)^{-1/2}\tilde\tau\right) ]\ge 1,$ where equality holds when the unit treatment effect, $\ypst-\ypsc$, is constant for all units in the experiment.
\end{theorem}

This theorem establishes the conditions under which the variance estimator has a consistent relationship with the asymptotic variance of $\hat\tau$.
Similarly to the tests proposed in Section~\ref{sec:hypothesis}, we can construct hypothesis tests based on this variance estimator.

\section{Proofs}

\subsection{Proposition \ref{prop:wasdom}}
\begin{proof}[Proof of Proposition \ref{prop:wasdom}]
Suppose pair $\p$ in a paired cRCT contains two clusters of sizes $\pairmi,\pairmii$ and cluster-mean potential outcomes $y_{bi1}, y_{bi0}$, $i=1,2$. Then the cluster treatment effects of these two clusters are $\tau_i = y_{bi1}- y_{bi0},\ i=1,2$ and the pair SATE is $\tau_{\text{pair }\p} = {(\pairmi \pairtaui + \pairmii \pairtauii)}/{(\pairmi + \pairmii)}.$ 

The IKN effect estimator estimates the pair SATE by $ \hat\tau_{\text{pair}} =\pairzi (\pairyti-\pairycii) + z_2 (\pairytii-\pairyci)$, with $z_1$ and $ z_2$ encoding treatment assignment of the pair, so that $z_1 + z_2 = 1$. It has $\E [\hat\tau_{\text{pair } \p}] =  (\pairyti-\pairycii)/2 + (\pairytii-\pairyci)/2  = (\pairtaui +\pairtauii )/2,$ and therefore a bias of $\E \left[\hat\tau_{\text{pair }\p}\right] - \tau_{\text{pair }\p} = (\pairmi-\pairmii)\cdot(\pairtaui - \pairtauii)/[2{(\pairmi+\pairmii)}].$ Averaged across pairs, the bias of the IKN effect estimator is 
\begin{align*}
    \frac{1}{\ssn \wps} \spb \left(\pairmi+\pairmii \right) \cdot \left( \E \left[\hat\tau_{\text{pair } \p}\right] - \tau_{\text{pair } \p} \right) =  \frac{1}{\ssn \wps} \spb \frac{1}{2} (\pairmi-\pairmii)\cdot(\pairtaui - \pairtauii),
\end{align*}
which simplifies to equation \eqref{eq:wasdom-bias}.
Similarly, for the FE estimator, the bias is
\begin{align*}
    \spb \tilde{w}_{\p}^\text{(f)} \cdot \left( \E \left[\hat\tau_{\text{pair } \p}\right] - \tau_{\text{pair } \p} \right),
\end{align*}
which yields equation \eqref{eq:wasdom-bias2}.
%Given non-negligible variation in cluster sizes and cluster size/effect correlation (%the paired variance $\spb \tilde w_\p\sum_{i,i'\in\p,i\neq i'}(\wps-\wpst)^{2}$ and $\spb \tilde w_\p\sum_{i,i'\in\p,i\neq i'}(\wps-\wpst)(\tau_i-\tau_{i'})$ tending to nonzero limits), the bias of  $\hat{\tau}_{\text{DA}}$ has a nonzero limit.  While its variance tends to zero if the effective sample size $\big(\spb \tilde w_{\p}^{2}\big)^{-1}(\spb \tilde w_{\p})^{2} = \big(\spb \tilde w_{\p}^{2}\big)^{-1}$ grows infinitely, $\hat{\tau}_{\text{DA}}$ fails to be consistent.
\end{proof}

\subsection{Theorem \ref{thm1}}

\subsubsection{Lemma \ref{lem:wlln}}

\begin{lemma}
\label{lem:wlln}
Under Conditions \ref{cond:wlln-pi} and \ref{cond:wlln-w}, $-\nabla \Phi_n \xrightarrow{p}  I_2$ as $\B\to\infty$.
\end{lemma}

\begin{proof}[Proof of Lemma \ref{lem:wlln}]
    For unit $\s$ and treatment $z$, Denote $\vsk =  \Indk\ypsk/\psk, \   \usk =  {\Indk}/{\psk} $, and $W=\ssn \wps$. Then we have 
    $$ \psi_i(r_1, r_0) = \left( \begin{array}{c}
     \{\Indt/\pst\} (\ypstt - r_1) \\
     \{\Indc/\psc\} (\ypsc - r_0)    \\
\end{array} \right) = \left( \begin{array}{c}
     \vst - r_1 \ust \\
     \vsc - r_0 \usc    \\
\end{array} \right) $$
    The gradient of $\Phi_n$ with respect to $\theta$ is
    $$ \nabla \Phi_n = \left(\begin{array}{ccc}
            -\ssn\wps\ust/\facB &  0\\
             0& -\ssn\wps\usc/\facB
            \end{array} \right).$$
    For $z=0,1$ and $b=1,\ldots,\B$, 
    $$ \E\left[\sump\wps\usk\right] = \E \left[\sump \frac{\Indk}{\ppk}\wps \right]  =  \sump \wps $$
    and
    \begin{equation*}
        \begin{split}
            \var\left( \sump\wps\usk\right) = & \  \E \left( \sump\frac{\Indk}{\ppk^2}\wps^2 + \sump\sumssp \frac{\Indk \Indsk }{\ppk^2}\wps\wpst  \right) - \left(\sump \wps\right)^2 \\
            =&\ \sump \frac{1}{\ppk}\wps^2 + \sump \sumssp \frac{\E \Indk \Indsk}{\ppk^2} \wps\wpst  - \left(\sump \wps\right)^2 \\
            =&\  \sump \left(\frac{1}{\ppk} - 1\right) \wps^2 + \sump \sumssp \left(\frac{\delbkk}{\ppk^2}-1\right) \wps\wpst 
            .%\\ =&\ \sump \left(\frac{1}{\ppk} - 1\right) \wps^2 + \sump \sumssp \frac{\npk - \np}{\npk(\np-1)} \wps\wpst.
        \end{split}
    \end{equation*}
    If $\npk=1$, then $\delbkk=0$ and the variance is simplified to
    \begin{align*}
        \var\left( \sump\wps\usk\right) = \left(\frac{1}{\ppk}-1\right) \sump \wps^2 - \sump \sumssp \wps \wpst ,
    \end{align*}
    so 
    $$ \left|\var\left( \sump\wps\usk\right) \right| \le \left(\frac{1}{\ppk}-1\right) \sump \wps^2 + \frac{1}{2} \sump \sumssp (\wps^2 + \wpst^2) = \frac{1}{\ppk} \sump \wps^2. $$
    Otherwise, under Condition \ref{cond:wlln-pi}, $$\frac{\delbkk}{\ppk^2}-1 = \frac{\npk - \np}{\npk(\np-1)} = O({\np}^{-1}).$$
    Under Condition \ref{cond:wlln-w}, $\spb \sump \wps^2 / (\ssn\wps)^2 \to 0 $ and by Cauchy-Schwartz inequality, $$\frac{1}{\facB^2}\spb \sumssp\frac{1}{\np} \wps \wpst \le \frac{1}{\facB^2}\spb \frac{1}{\np} \left(\sump\wps \right)^2 \le \frac{1}{\facB^2}\spb \left(\sump \wps^2 \right) \to 0. $$ 
    Thus, 
    $$  \E\left[\frac{1}{\facB}\ssn\wps \usk\right] = \frac{1}{\facB}\ssn \wps = 1 
    $$
    and the independence of treatment assignment between strata implies that
    \begin{align*} 
    \var \left(\frac{1}{\facB}\ssn\wps \usk\right) = \frac{1}{\facB^2} \sum_{b=1}^B  \left[\sump \left(\frac{1}{\ppk} - 1\right) \wps^2 + \sump \sumssp \left(\frac{\delbkk}{\ppk^2}-1\right)  \wps\wpst \right]\to 0.
    \end{align*}
    By Chebyshev's inequality, for any $ \epsilon>0$,
    $$ \P\left(\left|\frac{1}{\facB}\ssn\wps\usk - 1 \right|>\epsilon \right) \le \frac{\var\left(\frac{1}{\facB}\ssn\wps \usk\right)}{ \epsilon^2} \to 0. 
    % = \frac{\spb \var\left( \sump\wps \usk\right)}{\facB^2\epsilon^2}
    $$
    Therefore, $\ssn\wps\usk/W \xrightarrow{p} 1$ and
    $$ -\nabla \Phi_n = \left(\begin{array}{ccc}
            \ssn\wps\ust/\facB & 0 \\
            0 & \ssn\wps\usc/\facB 
            \end{array} \right) \xrightarrow{p} I_2.
    $$
\end{proof}

\subsubsection{Covariance of $\Phi_n(\rho_1, \rho_0)$}

In this section we show that equation \eqref{eq:covXb} about the covariance of $ \Phi_n(\rho_1, \rho_0) = W^{-1/2}\spb X_\p$ holds.
The expectation of $X_\p$ is 
$$\E(X_\p) =  \frac{1}{\sqrt{W}} \left( \sump \wps (\ypst - \rho_1), \ \sump   \wps (\ypsc - \rho_0)\right)^T. $$
Define the ``second-order" treatment assignment probability in stratum $\p$ as $\delbkj = \E [\indk \indsj] ,\ z,\tilde{z}\in\{0,1\}$, where $\s\neq \s'$ and $\s, \s' \in b$. If $z=\tilde{z}$, then $\delbkk = \npk(\npk-1)/[\np(\np-1)]$. Otherwise, $\pi_{\p,01}=\npt\npc/[\np(\np-1)] $.

\begin{proof}[Proof of Equation \eqref{eq:covXb}]
We calculate each element of the covariance matrix.
For $z=0$ or $1$, % Let $\vpk = \sump \wps \vsk$ and $\upk = \sump \wps \usk$ be the weighted sums of $\vsk, \usk$ in block $b$ under treatment $k$. 
    \begin{align*}
        & \var\left(  \sump \frac{\Indk}{\ppk}  \wps(\ypsk - \rho_z) \right)  \\
        & =  \sump \frac{1}{\ppk} \wps^2 (\ypsk-\rho_z)^2  + \sump\sumssp \frac{\delbkk} {\ppk^2} \wps \wpst (\ypsk- \rho_z)( \ypstk- \rho_z)  - \left[\sump \wps( \ypsk - \rho_z)\right]^2 ,
    \end{align*}
    and 
    \begin{align*}
        & \cov\left(\sump \frac{\Indc}{\ppc}  \wps(\ypsc - \rho_0), \sump \frac{\Indt}{\ppt}  \wps(\ypst - \rho_1) \right) \\
        &= -\sump \wps^2(\ypst-\rho_1) (\ypst-\rho_1 ) + \sump\sumssp \left(\frac{\delbct } {\ppc\ppt}-1 \right) \wps \wpst (\ypsc-\rho_0) (\ypsstt-\rho_1 ) 
    \end{align*}
    because
    $\cov(\Indc, \Indt) = -\ppc\ppt $ and
    $\cov(\Indk, \Indsj) =  \delbkj - \ppk\pi_{\p \tilde z}. $
    Then by the independence of treatment assignment across different strata, we have 
\begin{align*}
    &\var\left(\spb X_{\p z}\right)=\spb\var(X_{\p z})\\
     & = \frac{1}{\facB} \spb  \Bigg[  \left(\frac{1}{\ppk} - 1\right)\sump \wps^2(\ypsk - \rho_z)^2  +  \left(\frac{\delbkk}{\ppk^2}-1\right) \sump \sumssp \wps \wpst(\ypsk - \rho_z)(\ypstk - \rho_z) \Bigg] 
\end{align*}
for $z=0,1$, and 
\begin{align*}
 \cov\left(\spb X_{\p 0},\spb X_{\p 1}\right)= &\ \spb\cov(X_{\p 0},X_{\p 1}) \\
 = &\ \frac{1}{\facB} \spb \Bigg[ (-1) \sump \wps^2(\ypsc - \rho_0)(\ypstt - \rho_1)\\
&\ +  \left(\frac{\delbct}{\ppc\ppt}-1\right) \sump \sumssp \wps \wpst(\ypsc - \rho_0)(\ypsstt - \rho_1) \Bigg]  .
\end{align*}
The variances of $\gamsz$ for a particular stratum can be written as
    \begin{align*}
        S_{\p,z}^2 =&\  \frac{1}{\np - 1} \sump (\gamsz - \bgampz)^2 = \frac{1}{2\np(\np-1)} \sum_{\s,\s'\in\p} (\gamsz - \gamssz)^2 \\
        = &\  \frac{1}{2\np(\np-1)} \left( \sump2 \np \gamsz^2 - \sum_{\s,\s'\in\p} 2 \gamsz \gamssz \right) \\
        =&\  \frac{1}{\np(\np-1)} \left( (\np-1) \sump \gamsz^2 - \sump\sumssp \gamsz \gamssz \right) 
    \end{align*}
    Thus,
    \begin{align*}
        \sigma_z^2 = \frac{1}{W} \spb &\left[ \left( \frac{\np}{\npk}- 1 \right) \sump {\gamsz(\rho_z)^2}  - \left(\frac{\np}{\npk(\np-1)} - \frac{1}{\np-1}\right) \sump\sumssp {\gamsz(\rho_z)}{\gamssz(\rho_z)} \right] \\
        = \frac{1}{W} \spb &\left[ \left( \frac{\np}{\npk}- 1 \right) \sump {\gamsz(\rho_z)^2}  - \frac{\np-\npk}{\npk(\np-1)}  \sump\sumssp {\gamsz(\rho_z)}{\gamssz(\rho_z)} \right]\\
        = \frac{1}{\facB} \spb & \left[  \left(\frac{1}{\ppk} - 1\right)\sump \wps^2(\ypsk - \rho_z)^2  +  \left(\frac{\delbkk}{\ppk^2}-1\right) \sump \sumssp \wps \wpst(\ypsk - \rho_z)(\ypstk - \rho_z) \right],
    \end{align*}
    because
    $$ \frac{\delbkk}{\ppk^2}-1= \frac{\npk(\npk-1)}{\np (\np-1)} \cdot \frac{\np^2}{\npk^2} - 1 = \frac{\np-\npk}{\npk(\np-1)}. $$
    This proves that the diagonal elements of $ \cov(\Phi_n(\rho_1,\rho_0))$ equals $\sigma_1^2/W$ and $\sigma_0^2/W$, respectively. Besides,
    The covariance of $\gamsc,\gamst$ for a particular stratum can be written as 
    \begin{align*}
        S_{\p,01} =&\  \frac{1}{\np - 1} \sump (\gamsc- \bgampc)(\gamst - \bgampt) = \frac{1}{2\np(\np-1)} \sum_{\s,\s'\in\p} (\gamsc- \gamssc)(\gamst - \gamsst) \\
        =&\  \frac{1}{\np(\np-1)} \left( (\np-1) \sump \gamsc \gamst - \sump\sumssp \gamsc \gamsst \right).
    \end{align*}
    Hence, 
    \begin{align*}
        \sigma_{01} =  - \frac{1}{W} \spb & \Bigg[\sump {\gamsc(\rho_0)}{\gamst(\rho_1)} - \frac{1}{\np-1} \sump\sumssp {\gamsc(\rho_0)}{\gamsst(\rho_1)}\Bigg]\\
        = \frac{1}{\facB} \spb & \Bigg[ -\sump \wps^2(\ypsc - \rho_0)(\ypstt - \rho_1) \\
        &\  +  \left(\frac{\delbct}{\ppc\ppt}-1\right) \sump \sumssp \wps \wpst(\ypsc - \rho_0)(\ypsstt - \rho_1) \Bigg], 
    \end{align*}
    which is because 
    $$ \frac{\delbct}{\ppc\ppt}-1 = \frac{\npt\npc}{\np(\np-1)}\cdot \frac{\np}{\npc}\frac{\np}{\npt} - 1 = \frac{\np}{\np-1}-1 = \frac{1}{\np-1} .$$ 
    Therefore, the off-diagonal element of $ \cov(\Phi_n(\rho_1,\rho_0))$ equals $\sigma_{01}/W$.
    
\end{proof}

\subsubsection{The proof of Lemma \ref{lem:lfclt}}

\noindent First, we state the multivariate Lindeberg-Feller central limit theorem (CLT).
\begin{lemma}%[Multivariate Lindeberg-Feller CLT]
\label{lem:clt}
For each $n\ge1$, let $\{X_{n,m}:1\le m\le n\}$ be independent random vectors with $\E(X_{n,m})=\mu_{n,m}$. Let 
    \begin{equation*}
    V_n = \sum_{m=1}^n \cov (X_{n,m}) 
    \label{condclt1}
    \end{equation*}
    and $\nu_n^2 = \lambda_{min}(V_n)$ be the minimum eigenvalue of $V_n$. Suppose $\nu_n^2 > 0$ and
    \begin{equation}
        \lim_{n\to\infty} \frac{1}{\nu_n^2}\sum_{m=1}^n \E \left[\|X_{n,m}- \mu_{m,n}\|^2 \I \left(\|X_{n,m}-\mu_{m,n}\|^2 > \varepsilon \nu_n^2 \right) \right] = 0, \  \forall \varepsilon>0. 
        \label{eq:condclt}
    \end{equation} Then as $n\to\infty$,
    $$ V_n^{-\frac{1}{2}}  \sum_{m=1}^n \left( X_{n,m} -\mu_{n,m}  \right) \xrightarrow{d} \mathcal{N}(0,I).$$
\end{lemma}

\iffalse
\begin{lemma}
\label{lem:multi}
    The random vector $X=(X_1,\ldots,X_p)^T$ is $\mathcal{N}(\mu, \Sigma)$ distributed if and only if for any constant vector $c=(c_1,\ldots,c_p)^T$, $c^T X$ is $\mathcal{N}(c^T\mu, c^T\Sigma c)$ distributed.
\end{lemma}
\fi

\begin{lemma}
\label{lem:lfclt}
Under Conditions \ref{cond:lf1} and \ref{cond:lf2}, $\sqrt{W} \Phi_n(\rho_1,\rho_0) \simd\mathcal{N}(0, \Sigma_n)$ as $\B\to\infty$. %the upper left $2\times2$ block of $\Sigma_n^{-1/2}\in \mathbb{R}^{3 \times 3}$ is given by $ \Sigma^{-{1}/{2}}_{11}$ and its remaining entries are zeros.
%$$ \Sigma_n^{-\frac{1}{2}} = \left( \begin{array}{cc}
%        \Sigma^{-\frac{1}{2}}_{11} & 0 \\
%        0 & 0
%    \end{array} \right) \in \mathbb{R}^{3 \times 3} $$ 
%$\Sigma_{11}^{-1/2} = O\Lambda^{-1/2} O^T$ if $\Sigma_{11} =O\Lambda O^T $ is the diagonalization of $\Sigma_{11}$.
\end{lemma}

\begin{proof}[Proof of Lemma \ref{lem:lfclt}]
    In this proof, we apply the Lindeberg-Feller CLT to $X_\p$ defined in Section \ref{sec3-1}. We have derived the covariance matrix of $\spb X_\p$. Condition \ref{cond:lf2} implies the Lindeberg condition (\ref{eq:condclt}), so by Lemma \ref{lem:clt}, 
    $$\Sigma_{n}^{-\frac{1}{2}}\cdot  \spb \left(X_\p - \E X_\p\right) \ \xrightarrow{d} \  \mathcal{N}(0, I_2)\ \ \text{as}\ \B\to\infty.$$
    By the definition of $\rho_{z}$, the $(z+1)$th entry of $\spb  \E X_\p $ is
    \begin{equation*}
        \begin{split}
         \spb\sump  \wps(\ypsk-\rho_{z}) = \left(\frac{\spb\sump  \wps\ypsk}{\spb\sump  \wps}-\rho_{z} \right)
         \cdot \left(\spb\sump  \wps\right) = 0,\ z=0,1,
        \end{split}
    \end{equation*}
    therefore,
    $$\Sigma_{n}^{-\frac{1}{2}}\cdot\Phi_n(\rho_1,\rho_0) =\Sigma_{n}^{-\frac{1}{2}}\cdot \spb X_\p \ \xrightarrow{d} \  \mathcal{N}(0, I_2) .$$
    
\end{proof}

\subsubsection{The proof of Theorem \ref{thm1}}
\begin{proof}[Proof of Theorem \ref{thm1}]
By Lemmas \ref{lem:wlln} and \ref{lem:lfclt} and Slutsky's Theorem,
\begin{align*}
    \Sigma_{n}^{-\frac{1}{2}} \cdot
        \sqrt{\facB} (\hat\theta_n - \theta_n) &= \ \Sigma_{n}^{-\frac{1}{2}}\cdot I_2 \cdot \left(-\nabla \Phi_n\right)^{-1} \sqrt{\facB}\Phi_n(\theta_n) \\
    &= \  \Sigma_{n}^{-\frac{1}{2}} \cdot I_2 \cdot \left(-\nabla \Phi_n\right)^{-1} \cdot  \Sigma_{n}^{\frac{1}{2}} \cdot  \Sigma_{n}^{-\frac{1}{2}} \cdot \sqrt{\facB} \Phi_n(\rho_1,\rho_0)  \\
    &\xrightarrow{d}  \ \mathcal{N}(0, I_2).
\end{align*}
Then, note that
    \begin{align*}
          \left( \begin{array}{ccc}
       1 & -1 
    \end{array} \right) \left( \begin{array}{ccc}
        \sigma_0^2 & \sigma_{01}  \\
        \sigma_{01} & \sigma_1^2   
        \end{array} \right)\left( \begin{array}{ccc}
       1  \\
       -1
    \end{array} \right) 
     = \sigma_0^2 - 2\sigma_{01} + \sigma_1^2,
    \end{align*}
    thus, $(\sigma_0^2 - 2\sigma_{01} + \sigma_1^2)^{-1/2} (\hat\tau-\tau) \xrightarrow{d} \mathcal{N}(0,1).$
\end{proof}

\subsection{Proposition \ref{lem:ab2}}

We first state a more complete version of Proposition \ref{lem:ab2}.
\begin{proposition}
\label{lemS:ab}
Consider an RCT with $n$ units, where $n_1$ units are randomly assigned to treatment and $n_0$ to control. For each subject $\s$, let $\ai$ and $\bi$ be its potential outcomes to treatment and control. Denote $\baraobs$ and $\barbobs$ as the observed mean responses for units assigned to treatment and control, respectively. Let $\bara = \sum_\s \ai / n$, $\barb = \sum_\s \bi /n$, $\sigma_1^2 = \sum_\s (\ai -\bara)^2 / n$, $\sigma_0^2 = \sum_\s (\bi -\barb)^2 / n$, and $\sigma_{10} = \sum_\s (\ai-\bara)(\bi-\barb) / n$. Then
$$ \E (\baraobs) = \bara,\ \ \var(\baraobs) = \frac{n-n_1}{n-1}\frac{\sigma_1^2}{n_1}, \ \ \cov(\baraobs, \barbobs) = - \frac{1}{n-1} \sigma_{10},\ \  \text{and} $$
$$ \var(\baraobs - \barbobs) = \frac{n-n_1}{n-1}\frac{\sigma_1^2}{n_1} + \frac{n-n_1}{n-1}\frac{\sigma_0^2}{n_1} + \frac{2}{n-1} \sigma_{10} \le \frac{n}{n-1} \left(\frac{\sigma_1^2}{n_1} + \frac{\sigma_0^2}{n_1} \right). $$
Let $y_{1|1},y_{2|1},\ldots,y_{n_1|1}$ and $y_{1|0},y_{2|0},\ldots, y_{n_1|0}$ denote the $n_1$ and $n_0$ observed responses for treatments A and B, respectively. Let $s_1^2$ and $s_0^2$ be the sample variances of observed responses $y_{i|1} $'s and $y_{i|0}$'s, respectively. The following are two upward-biased or conservative estimators of $\var(\baraobs - \barbobs)$:
$$ \frac{s_1^2}{n_1} + \frac{s_0^2}{n_0} , $$
and it has a nonnegative bias of $ ( \sigma_1^2 + \sigma_0^2- 2\sigma_{10} )/(n-1)$;
$$ \frac{1}{n_1 n_0} \sum_{i=1}^{n_1} \sum_{j=1}^{n_1} (y_{i|1}-   y_{j|0})^2 - \frac{1}{n_1} \sum_{i=1}^{n_1} (y_{i|1} - \baraobs)^2 - \frac{1}{n_0} \sum_{j=1}^{n_0} (y_{j|0} - \barbobs)^2 %- \frac{n_1-1}{n_1} s_1^2- \frac{n_1-1}{n_1} s_0^2
, $$ which is nonnegative when $n_1=1$ or $n_0=1$, having a nonnegative bias of $(\bara - \barb)^2$.
\end{proposition}

\begin{proof}[Proof of Proposition \ref{lemS:ab}]
By simple algebra,
$$ \ssn \sum_{j\neq \s}(\ai-\bara)(y_{j0}-\barb) + \ssn(\ai-\bara)(\bi-\barb)  =  \ssn \sum_{j=1}^n(\ai-\bara)(y_{j0}-\barb) = 0,$$
so we have
\begin{align*}
    &\E \left[ \frac{1}{n_1 n_0} \sum_{i=1}^{n_1} \sum_{j=1}^{n_0} (y_{i|1}-   y_{j|0})^2 \right] = 
    \frac{1}{n_1 n_0} \sum_{i=1}^n \sum_{\substack{j=1 \\ j\neq i}}^n \frac{n_1}{n} \frac{n_0}{n-1} (\ai - y_{j0})^2 \\
    &=  \frac{1}{n(n-1)} \sum_\s \sum_{j\neq i} [\ai - \bara -(y_{j0} - \barb) + \bara - \barb]^2 \\
    &=   \frac{1}{n(n-1)}\sum_\s \sum_{j\neq i} [(\ai - \bara)^2 + (y_{j0} - \barb)^2 - 2(\ai-\bara)(y_{j0}-\barb)] + (\bara - \barb)^2 \\
    &=  \sigma_1^2 + \sigma_0^2 + \frac{2}{n-1}\sigma_{10} + (\bara - \barb)^2.
\end{align*}
Combined with the fact that $s_1^2 $ is unbiased for $ n/(n-1)\cdot \sigma_1^2$, the expectation of the proposed estimator is
\begin{align*}
    &E \left[ \frac{1}{n_1 n_1} \sum_{i=1}^{n_1} \sum_{j=1}^{n_1} (y_{i|1}-   y_{j|0})^2  \right] - \frac{n(n_1-1)}{n-1}\frac{\sigma_1^2}{n_1} - \frac{n(n_0-1)}{n-1}\frac{\sigma_0^2}{n_0} \\
    &= \frac{n-n_1}{n-1}\frac{\sigma_1^2}{n_1} + \frac{n-n_0}{n-1}\frac{\sigma_0^2}{n_0} + \frac{2}{n-1}\sigma_{10} + (\bara - \barb)^2 \\
    &\ge   \frac{n-n_1}{n-1}\frac{\sigma_1^2}{n_1} + \frac{n-n_0}{n-1}\frac{\sigma_0^2}{n_0} + \frac{2}{n-1}\sigma_{10} = \var(\baraobs - \barbobs).
\end{align*}
Thus,
$$ \frac{1}{n_1 n_0} \sum_{i=1}^{n_1} \sum_{j=1}^{n_0} (y_{i|1}-   y_{j|0})^2 - \frac{n_1-1}{n_1} s_1^2- \frac{n_0-1}{n_0} s_0^2$$
is a conservative estimator of $\var(\baraobs - \barbobs) $, and it has a bias of $ (\bara - \barb)^2 $.
When $n_1=1$, the above estimator becomes
$$\frac{1}{n_0} \sum_{j=1}^{n_1} (y_{1|1}-   y_{j|0})^2- \frac{n_0-1}{n_0} s_0^2 = \frac{1}{n_0} \sum_{j=1}^{n_0} (y_{1|1}-y_{j|0})^2 - \frac{1}{n_0} \sum_{j=1}^{n_0} (y_{j|0} - \barbobs)^2, $$
which is nonnegative because $\barbobs$ minimizes $ \sum_{j=1}^{n_1} (y_{j|0} - u)^2$ as a function of $u.$  
\end{proof}

\subsection{Remark \ref{rmk:alterv}}

\begin{remark}
\label{rmk:alterv}
    A computationally expedient equivalent to the small stratum variance estimator $ \nu_{\p}^{(s)} (r_1,r_0)$ (defined in \eqref{eq:est3}) is
    \begin{align*}
      %= \frac{2}{\npc \npt}  \binom{\np }{2} s_{\p(0,1)}^2(r_0,r_1)- \left( \frac{1}{\npc} + \frac{1}{\npt} \right)\left[ (\npc-1) s_{\p 0}^2 + (\npt-1) s_{\p 1}^2 \right] 
      &\frac{2}{\npt \npc}  \binom{\np }{2} s_{\p(1,0)}^2(r_1,r_0)\\
      &- 
      \left( \frac{1}{\npt} + \frac{1}{\npc} \right)\left[ \sum_{\substack{i\in\p, \\ \Zps=1}} \{\gamst(r_{1})- \bgamobst(r_{1})\}^2 + \sum_{\substack{i\in\p, \\ \Zps=0}} \{\gamsc(r_{0}) - \bgamobsc(r_{0})\}^2 \right],
    \end{align*}
    where $ s_{\p(1,0)}^2(r_1,r_0) $ is the sample variance of observed $\gamst(r_{1})$'s and $\gamsc(r_{0})$'s in stratum $\p$: $\{\gamst(r_{1}): \s\in\p, \Zps = 1\} \cup \{\gamsc(r_{0}): \s\in\p, \Zps = 0\} $.
\end{remark}

\begin{proof}[Proof of Remark \ref{rmk:alterv}]
    Follow the notation of Proposition \ref{lem:ab2}.
    A known result of the sample variance is
    $$ \frac{1}{n-1}\sum_{i=1}^n (x_i-\bar{x}) = \frac{1}{\binom{n}{2}} \sum_{i=1}^n \sum_{j\neq i} \frac{1}{2} (x_i - x_j)^2, $$
    thus,
    $$ \binom{n_1+n_1}{2} s_{(1,0)}^2 = \sum_{i=1}^{n_1} \sum_{j \neq i} \frac{1}{2}(y_{i|1}-y_{j|1})^2 + \sum_{i=1}^{n_1} \sum_{j \neq i} \frac{1}{2}(y_{i|0}-y_{j|0})^2 + \sum_{i=1}^{n_1}\sum_{j=1}^{n_1}\frac{1}{2} (y_{i|1}-y_{j|0})^2, $$
    where $s_{(1,0)}^2$ is the sample variance of $\{y_{1|1},\ldots y_{n_1|1}, y_{1|0},\ldots, y_{n_1|0} \}$. Combined with the fact that
    $$\binom{n_1}{2} s_1^2 = \sum_{i=1}^{n_1} \sum_{j \neq i} \frac{1}{2}(y_{i|1}-y_{j|1})^2 \  \text{ and }  \ \binom{n_1}{2} s_0^2 = \sum_{i=1}^{n_1} \sum_{j \neq i} \frac{1}{2}(y_{i|0}-y_{j|0})^2, $$
    we have
    $$\frac{1}{2} \sum_{i=1}^{n_1}\sum_{j=1}^{n_1}(y_{i|1}-y_{j|0})^2 =\binom{n_1+n_1}{2} s_{(1,0)}^2 - \binom{n_1}{2} s_1^2 - \binom{n_1}{2} s_0^2. $$
    Therefore, the small strata variance estimator can be written as
    \begin{align*}
        \nu_{\p}^{(s)} (\rho_1,\rho_0) = &\ \frac{2}{\npc \npt} \left[ \binom{\np }{2} s_{\p(0,1)}^2 - \binom{\npc}{2}s_{\p 0}^2 - \binom{\npt}{2}s_{\p1}^2 \right]- \frac{\npc-1}{\npc} s_{\p 0}^2 - \frac{\npt-1}{\npt} s_{\p 1}^2, \\ 
        =&\  \frac{2}{\npc \npt}  \binom{\np }{2} s_{\p(0,1)}^2- \left( \frac{1}{\npc} + \frac{1}{\npt} \right)\left[ (\npc-1) s_{\p 0}^2 + (\npt-1) s_{\p 1}^2 \right].
    \end{align*}
\end{proof}

\subsection{Remark \ref{rmk:varbias}}
\label{sec:sm-remark-varbias}
\begin{remark}\label{rmk:varbias-full}
Under constant treatment effect for all units in all strata, $ \nu_{\p}^{(e_\p)} (\rho_1,\rho_0)$, $e_\p\in\{l,s\}$ are unbiased for $\var \left( \bgamobst(\rho_1) - \bgamobsc(\rho_0) \right)$. In general, when treatment effects mainly exhibit heterogeneity across strata, $\nu^{(l)}$ tends to have smaller bias than $\nu^{(s)}$.
Specifically, the bias of $ \nu_{\p}^{(l)} (r_1,r_0)$ as an estimator of $\var \left( \bgamobst(r_1) - \bgamobsc(r_0) \right)$ is 
\begin{align*}
&\frac{1}{\np-1} 
\var\{\gamst-\gamsc: \s\in \p\} \\ &= \frac{1}{(\np-1)} \sump \left[\np\wps \{(\ypst-\ypsc) - (r_1-r_0)\} - \sump\wps \{(\ypst-\ypsc) - (r_1-r_0) \} \right]^2,
\end{align*}
which increases with within-stratum heterogeneity of treatment effects.
The bias of $ \nu_{\p}^{(s)} (r_1,r_0)$ is $$(\bgampt-\bgampc)^2 = \left[ \sump\wps \{(\ypst-\ypsc) - (r_1-r_0) \}\right]^2,$$ which increases with heterogeneity of treatment effects across strata. %, since $r_1, r_0$ are likely estimated using all strata.
\end{remark}

Besides the results stated in Remark \ref{rmk:varbias-full}, in this proof we also show that with constant treatment effects and $\npc=\npt$, $\var[\nu_b^{(s)}(\rho_1,\rho_0)] = [2(\npt-1)]^2  \cdot \var[\nu_b^{(l)}(\rho_1,\rho_0)]. $
\begin{proof}[Proof of Remark \ref{rmk:varbias-full}]
    By Lemma \ref{lemS:ab}, the bias of $ \nu_{\p}^{(l)} (\rho_1,\rho_0)$ is
    $$\frac{1}{\np-1} \left[ {\np} \sump (\gamsc - \bgampc)^2 + {\np} \sump (\gamst - \bgampt)^2 - {2}{\np} \sump (\gamsc - \bgampc)(\gamst - \bgampt) \right], $$
    which can be simplified to 
    \begin{align*}
        &\ \frac{1}{\np-1} \left[ {\np} \sump \gamsc (\gamsc - \bgampc) + {\np} \sump \gamst (\gamst - \bgampt) - {2}{\np} \sump (\gamsc - \bgampc)(\gamst - \bgampt) \right] \\
        =&\ \frac{1}{\np-1} \left[ {\np} \sump (\gamsc - \gamst) (\gamsc - \bgampc) + {\np} \sump (\gamst - \gamsc) (\gamst - \bgampt) \right]\\
        =&\ \frac{\np}{\np-1} \left[ \sump (\gamsc - \gamst) \gamsc - \np \bgampc (\bgampc - \bgampt) + \sump (\gamst - \gamsc) \gamst - \np\bgampt (\bgampt - \bgampc) \right]\\
        =&\ \frac{\np}{\np-1} \left[ \sump (\gamst - \gamsc)^2 -\np (\bgampt - \bgampc)^2 \right]\\
        =&\ \frac{\np^2}{\np-1}\cdot \var\{\gamsc-\gamst: \s\in \p\}.
    \end{align*}
    Since $\gamst-\gamsc = \wps [\ypst-\rho_1 - (\ypsc-\rho_0)]$, we have
    $$ \bgampt-\bgampc = \np^{-1}\sump \wps[(\ypst-\ypsc) - (\rho_1-\rho_0) ] , $$
    thus
    \begin{align*}
    &\frac{\np^2}{\np-1} 
    \var\{\gamst-\gamsc: \s\in \p\} \\ &= \frac{\np}{\np-1} \sump \left\{ \wps [(\ypst-\ypsc) - (\rho_1-\rho_0)] - \np^{-1}\sump\wps [(\ypst-\ypsc) - (\rho_1-\rho_0) ] \right\}^2.
    \end{align*}
    By Lemma \ref{lem:ab2}, the bias of $ \nu_{\p}^{(s)} (\rho_1,\rho_0)$ is 
    $$\np^2(\bgampt-\bgampc)^2 = \left\{ \sump\wps [(\ypst-\ypsc) - (\rho_1-\rho_0) ]\right\}^2.$$
    When the treatment effect is constant across all units in the experiment, $\ypst - \ypsc = \rho_1-\rho_0$ for all $i=1,\ldots,n$, so the biases of $ \nu_{\p}^{(e_\p)} (\rho_1,\rho_0)$, $e_\p=l,s$ are both zero. 
    Next, we compare the variances of $\nu_{\p}^{(l)} (\rho_1,\rho_0)$ and $\nu_{\p}^{(s)} (\rho_1,\rho_0)$ under constant treatment effect.
    We have $\gamst(\rho_1)= \wps(\ypst - \rho_1) = \wps(\ypsc - \rho_0) = \gamsc(\rho_0)$. Then the estimators are simplified to
    $$\nu_b^{(s)}(\rho_1,\rho_0)= \np^2(\bgamobst - \bgamobsc)^2 = \np^2\left[\frac{\np}{\npc}(\bgamobst - \bgampt )\right]^2,  $$
    and
    \begin{align*}
        \nu_b^{(l)}(\rho_1,\rho_0) =&\ \np^2 \frac{1}{\npc} s_0^2 + \np^2 \frac{1}{\npt} s_1^2 \\
        =&\ \np^2 \frac{1}{\npc(\npc-1)}\sum_{\substack{i: i\in\p, \\ \zps=0}} \gamst^2 + \np^2 \frac{1}{\npt(\npt-1)}\sum_{\substack{i: i\in\p, \\ \zps=1}} \gamst^2  \\
        &\ +\frac{2\np^3 \npt \bgampt }{\npc^2(\npc-1)}\bgamobst - \np^2\frac{\npt^2(\npt-1)+\npc^2(\npc-1)}{\npc^2(\npc-1)(\npt-1)}\bgamobst^2 + \text{constant},
    \end{align*}
    where the constant term does not depend on the random treatment assignment.
    When $\npc=\npt$, $\nu_b^{(s)}(\rho_1,\rho_0)=4 \np^2 (\bgamobst(\rho_1) - \bgampt(\rho_1) )^2$ and $$\nu_b^{(l)}(\rho_1,\rho_0) = -\frac{2\np^2 }{\npt -1} \bgamobst(\rho_1) ^2 + \frac{4 \np^2 }{\npt - 1}\bgampt(\rho_1) \bgamobst(\rho_1) + \text{constant}. $$
    Thus, with constant treatment effects and $\npc=\npt$, $$\var[\nu_b^{(s)}(\rho_1,\rho_0)] = [2(\npt-1)]^2  \cdot \var[\nu_b^{(l)}(\rho_1,\rho_0)]. $$
\end{proof}

\subsection{Theorem \ref{prop:cons1}}
\subsubsection{A useful lemma and claims}

To prove the theorem, we first develop a weak law of large numbers for averages of independent but not necessarily identically distributed random functions following Chapter 9 of \cite{keener2010theoretical}. Let $\{X_{n,i}, 1\le i\le n\}$, $n=1,2,\ldots$ be a triangular array of independent random variables, $K$ be a compact set in $\mathbb{R}^p$, and $C(K)$ denote the set of continuous functions defined on $K$. For $n=1,2,\ldots$ and $i=1,\ldots, n$, define $W_{n,i}(t) = h_{n,i}(t, X_{n,i}),\ t\in K$, where $h_{n,i}(t,x)$ is a continuous function of $t$ for all $x$. Then $W_{n,1},W_{n,2},\ldots,W_{n,n}$ are independent random functions taking values in $C(K)$. For $\omega \in C(K)$, define the supremum norm of $\omega$ as $\|\omega\|_\infty = \sup_{t\in K} |\omega(t)|.$

\begin{lemma}
\label{lem:wllnrandfunc}
    Let $\{W_{n,i}(t) = h_{n,i}(t,X_{n,i}), 1\le i\le n\},\ n=1,2,\ldots$ be a triangular array of independent random functions in $C(K)$, with $\E\|W_{n,i}\|_\infty < \infty,\ \forall \ n,i$. Let %$\bar{h}_n(t,(X_1,\ldots,X_n)) = [h_{n,1}(t,X_{n,1}) + \cdots + h_{n,n}(t,X_{n,n})]/n$, 
    $\bar{W}_n = \sum_{i=1}^n W_{n,i}/n$ and $\mu_n = \E\bar{W}_n$. Suppose that for any $t\in K$, $\sum_{i=1}^n\var(W_{n,i}(t))/n^2\to 0 $ as $n\to \infty$, and there exists a constant $L>0$ such that for any $n$, $i$, and $x$, 
    $$|h_{n,i}(t,x) - h_{n,i}(s,x) | \le L \|t-s \|,\ \ \forall\  t,s\in K.$$ 
    %$$|\bar{h}_n(t,(x_1,\ldots,x_n)) - \bar{h}_n(s,(x_1,\ldots,x_n)) | \le L \|t-s \|,\ \ \forall t,s\in K.$$ 
    Then $$\|\bar{W}_n-\mu_n\|_\infty \xrightarrow{p}0\  \text{  as  }\  n\to\infty.$$
\end{lemma}

\begin{proof}[Proof of Lemma \ref{lem:wllnrandfunc}]
    Let $\bar{h}_n(t,(X_{n,1},\ldots,X_{n,n})) = \sum_{i=1}^n h_{n,i}(t,X_{n,i}) /n. $ Then for any $n\ge 1$ and $(x_1,\ldots,x_n)$,
    \begin{align*}
        |\bar{h}_n(t,(x_1,\ldots,x_n)) - \bar{h}_n(s,(x_1,\ldots,x_n)) | = & \left| \frac{1}{n} \sum_{i=1}^n  h_{n,i}(t,x_i) - \frac{1}{n} \sum_{i=1}^n  h_{n,i}(s,x_i) \right| \\
        \le &\  \frac{1}{n} \sum_{i=1}^n | h_{n,i}(t,x_i) - h_{n,i}(s,x_i) | \\
        \le &\ \frac{1}{n} \sum_{i=1}^n L \|t-s\| = L \|t-s\|,\ \ \forall\ t,s\in K.
    \end{align*}
    For any $\epsilon>0$, there exists $\delta>0$ such that for any $n\ge 1$ and $(x_1,\ldots,x_n)$,
    $$|h_{n,i}(t,x_i) - h_{n,i}(s,x_i) | < \epsilon \ \ \text{and}\ \ |\bar{h}_n(t,(x_1,\ldots,x_n)) - \bar{h}_n(s,(x_1,\ldots,x_n)) | < \epsilon ,$$
    for all $t,s\in K$ with $\|t-s\|< \delta $.
    By the definition of $W_{n,i}$ and $\bar{W}_n$, we have
    \begin{equation}
        \label{eq:sup-randfunc}
        \E|W_{n,i}(t)-W_{n,i}(s)|< \epsilon,\ \ \text{and}\ \  \E |\bar{W}_n(t) - \bar{W}_n(s)| < \epsilon,
    \end{equation}
    for all $t,s\in K$ with $\|t-s\|< \delta $. Thus, 
    $$ |\mu_n(t)-\mu_n(s)| = |\E [\bar{W}_n(t) - \bar{W}_n(s)]| \le \E |\bar{W}_n(t) - \bar{W}_n(s)| < \epsilon. $$
    Let $B_\delta(t) = \{ s: \|s-t\|< \delta\}$ denote the open ball with radius $\delta$ about $t$. Since $K$ is compact, the open sets $B_\delta(t),\  t\in K$, covering $K$ have a finite subcover $O_j=B_\delta(t_j)$, $j=1,\ldots,m$. Then
    \begin{align*}
        \| \bar{W}_n - \mu_n \|_\infty = &\ \max_{j=1,\ldots,m} \sup_{t\in O_j} |\bar{W}_n(t) - \mu_n(t)|\\
        \le & \  \max_{j=1,\ldots,m} \sup_{t\in O_j} \left[ |\bar{W}_n(t) - \bar{W}_n(t_j)| + |\bar{W}_n(t_j) - \mu_n(t_j)| + | \mu_n(t_j) -\mu_n(t)|   \right] \\
        < &\  \max_{j=1,\ldots,m} \sup_{t\in O_j}|\bar{W}_n(t) - \bar{W}_n(t_j)| + \max_{j=1,\ldots,m}|\bar{W}_n(t_j) - \mu_n(t_j)| + \epsilon.
    \end{align*}
    Now
    \begin{align*}
        \sup_{t\in O_j}|\bar{W}_n(t) - \bar{W}_n(t_j)| = \frac{1}{n} \sup_{t\in O_j} \left| \sum_{i=1}^n [W_{n,i}(t) - W_{n,i}(t_j)] \right| 
        \le \frac{1}{n} \sum_{i=1}^n\sup_{t\in O_j} \left| W_{n,i}(t) - W_{n,i}(t_j) \right|
    \end{align*}
    is an average of independent random variables. Let $\E \sup_{t\in O_j} \left| W_{n,i}(t) - W_{n,i}(t_j) \right| = \lambda_{n,i}^\delta(t_j) $, then by Chebyshev's inequality and independence,
    \begin{align*}
        &\P \left( \left| \frac{1}{n} \sum_{i=1}^n\sup_{t\in O_j} \left| W_{n,i}(t) - W_{n,i}(t_j) \right| - \frac{1}{n} \sum_{i=1}^n \lambda_{n,i}^\delta(t_j) \right|  > \epsilon \right) \\
        &\quad  \le \frac{1}{\epsilon^2 n^2} \sum_{i=1}^n \var\left(\sup_{t\in O_j} \left| W_{n,i}(t) - W_{n,i}(t_j) \right| \right)
        \to  0 , \ \ \text{as}\ n \to \infty,
    \end{align*}
    where the last step is because
    \begin{align*}
        &\var\left(\sup_{t\in O_j} \left| W_{n,i}(t) - W_{n,i}(t_j) \right| \right) \le \E\left[\left(\sup_{t\in O_j} \left| W_{n,i}(t) - W_{n,i}(t_j) \right| \right)^2 \right] \\
        &\le\E  \left[\sup_{t\in O_j} \left| W_{n,i}(t) - W_{n,i}(t_j) \right|^2 \right] = \E \left[\sup_{t\in O_j} \left| h_{n,i}(t, X_{n,i}) - h_{n,i}(t_j, X_{n,i}) \right|^2 \right] \\
        &\le  \E\left[ L^2 \|t-t_j\|^2 \right] \le L^2\delta^2 .
    \end{align*}
    Thus, $$\frac{1}{n} \sum_{i=1}^n\sup_{t\in O_j} \left| W_{n,i}(t) - W_{n,i}(t_j) \right| - \frac{1}{n} \sum_{i=1}^n \lambda_{n,i}^\delta(t_j) \xrightarrow{p} 0, $$ and $ \sum_{i=1}^n \lambda_{n,i}^\delta(t_j) /n < \epsilon $ follows from the first part of (\ref{eq:sup-randfunc}). Besides, for any $t\in K$,
    \begin{align*}
        \P \left( \left| \frac{1}{n} \sum_{i=1}^n  W_{n,i}(t)  - \mu_n(t) \right| > \epsilon \right) \le \frac{1}{\epsilon^2 n^2} \sum_{i=1}^n \var \left(  W_{n,i}(t)\right) \to  \ 0 , \ \ \text{as}\ n \to \infty.
    \end{align*}
    Hence, $\bar{W}_n(t) - \mu_n(t) \xrightarrow{p} 0 $ for any $t\in K$.  Therefore,
    \begin{align*}
        \| \bar{W}_n - \mu_n \|_\infty <&\  \max_{j=1,\ldots,m}\left[ \frac{1}{n} \sum_{i=1}^n\sup_{t\in O_j} \left| W_{n,i}(t) - W_{n,i}(t_j) \right| - \frac{1}{n} \sum_{i=1}^n \lambda_{n,i}^\delta(t_j)  \right]  + \epsilon \\
        & + \max_{j=1,\ldots,m}|\bar{W}_n(t_j) - \mu_n(t_j)| + \epsilon ,
    \end{align*}
    where the two maximums both converge to zero in probability. Thus,
    $$\P\left(\| \bar{W}_n - \mu_n \|_\infty  > 4\epsilon \right) \to 0 \ \ \text{as}\ n\to\infty. $$ 
\end{proof}

\noindent Now we state some useful claims without proof. 

\begin{claim}
    \label{clm:sumto0}
    Let $a_1,\ldots, a_m\ge0$ and $n_1,\ldots,n_m\ge1$. Denote $N = \sum_{i=1}^m n_i^2$. Then $$ \left[\sum_{i=1}^m \frac{n_i}{N} a_i\right]^2 \le \frac{1}{N}\left[\sum_{i=1}^m \frac{n_i^2}{N} a_i\right] \left( \sum_{i=1}^m a_i\right). $$ 
\end{claim}

\begin{claim}
    \label{clm:lip-const}
    Let $f(x) = x^T A x + B^Tx + c$ be a quadratic function on $\bar{B}_r(0)$, where $A$ is an $m\times m$ symmetric matrix, $B$ is an $m$-dimensional constant vector, $c$ is a constant, and $\bar{B}_r(0) $ is a compact ball of radius $r$ in $\mathbb{R}^m$. Then $f$ is Lipschitz continuous with constant 
    $$ \max_{x\in \bar{B}_r(0)} \|2 Ax + B\| \le 2r \cdot \max_{\|x\|\le1} \|Ax\| + \|B\| = 2r \cdot \max \{ |\lambda|: \lambda \text{ is an eigenvalue of } A \} + \|B\| .$$
\end{claim}

\begin{claim}
    \label{clm:eigenvalue}
    Suppose $$A_i = \left( \begin{array}{cc}
        a_i & b_i \\
        b_i & c_i
    \end{array} \right),\ \ i=1,2,\ldots $$ are a sequence of $2$ by $ 2$ symmetric matrices. If $$ M^2 \pm (a_i+c_i)M + (a_ic_i-b_i^2) \ge 0 \ \ \text{ and }\ \ \frac{a_i+c_i}{2} \in (-M,M) ,\ \ i=1,2,\ldots  $$
    for some constant $M>0$, then eigenvalues of $A_i$'s are all within the range of $[-M,M]$. Such $M$ exists if $a_i$, $b_i$, and $c_i$ have finite bounds that are uniform for all $i\ge1$.
\end{claim}

\subsubsection{The proof of Theorem \ref{prop:cons1}}

\begin{proof}[Proof of Theorem \ref{prop:cons1}]
Denote the limit of $(\rho_1,\rho_0) $ as $(\rho_{1,\infty},\rho_{0,\infty}). $ There exists a compact set $K\subseteq \mathbb{R}^2$ such that $(\rho_1,\rho_0) \in K $ for all $\B\ge1$. Thus, we can consider
\begin{equation}
    \left\{ \frac{B}{\E  \spb \np^2 \nu_\p^{(e_\p)} (\rho_1,\rho_0)} \np^2 \nu^{(e_\p)}_\p(r_1,r_0), \p=1,\ldots \B \right\},\ B=1,2,\ldots
    \label{eq:tri-randfunc}
\end{equation}
as a triangular array of random functions in $C(K)$ with
$$ \E \left\|\frac{B}{\E  \spb \np^2 \nu_\p^{(e_\p)} (\rho_1, \rho_0)} \np^2 \nu^{(e_\p)}_\p(r_1,r_0)\right\|_\infty< \infty, $$
because $ \nu^{(e_\p)}_\p(r_1,r_0)$ is a quadratic function of $(r_1,r_0)$ on a compact set for any realization of treatment assignments $\zps, i=1,\ldots,n$. 
To apply Lemma \ref{lem:wllnrandfunc}, we need to check its remaining conditions. 
For any $(r_1,r_0)\in K$, we need to show that
\begin{align*}
    & \frac{1}{\B^2} \spb \var\left( \frac{\B}{\E  \spb \np^2 \nu_\p^{(e_\p)} (\rho_1, \rho_0)} \np^2 \nu^{(e_\p)}_\p(r_1,r_0) \right) \\
    & = \frac{1}{\left[ \E  \spb \np^2\nu_\p^{(e_\p)} (\rho_1, \rho_0)\right]^2} \spb \var\left( \np^2 \nu^{(e_\p)}_\p(r_1,r_0) \right) \to 0 \ \ \text{as}\ \B\to\infty.
\end{align*}
To simplify notation, we write $\gamsz$ for $\wps(\ypsk - r_z)$
and $s_{\p z}^2$ for $s_{\p z}^2(r_z)$ in the remainder of this proof. First,
\begin{align*}
    \var \left(\nu^{(l)}_\p(r_1,r_0) \right) = \var\left(  \frac{1}{\npt} s_{\p1}^2 +\frac{1}{\npc} s_{\p0}^2  \right) \le   2\var\left(\frac{1}{\npt} s_{\p1}^2 \right) + 2 \var\left( \frac{1}{\npc} s_{\p0}^2\right),
\end{align*}
and for fixed $(r_1,r_0)$, $s_{\p z}^2 $ is the sample variance of a simple random sample $\{ \gamsz: i\in \p, \zps=z\}$. Since $\npk\ge2$, the stratum size is at least 4, so we can use claim \ref{clm:sample-var} to find an upper bound of the variance
\begin{align*}
    \var\left(  s_{\p z}^2\right) \le &\  \frac{\np(\np-\npk)}{\npk(\npk-1)(\np-1)(\np-2)(\np-3)} (\npk \np - \np - \npk - 1) \cdot \\
    &\ \frac{1}{\np}  \sump \left( \wps (\ypsk - r_z) - \frac{1}{\np} \sump \wps (\ypsk - r_z) \right)^4\\
    =&\ \frac{1}{\npk}  \frac{\np}{\np-3} \cdot \frac{\np-\npk}{\np-2} \cdot \frac{\npk \np - \np - \npk - 1}{(\npk-1)(\np-1)} \cdot \\
    &\ \frac{1}{\np} \sump \left( \wps (\ypsk - r_z) - \frac{1}{\np} \sump \wps (\ypsk - r_z) \right)^4.
\end{align*}
Thus,
\begin{equation}
    \var\left(\np^2 \nu^{(l)}_\p(r_1,r_0) \right) \le c_1 \sump \sum_{z=1,0} \left( \wps (\ypsk - r_z) - \frac{1}{\np} \sump \wps (\ypsk - r_z) \right)^4.
    \label{eq:varnul-upbd}
\end{equation}
for some constant $c_1$. Next, 
\begin{align*}
    &\ \var\left(\nu^{(s)}_\p(r_1,r_0) \right) \\
    = &\ \var\left( \frac{1}{\npt \npc} \sum_{\substack{\s\in\p\\ \zps=1}} \sum_{\substack{\s\neq\s'\in\p\\ z_{\s'}=0}} (\gamst-\gamssc)^2 - \frac{\npt-1}{\npt} s_{\p 1}^2 - \frac{\npc-1}{\npc} s_{\p 0}^2 \right) \\
    \le &\  2 \var\left(  \frac{1}{\npt \npc} \sum_{\substack{\s\in\p\\ \zps=1}} \sum_{\substack{\s\neq\s'\in\p\\ z_{\s'}=0}}(\gamst-\gamssc)^2 \right) + 2\var\left(\frac{\npt-1}{\npt} s_{\p 1}^2 \right) + 2\var\left(\frac{\npc-1}{\npc} s_{\p 0}^2\right).
\end{align*}
For any stratum, the probability of getting every plausible treatment assignment is equal, so 
\begin{align*}
    &\E \left[ \left(\frac{1}{\npt \npc} \sum_{\substack{\s\in\p\\ \zps=1}} \sum_{\substack{\s\neq\s'\in\p\\ z_{\s'}=0}}(\gamst-\gamssc)^2\right)^2 \right] \\
    =&\ \frac{1}{\npc^2 \npt^2} \biggr[  \sump \sum_{j\in \p} \P(z_i=1,z_j=0) (\gamst-\gamma_{j0})^4 \\
    &\ \qquad  + \sump \sum_{j,j'\in\p}\P(z_i=1,z_j=z_{j'}=0) (\gamst-\gamma_{j0})^2(\gamst-\gamma_{j'0})^2 \\
    &\ \qquad + \sum_{i,i'\in\p} \sum_{j\in\p}\P(z_i=z_{i'}=1,z_j=0) (\gamst-\gamma_{j0})^2(\gamma_{i'1}-\gamma_{j0})^2\\
    &\ \qquad + \sum_{i,i'\in\p} \sum_{j,j'\in\p}\P(z_i=z_{i'}=1,z_j=z_{j'}=0) (\gamst-\gamma_{j0})^2(\gamma_{i'1}-\gamma_{j'0})^2 \biggr] \\
    =&\  \frac{1}{\npc^2 \npt^2} \biggr[  \sump \sum_{j\in \p} \frac{\npt\npc}{\np(\np-1)} (\gamst-\gamma_{j0})^4 \\
    &\ \qquad  + \sump \sum_{j,j'\in\p}\frac{\npt\npc(\npt-1)}{\np(\np-1)(\np-2)} (\gamst-\gamma_{j0})^2(\gamst-\gamma_{j'0})^2 \\
    &\ \qquad + \sum_{i,i'\in\p} \sum_{j\in\p}\frac{\npt\npc(\npc-1)}{\np(\np-1)(\np-2)} (\gamst-\gamma_{j0})^2(\gamma_{i'1}-\gamma_{j0})^2\\
    &\ \qquad + \sum_{i,i'\in\p} \sum_{j,j'\in\p}  \frac{\npt\npc(\npt-1)(\npc-1)}{\np(\np-1)(\np-2)(\np-3)} (\gamst-\gamma_{j0})^2(\gamma_{i'1}-\gamma_{j'0})^2 \biggr] ,
\end{align*}
where if $\np=3$, the last term is zero, and if $\np=2$, the last three terms are zero. Besides,
\begin{align*}
    \E \left(\frac{1}{\npt \npc} \sum_{\substack{\s\in\p\\ \zps=1}} \sum_{\substack{\s\neq\s'\in\p\\ z_{\s'}=0}}(\gamst-\gamssc)^2  \right) =\frac{1}{\np(\np-1)} \sump \sum_{j\in \p} (\gamst-\gamma_{j0})^2  ,
\end{align*}
so by calculation, the variance is
\begin{align*}
    &\var\left(  \frac{1}{\npt \npc} \sum_{\substack{\s\in\p\\ \zps=1}} \sum_{\substack{\s\neq\s'\in\p\\ z_{\s'}=0}}(\gamst-\gamssc)^2 \right) \\
    = &\ \frac{\np(\np-1)-\npt\npc}{\np^2(\np-1)^2\npt\npc}  \sump \sum_{j\in \p} (\gamst-\gamma_{j0})^4 \\
    &+ \frac{\np(\np-1)(\npt-1) - (\np-2)\npt\npc}{\np^2(\np-1)^2(\np-2)\npt\npc }  \sump \sum_{j,j'\in\p}(\gamst-\gamma_{j0})^2(\gamst-\gamma_{j'0})^2 \\
    &+ \frac{\np(\np-1)(\npc-1) - (\np-2)\npt\npc}{\np^2(\np-1)^2(\np-2)\npt\npc } \sum_{i,i'\in\p} \sum_{j\in\p}(\gamst-\gamma_{j0})^2(\gamma_{i'1}-\gamma_{j0})^2\\
    &+ \frac{\np(\np-1)(\npc-1)(\npt-1) -(\np-2)(\np-3)\npt\npc }{\np^2(\np-1)^2(\np-2)(\np-3)\npt\npc} \sum_{i,i'\in\p} \sum_{j,j'\in\p}   (\gamst-\gamma_{j0})^2(\gamma_{i'1}-\gamma_{j'0})^2.
\end{align*}
%where the four terms are at most constant times of $\np^2$, $\np^3$, $\np^3$, and $\np^3$ respectively as the sample size grows. 
Thus, there exist constants $c_2$ and $c_2'$ such that
\begin{equation}
\begin{split}
    \var \left(\np^2  \nu^{(s)}_\p(r_1,r_0) \right) \le &\  c_2 \np  \sump \sumssp \left[ \wps(\ypst - r_1) - \wpst(y_{i'0}-r_0) \right]^4  \\
    %& + c_2' \np^2 \sump \sum_{z=0,1} \left( \wps (\ypsk - r_z) - \frac{1}{\np} \sump \wps (\ypsk - r_z) \right)^4. 
\end{split}
    \label{eq:varnus-upbd}
\end{equation}
Besides,
\begin{equation}
\begin{split}
    \E [ \np^2 \nu_\p^{(l)} (\rho_1, \rho_0) ]& = \sum_{z=0,1} \frac{1}{\npk} \np^2 \frac{1}{\np-1} \sump \left[ \wps (\ypsk - \rho_z) - \frac{1}{\np} \sum_{i'\in\p} \wpst (\ypstk - \rho_z) \right]^2 \\
    &\ge c_3  \sump \sum_{z=1,0} \left[ \wps (\ypsk - \rho_z) - \frac{1}{\np} \sum_{i'\in\p} \wpst (\ypstk - \rho_z) \right]^2
\end{split}
    \label{eq:Enul-lwbd}
\end{equation}
for some constant $c_3$ and
\begin{equation}
\begin{split}
    \E [\np^2\nu_\p^{(s)} (\rho_1, \rho_0)] &=  \frac{1}{\np(\np-1)} \sump \sumssp \left[ \np\wps(\ypst - \rho_1) - \np\wpst(y_{\s'0} - \rho_0) \right]^2 \\
    & \ge c_4   \sump \sumssp \left[ \wps(\ypst - \rho_1) -   \wpst(y_{\s'0} - \rho_0) \right]^2
\end{split}
    \label{eq:Enus-lwbd}
\end{equation}
for some constant $c_4$. 
Under the assumptions, when $e_\p=l$ for some strata, we have that as the sample size grows, by \eqref{eq:varnul-upbd} and \eqref{eq:Enul-lwbd},
\begin{align*}
    &\frac{\sum_{b:e_\p=l} \var\left( \np^2  \nu^{(l)}_\p(r_1,r_0) \right)}{ \left[\E  \sum_{b:e_\p=l} \np^2 \nu_\p^{(l)} (\rho_1, \rho_0)\right]^2} \\
    &\le  \frac{\sum_{b:e_\p=l}\np\bwp^4}{\left( \sum_{b:e_\p=l}\np\bwp^2\right)^2}\cdot \frac{\sum_{b:e_\p=l} c_1 \frac{1}{\sum_{b:e_\p=l}\np\bwp^4}  \sump\sum_{z=1,0} \left[ \wps (\ypsk - r_z) - \frac{1}{\np} \sum_{i'\in\p} \wpst (\ypstk - r_z) \right]^4}{ \left\{ \sum_{b:e_\p=l} c_3 \frac{1}{\sum_{b:e_\p=l}\np\bwp^2} \sump \left[ \wps (\ypsk - \rho_z) - \frac{1}{\np} \sum_{i'\in\p} \wpst (\ypstk - \rho_z) \right]^2 \right\}^2 } .
\end{align*}
When $e_\p=s$ for some strata, by \eqref{eq:varnus-upbd} and \eqref{eq:Enus-lwbd} we have
\begin{align*}
    &\frac{\sum_{b:e_\p=s} \var\left( \np^2  \nu^{(s)}_\p(r_1,r_0) \right)}{ \left[\E  \sum_{b:e_\p=s} \np^2 \nu_\p^{(s)} (\rho_1, \rho_0)\right]^2} \\
    &\le \frac{\sum_{b:e_\p=s}\np^2\bwp^4}{\left( \sum_{b:e_\p=s}\np^2\bwp^2\right)^2} \cdot  \frac{\sum_{b:e_\p=s}\np \cdot c_2 \frac{1}{\sum_{b:e_\p=s}\np^2\bwp^4} \sump \sumssp \left[ \wps(\ypst - r_1) -   \wpst(y_{\s'0} - r_0) \right]^4}{ \left\{ \sum_{b:e_\p=s} c_4 \frac{1}{\sum_{b:e_\p=s}\np^2\bwp^2} \sump \sumssp \left[ \wps(\ypst - \rho_1) -   \wpst(y_{\s'0} - \rho_0) \right]^2 \right\}^2 },
\end{align*}
which tends to zero following Claim \ref{clm:sumto0}. As $\B\to\infty$, we have ${\sum_{b:e_\p=l}\np\bwp^4} / {(\sum_{b:e_\p=l}\np\bwp^2)^2} $ and/or $ {\sum_{b:e_\p=s}\np^2\bwp^4} /{( \sum_{b:e_\p=s}\np^2\bwp^2)^2}$ tending to zero. 
Therefore, with bounded $\np$'s for all $\p$ with $e_\p=s$,
\begin{align*}
    \frac{1}{\left[ \E \spb \np^2 \nu_\p^{(e_\p)} (\rho_1, \rho_0)\right]^2} \spb \var\left( \np^2 \nu^{(e_\p)}_\p(r_1,r_0) \right) \to 0 \ \ \text{as}\ \B\to\infty.
\end{align*}

\noindent 
In addition, we need to show that there exists a constant $L>0$ such that 
$$ \frac{\B}{\E  \spb \np^2 \nu_\p^{(e_\p)} (\rho_1, \rho_0)}\left|\np^2 \nu^{(e_\p)}_\p(r_1,r_0) - \np^2 \nu^{(e_\p)}_\p(r_1',r_0')  \right| \le L \|(r_1,r_0) - (r_1',r_0') \|, $$ 
for any$ \ (r_1,r_0), (r_1',r_0') \in K$ and any $  1\le \p \le \B .$
Since $\nu^{(e_\p)}$ is a quadratic function of $(r_1,r_0) $, it suffices to show that coefficients of $r_1^2$, $r_0^2$, $r_1r_0$, $r_1$, and $r_0$ are within a uniform finite range. Denote the mean of unit weights in for units \textit{assigned to} treatment $z$ in stratum $\p$ as $\bar{w}_{\p |z} = \sum_{\s\in\p, \zps=z} \wps / \npk $ and the mean weights in stratum $\p$ as $\bar{w}_b=\sump \wps/\np$. By calculation, the coefficient of $r_z^2$ in $\nu^{(l)}_\p$ is
\begin{align*}
    \frac{1}{\npk} \frac{1}{\npk-1} \sumpindk \left( \wps - \bar{w}_{\p z}^{obs} \right)^2 \le \frac{1}{\npk} \frac{\np}{\npk-1} \frac{1}{\np} \sump (\wps - \bar{w}_\p)^2 , 
\end{align*}
which is bounded by a constant time of $\np^{-1}\bwp^2$. The coefficient of $r_z$ in $\nu^{(l)}_\p$ is
\begin{align*}
    &\ 2 \frac{1}{\npk} \frac{1}{\npk-1}  \sumpindk \left( \bar{w}_{\p |z} -\wps \right) \wps \ypsk \\
    \le&\ \frac{1}{\npk} \frac{1}{\npk-1}  \sumpindk \left[ \left( \bar{w}_{\p |z} -\wps \right)^2 +  \left( \wps \ypsk  - \frac{1}{\npk} \sum_{\substack{i'\in\p \\ z_{\s'}=z}}  \wpst \ypstk  \right)^2 \right]\\
    \le&\  \frac{1}{\npk} \frac{\np}{\npk-1} \left[ \frac{1}{\np} \sump(\wps - \bar{w}_\p)^2 + \frac{1}{\np} \sump \left(\wps\ypsk - \frac{1}{\np} \sum_{i'\in\p}  \wpst \ypstk \right)^2 \right],
\end{align*}
which is bounded by a constant time of $\np^{-1}\bwp^2$.
The coefficient of $r_z^2$ in $\nu^{(s)}_\p$ is
$$  %\left\{ \begin{array}{ll}
        %\frac{1}{\ppk^2} \sump \indk \wps^2, & \text{if } \npk=1 \\
        \frac{1}{\npk} \sumpindk \left[ \wps^2 -  \left( \wps - \bar{w}_{\p z}^{obs} \right)^2\right] = \left(\bar{w}_{\p z}^{obs}\right)^2, % & \text{if } \npk > 1 
    %\end{array} \right.
$$
and the coefficient of $r_1r_0$ in $\nu^{(s)}_\p$ is
$$ \frac{1}{\npc\npt} \sum_{\substack{\s\in\p\\ \zps=0}} \sum_{\substack{\s\neq\s'\in\p\\ z_{\s'}=1}} \wps \wpst. $$
Those coefficients are bounded by constant times of $\bwp^2$. 
The coefficient of $r_0$ in $\nu^{(s)}_\p$ is
\begin{align*}
    &\ \frac{1}{\npc\npt} 2 \sum_{\substack{\s\in\p\\ \zps=0}} \sum_{\substack{\s\neq\s'\in\p\\ z_{\s'}=1}} w_{i'}\left( \wps \ypst - w_{i'}y_{i'0} \right)  - \frac{1}{\npc^2}  2 \sum_{\substack{\s\in\p\\ \zps=0}} \left( \bar{w}_{\p 0|0} -\wps \right) \wps \ypsc \\
    =&\  2 \bar{w}_{\p|0} \left( \frac{1}{\npt} \sum_{\substack{\s\in\p\\ \zps=1}} \wps \ypst - \frac{1}{\npc} \sum_{\substack{\s\in\p\\ \zps=0}} \wps \ypsc \right),
\end{align*} 
which is bounded by a constant time of $ \bwp^2$ under the assumptions. Similarly, the coefficient of $r_1$ in $\nu^{(s)}_\p$ is also bounded by $ \bwp^2$. By \eqref{eq:Enul-lwbd} and \eqref{eq:Enus-lwbd}, $$\E  \spb \np^2 \nu_\p^{(e_\p)} (\rho_1, \rho_0) \ge  \sum_{\p:e_\p=l} c_3 \np \bwp^2 + \sum_{\p:e_\p=s} c_4 \np^2 \bwp^2  . $$
The assumptions suggest that $B\np\bwp^2/\spb\np\bwp^2$ is bounded for all $\p$ and $B\np^2\bwp^2/\spb\np\bwp^2$ is bounded for all $\p$ with $e_\p=s$, therefore, there exists a constant $L$ that is independent of the sample size such that for any $(r_1,r_0), (r_1', r_0') \in K$,
$$\frac{\B }{\E \spb\np^2  \nu_\p^{(e_\p)} (\rho_1, \rho_0)}\left|\np^2 \nu^{(e_\p)}_\p (r_1,r_0) - \np^2 \nu^{(e_\p)}_\p(r_1', r_0')  \right| \le L \|(r_1,r_0) - (r_1', r_0') \|. $$
Applying Lemma \ref{lem:wllnrandfunc} to random functions (\ref{eq:tri-randfunc}), we get
$$ \sup_{(r_1,r_0)\in K} \left| \frac{\spb\np^2 \nu^{(e_\p)}_\p(r_1,r_0)}{\E \spb\np^2 \nu^{(e_\p)}_\p(\rho_1, \rho_0) } - \frac{\E\spb\np^2 \nu^{(e_\p)}_\p(r_1,r_0)}{\E \spb\np^2 \nu^{(e_\p)}_\p(\rho_1, \rho_0) }  \right| \xrightarrow{p} 0. $$
Denote 
$$ \mu(r_1,r_0) = \E\left[\frac{1}{W}\spb\np^2 \nu^{(e_\p)}_\p (r_1,r_0) \right], $$
then $\mu$ is a continuous function on $K$. The remainder of the proof follows from the arguments in the proof of Theorem \ref{prop-cons2} after defining $\mu(r_1,r_0)$. Letting $\xi_\B = \E \left[ \spb\np^2 \nu^{(e_\p)}_\p  (\rho_1, \rho_0) /W \right] $ satisfies the theorem. Under constant treatment effect, Remark \ref{rmk:varbias}'s conclusion on the bias of $\nu_{\p}^{(e_\p)}(\rho_1, \rho_0), e_\p=l,s$ suggests that $\liminf_{\B\to\infty}[\xi_\B/\var( \sqrt{W}\tilde\tau) ] = 1$. 

\end{proof}

\iffalse
\begin{proof}[Proof of Proposition \ref{prop:cons1}]
%Since $(\hat\rho_0, \hat\rho_1)^T - (\rho_0, \rho_1)^T \xrightarrow{p}0 $ follows from Theorem \ref{thm1} and functions $\spb \nu_\p^{(i)}(\cdot,\cdot)/W^2,\ i=1,2$ are uniformly continuous, we have $$\frac{1}{W^2}  \spb \nu_\p^{(i_\p)}(\hat\rho_1,\hat\rho_0) - \frac{1}{W^2} \spb \nu_\p^{(i_\p)}(\rho_1, \rho_0) \xrightarrow{p} 0.$$
By derivations of $\nu_{\p}^{(i)}(\rho_0, \rho_1) $, $i=1,2$, we have
$$\E \left[ \frac{1}{\facB} \spb \nu_{\p}^{(i_\p)}(\rho_0, \rho_1) \right] \ge \frac{1}{\facB} \spb \var \left(\bgamobst - \bgamobsc \right) , $$
abbreviating $\bgamobsz(\rho_{z})$ to $\bgamobsz$; thus
\begin{align*}
 \liminf_{\B\to\infty} \E \left[ \frac{1}{\facB} \spb \nu_{\p}^{(i_\p)}(\rho_0, \rho_1) \right] \ge &\  \lim_{\B\to\infty} \frac{1}{\facB} \spb \var \left(\bgamobst - \bgamobsc \right) \\
=&\  \lim_{\B\to\infty}\frac{1}{\facB} \spb \var\left( \sump ({\vst - \rho_1\ust}) -  \sump({\vsc - \rho_0\usc} ) \right)  \\
=&\  \lim_{\B\to\infty} \left( {\sigma_1^2} + {\sigma_0^2}- 2{\sigma_{01}}\right)
=\lim_{\B\to\infty} \var\left(\sqrt{\facB}\tilde\tau \right).
\end{align*}
\end{proof}
\fi

\subsection{Theorem \ref{thm2}}

\subsubsection{Lemma \ref{lem:wlln2}}

\begin{lemma}
\label{lem:wlln2}
Under Conditions \ref{cond:wlln2-pi} and \ref{cond:wlln2-w}, $-\nabla \Phi_n \xrightarrow{p}  I_2$ as $\np\to\infty$ for all $\p$, where $I_2$ is a $2\times 2$ identity matrix.
\end{lemma}
This lemma is similar to Lemma \ref{lem:wlln}. They imply that $-\nabla \Phi_n$ tends to the same limit $I_2$ in probability, but their underlying asymptotic mechanisms are distinct. Lemma \ref{lem:wlln} requires that the number of strata grows infinitely, while Lemma \ref{lem:wlln2} requires that the number of strata is fixed as their sizes grow.

\begin{proof}[Proof of Lemma \ref{lem:wlln2}]
    For unit $\s$ and treatment $z$, define $\vsk =  \Indk\ypsk/\psk $ and $\usk =  {\Indk}/{\psk} $.
    Then 
    $$ \psi_i(r_1, r_0) = \left( \begin{array}{c}
     \{\Indt/\pst\} (\ypstt - r_1)  \\
     \{\Indc/\psc\} (\ypsc - r_0) 
\end{array} \right) = \left( \begin{array}{c}
    \vst - r_1 \ust  \\
    \vsc - r_0 \usc   
\end{array} \right) $$
    The gradient of $\Phi_n$ with respect to $\theta=(r_1,r_0)$ is
    $$ \nabla \Phi_n = \left(\begin{array}{ccc}
            -\ssn\wps\ust/\facB &  0\\
             0& -\ssn\wps\usc/\facB 
            \end{array} \right).$$
    For $z=1,0$ and $b=1,\ldots,\B$, 
    $$ \E\left[\sump\wps\usk\right] = \E \left[\sump \frac{\Indk}{\ppk}\wps \right]  =  \sump \wps $$
    and
    \begin{equation*}
        \begin{split}
            \var\left( \sump\wps\usk\right) = & \  \E \left( \sump\frac{\Indk}{\ppk^2}\wps^2 + \sump\sumssp \frac{\Indk \Indsk }{\ppk^2}\wps\wpst  \right) - \left(\sump \wps\right)^2 \\
            =&\ \sump \frac{1}{\ppk}\wps^2 + \sump \sumssp \frac{\E \Indk \Indsk}{\ppk^2} \wps\wpst  - \left(\sump \wps\right)^2 \\
            =&\  \sump \left(\frac{1}{\ppk} - 1\right) \wps^2 + \sump \sumssp \left(\frac{\delbkk}{\ppk^2}-1\right) \wps\wpst 
            .%\\ =&\ \sump \left(\frac{1}{\ppk} - 1\right) \wps^2 + \sump \sumssp \frac{\npk - \np}{\npk(\np-1)} \wps\wpst.
        \end{split}
    \end{equation*}
    Under Condition \ref{cond:wlln2-pi}, $$\frac{\delbkk}{\ppk^2}-1 = \frac{\npk - \np}{\npk(\np-1)} = O({\np}^{-1}).$$
    Note that under Condition \ref{cond:wlln2-w}, $\spb \sump \wps^2 / \facB^2 \to 0 $ and by Cauchy-Schwartz inequality, $$\frac{1}{\facB^2}\spb \sumssp\frac{1}{\np} \wps \wpst \le \frac{1}{\facB^2}\spb \frac{1}{\np} \left(\sump\wps \right)^2 \le \frac{1}{\facB^2}\spb \left(\sump \wps^2 \right) \to 0. $$ 
    Thus,
    $$  \E\left[\frac{1}{\facB}\ssn\wps \usk\right] = \frac{1}{\facB}\sum_{b=1}^B  \sump \wps = 1 
    $$
    and the independence of treatment assignment between blocks implies that
    \begin{align*} 
    \var \left(\frac{1}{\facB}\ssn\wps \usk\right) %=\frac{1}{\facB^2} \spb \var\left( \sump\wps\usk\right) \\
    =  \frac{1}{\facB^2} \sum_{b=1}^B  \left[\sump \left(\frac{1}{\ppk} - 1\right) \wps^2 + \sump \sumssp \left(\frac{\delbkk}{\ppk^2}-1\right)  \wps\wpst \right]\to 0.
    \end{align*}
    By Chebyshev's inequality, for any $ \epsilon>0$,
    $$ \P\left(\left|\frac{1}{\facB}\ssn\wps\usk - 1 \right|>\epsilon \right) \le \frac{\var\left(\frac{1}{\facB}\ssn\wps \usk\right)}{ \epsilon^2} \to 0. 
    % = \frac{\spb \var\left( \sump\wps \usk\right)}{\facB^2\epsilon^2}
    $$
    Therefore, $\ssn\wps\usk/W \xrightarrow{p} 1$ and
    $$ -\nabla G_n = \left(\begin{array}{ccc}
            \ssn\wps\ust/\facB & 0\\
            0 & \ssn\wps\usc/\facB 
            \end{array} \right) \xrightarrow{p} I_2.
    $$
\end{proof}

\subsubsection{Lemma \ref{lem:fpclt}}

\begin{lemma}
\label{lem:fpclt}
Under Conditions \ref{cond:clt-var}--\ref{cond:clt-square}, as $\np\to\infty$, $$ \frac{1}{\bwp}\left( \begin{array}{cc}
    \left( \frac{1}{\ppt}-1\right) S_{\p,1}^2 & - S_{\p,01} \\
    - S_{\p,01} & \left( \frac{1}{\ppc}-1\right) S_{\p,0}^2
\end{array} \right) = \tSigma_{n,\p} \to \tSigma_{\infty,\p} ,$$
and 
$$ \frac{1}{\sqrt{\np\bwp}} \left( \begin{array}{c}
    \sump \left(\frac{\Indt}{\ppt}- 1 \right) \wps (\ypst- \rho_1) \\
    \sump \left(\frac{\Indc}{\ppc} - 1 \right) \wps (\ypsc - \rho_0)  \\  
\end{array} \right)\ \simd \  \mathcal{N}(0,\tSigma_{n,\p}). $$
\end{lemma}

Before proving Lemma \ref{lem:fpclt}, we state a lemma that follows from Theorem 5 of \cite{Li2017}. First, we describe the settings. In a completely randomized experiment with $N$ units and $Q$ treatments, let $\pi_q$ be the probability of assignment to treatment $q$ and $n_q$ be the observed number of units treated by treatment $q$. For any unit $i$ and treatment $q$, let $Y_i(q)$ be the $p$-dimensional potential outcome vector and let $\bar{Y}(q)=\sum_{i=1}^N Y_i(q)/N$ be the mean of potential outcomes. We consider an estimand of the form $\tau_i(A) = \sum_{q=1}^Q A_q Y_i(q)$, where $A_q$ is a $k$ by $p$ constant matrix. The population average causal estimand is $\tau(A) = \sum_{i=1}^N \tau_i(A)/N = \sum_{q=1}^Q A_q \bar{Y}(q)$. Then an estimator of $\tau(A) $ is $\hat\tau(A) = \sum_{q=1}^Q A_q \hat{\bar{Y}}(q), $ where $\hat{\bar{Y}}(q)$ is the mean of observed outcomes under treatment $q$. For $1\le q \neq r \le Q$, let the finite population covariances of the potential outcomes be $$S_q^2 = \frac{1}{N-1} \sum_{i=1}^N [Y_i(q) - \bar{Y}(q)] [Y_i(q) - \bar{Y}(q)]^T,\  \ \ S_{qr}= \frac{1}{N-1} \sum_{i=1}^N [Y_i(q) - \bar{Y}(q)] [Y_i(r) - \bar{Y}(r)]^T .$$
%If attrition occurs, let $N$, $\bar{Y}(q)$, $\tau(A)$, $\hat\tau(A)$, $S_q^2$, and $S_{qr}$ denote the values excluding the units that dropped out. Suppose that the number of units that dropped out divided by $\sqrt{N}$ tends to zero.  

\begin{lemma}
\label{lem-extfpclt}
    If for any $1\le q \neq r \le Q$, $S_q^2$ and $S_{qr}$ have limiting values, $\pi_q$ has a positive limiting value, and $\max_{1\le q\le Q}\max_{1\le i\le N}\|Y_i(q)-\bar{Y}(q)\|^2_2/N\to0$, then 
    $$ \sqrt{N}(\hat\tau(A) - \tau(A))\xrightarrow{d} \mathcal{N}(0, V)\ \ as \ \ N\to\infty, $$
    where $V$ is the limit of $$ \sum_{q=1}^Q \left(\frac{N}{n_q} - 1\right) A_q S^2_q A_q^T - \sum_{q=1}^Q \sum_{r\neq q} A_q S_{qr} A_r^T .$$
\end{lemma}
\begin{proof}
    First, we show that the covariance matrix of $\hat\tau(A)$ in Theorem 3 of \cite{Li2017} is equal to the finite version of $V$:
    \begin{align*}
        S^2_{\tau(A)} =&\  \frac{1}{N-1} \sum_i (\tau_i(A)-\tau(A))(\tau_i(A)-\tau(A))^T \\
        =&\  \frac{1}{N-1} \sum_i \left[ \sum_{q=1}^Q A_q(Y_i(q)-\bar{Y}(q)) \right]\left[ \sum_{q=1}^QA_q(Y_i(q)-\bar{Y}(q)) \right]^T \\
        =&\ \frac{1}{N-1} \sum_i\Bigg[ \sum_{q=1}^Q A_q (Y_i(q)-\bar{Y}(q)) (Y_i(q)-\bar{Y}(q))^T \\
        & \qquad \qquad \qquad + \sum_{q=1}^Q \sum_{r\neq q}A_q(Y_i(q)-\bar{Y}(q))(Y_i(r)-\bar{Y}(r))^T A_r^T\Bigg]\\
        =&\ \sum_{q=1}^Q A_q S^2_q A_q^T + \sum_{q=1}^Q \sum_{r\neq q}A_q S_{qr} A_r^T,
    \end{align*}
    thus,
    $$N\cov(\hat\tau(A)) = \sum_{q=1}^Q \frac{N}{n_q}A_q S^2_q A_q^T - \frac{N}{N}S^2_{\tau(A)} = \sum_{q=1}^Q \left( \frac{N}{n_q}-1\right) A_q S^2_q A_q^T - \sum_{q=1}^Q \sum_{r\neq q}A_q S_{qr} A_r^T .$$
    According to Theorem 5 of \cite{Li2017}, this matrix has a limit $V$ and $ \sqrt{N}(\hat\tau(A) - \tau(A))\xrightarrow{d} \mathcal{N}(0, V)$ as $N$ tends to infinity.

\end{proof}

\iffalse
\begin{lemma}
\label{lem-randfunc}
    Let $\{X_n(t,\omega)\}_n$ be a sequence of random functions and $\{Y_n\}_n$ and $\{\tilde{Y}_n\}_n$ be sequences of random variables. Let $Z$ be a random variable. If $Y_n-\tilde{Y}_n\xrightarrow{p}0$, $X_n(\tilde{Y}_n,\cdot) \xrightarrow{d} Z$, and $X_n(\cdot,\omega)$ is a uniformly continuous function for any $\omega$, then $X_n(Y_n,\cdot) \xrightarrow{d}Z$.
\end{lemma}
\begin{proof}
For any $m\ge1$ and $\omega$, we have
$$X_m(Y_n,\omega) - X_m(\tilde{Y}_n,\omega) \xrightarrow{p} 0. $$
Then for any continuous point $z$ of $Z$, as $n\to\infty,$
$$P(X_m(Y_n,\cdot) \le z) = P(\{\omega: X_m(Y_n,\omega) \le z\}) \to P(\{\omega: X_m(Y,\omega) \le z\}) = P(X_m(Y,\cdot) \le z) . $$
Then let $m\to\infty$, $P(X_m(Y,\cdot) \le z) \to P(Z\le z)$. Thus, $X_n(Y_n,\cdot) \xrightarrow{d}Z$.
\end{proof}
\fi

\begin{proof}[Proof of Lemma \ref{lem:fpclt}]
For unit $\s$ and treatment $z$, we defined $\vsk =  \Indk\ypsk/\psk, $ and $\usk =  {\Indk}/{\psk} $. Consider $\gamsz(\cdot)$ evaluated at $\rho_z$ as a transformed potential outcome of unit $\s$ under treatment $z$. The mean of $\gamsz(\cdot)$ evaluated at $\rho_z$ for units assigned to treatment $z$ in stratum $\p$ is 
$$\bgamobsz= \frac{1}{\npk} \sumpindk \gamsz(\rho_z) =  \sump\frac{\Indk }{\npk }\wps (\ypsk - \rho_z) = \frac{1}{\np} \sump(\vsk-\rho_z\usk),$$
which can be considered as an estimate of the finite population mean evaluated at $\rho_z$:
$$ \bgampz(\rho_z) = \frac{1}{\np}\sump \gamsz(\rho_z) =  \frac{1}{\np} \sump \wps (\ypsk - \rho_z). $$
Define vectors $\delta_0 = (1,0)^T$ and $\delta_1 = (0,1)^T$, so the $(z+1)$th entry of $\delta_z$ is 1. Then, we can define an individual-effect-type of estimand (induced by $\gamsz(\rho_z)$ and $\delta_z$) for unit $\s$ in stratum $\p$ as $\tau_{\p\s} = \sum_{z=0}^1 \delta_z  \np\gamsz(\rho_z)$. Accordingly, the finite population average effect is defined as 
$$
\tau_\p=  \sum_{z=0}^1 \delta_z  \np\bgampz(\rho_z) =  \left( \begin{array}{c}
    \sump \wps (\ypst - \rho_1) \\
    \sump \wps (\ypsc - \rho_0)   \\
\end{array} \right),
$$ 
and its estimator is defined as
$$\hat\tau_\p = \sum_{z=0}^1 \delta_z \bgamobsz(\rho_z) =  \left( \begin{array}{c}
    \sump \frac{\Indt}{\ppt}\wps (\ypst - \rho_1) \\
    \sump \frac{\Indc}{\ppc}\wps (\ypsc - \rho_0)   \\
\end{array} \right)= \left( \begin{array}{c}
    \sump \wps (\vst - \rho_1 \ust ) \\
    \sump \wps (\vsc - \rho_0 \usc )  \\
\end{array} \right).
$$
\iffalse
Let the finite population covariance of the individual causal effects in block $\p$ be
\begin{align*}
    S^2_{\tau_\p } =&\  \frac{1}{\np-1} \sump (\tau_\s - \tau_\p)(\tau_\s - \tau_\p)^T \\
    =&\  \frac{1}{\np-1} \sump \left[ \left( \begin{array}{c}
        \gamsc  \\
         \gamst 
    \end{array}\right) - \left( \begin{array}{c}
        \bgampc  \\
         \bgampt 
    \end{array}\right) \right] \cdot \left[ \left( \begin{array}{c}
        \gamsc  \\
         \gamst 
    \end{array}\right) - \left( \begin{array}{c}
        \bgampc  \\
         \bgampt 
    \end{array}\right) \right]^T 
    = \left( \begin{array}{cc}
        S_{\p,0}^2  & S_{\p,01}  \\
        S_{\p,01} & S_{\p,1}^2
    \end{array}\right).
\end{align*}
\fi
Now we check if the conditions of Lemma \ref{lem-extfpclt} are satisfied. By Condition \ref{cond:clt-var}, $S_{\p,0}^2/ \bwp $, $S_{\p,1}^2/ \bwp $, and $S_{\p,01}/ \bwp$ have limiting values. By Condition \ref{cond:clt-prop}, the assignment probabilities $\ppk$ have positive limiting values. Note that
$$ \max_{z=1,0} \max_{i\in b} \frac{1}{\np\bwp} \| \gamsz - \bgampz \|^2_2  = \max_{z=1,0} \max_{i\in b} \frac{1}{\np\bwp}  \left\|  \wps (\ypsk-\rho_z) - \frac{1}{\np  } \sump \wps(\ypsk - \rho_z) \right\|^2_2 \to 0 $$
holds under Condition \ref{cond:clt-square}.
Therefore, we can apply the lemma to $\hat\tau_\p$ for all $\B$ strata. 
For a given stratum $\p$, there exists a matrix $\tSigma_{\infty,\p}$ such that
\begin{align*}
    &\frac{1}{ \bwp} \left[\sum_{z=0}^1 \left( \frac{1}{\ppk}-1\right) \delta_z S_{\p,z}^2 \delta_z^T -2 \delta_0 S_{\p,01}\delta_{1}^T\right] \\
    &= \frac{1}{ \bwp}\left( \begin{array}{cc}
    \left( \frac{1}{\ppt}-1\right) S_{\p,1}^2 & - S_{\p,01} \\
    - S_{\p,01} & \left( \frac{1}{\ppc}-1\right) S_{\p,0}^2
\end{array} \right)  = \tSigma_{n,\p} \to \tSigma_{\infty,\p} 
\end{align*}
and
$$ \frac{1}{\sqrt{\np\bwp}} \left(\hat\tau_\p - \tau_\p \right) \simd \mathcal{N}(0,\tSigma_{n,\p}). $$

\end{proof}

\subsubsection{The proof of Theorem \ref{thm2}}
\begin{proof}[Proof of Theorem \ref{thm2}]
We combine the results of Lemma \ref{lem:fpclt} for each stratum to derive the asymptotic distribution of $\hat\tau$. Using the same notation as in the proof of Lemma \ref{lem:fpclt}, by Condition \ref{cond:clt-prop}, 
$$ \frac{1}{ \sqrt{W}} \spb {\sqrt{\np\bwp}} \cdot \frac{1}{\sqrt{\np\bwp}}   (\hat\tau_\p - \tau_\p ) \ \simd\  \mathcal{N}\left(0, \ \spb \frac{\np\bwp}{ W} \tSigma_{n,\p}\right). $$
Note that 
$$ \frac{1}{ \sqrt{W}} \spb \hat\tau_\p = \frac{1}{ \sqrt{W}} \spb \left( \begin{array}{c}
    \sump \wps (\vst - \rho_1 \ust) \\
    \sump \wps (\vsc - \rho_0 \usc)   \\
\end{array} \right), $$
which equals $\sqrt{W} \Phi_n(\theta_n)$.
By the definition of $\rho_1$ and $\rho_0$ we have
$$\frac{1}{\sqrt{W}} \spb  \sump \wps(\ypsk - \rho_z) = 0,\ \ z=1,0.$$
Thus, 
$$\frac{1}{ \sqrt{W}} \spb \tau_\p = \frac{1}{ \sqrt{W}} \spb \left( \begin{array}{c}
    \sump \wps (\ypst - \rho_1) \\
    \sump \wps (\ypsc - \rho_0)   \\
\end{array} \right) = 0, $$
and 
\begin{equation}
    \sqrt{W} \Phi_n(\theta_n) = \sqrt{W} \frac{1}{W}\ssn \wps \left( \begin{array}{c}
    \vst - \rho_1 \ust \\
    \vsc - \rho_0 \usc   \\
\end{array} \right) \ \simd \ \mathcal{N}\left( 0,  \tSigma_n \right),
    \label{eq:proof-aympnormGn}
\end{equation}
where 
$$ \tSigma_{n} =\spb  \frac{\np\bwp}{ W} \tSigma_{n,\p}= \frac{\np}{W}\left( \begin{array}{cc}
    \left( \frac{1}{\ppt}-1\right) S_{\p,1}^2 & - S_{\p,01} \\
    - S_{\p,01} & \left( \frac{1}{\ppc}-1\right) S_{\p,0}^2
\end{array} \right). $$
By Slutsky's Theorem combined with Lemma \ref{lem:wlln2} and equation \eqref{eq:proof-aympnormGn} we have
\begin{align*}
    \Sigma_{n}^{-\frac{1}{2}}\cdot I_2 \cdot
    \sqrt{\facB} (\hat\theta_n - \theta_n)
    &= \  \Sigma_{n}^{-\frac{1}{2}}\cdot I_2 \cdot \left(-\nabla \Phi_n\right)^{-1} \cdot  \Sigma_{n}^{\frac{1}{2}} \cdot  \Sigma_{n}^{-\frac{1}{2}} \cdot \sqrt{\facB}\Phi_n(\theta_n)  \xrightarrow{d}  \ \mathcal{N}(0, I_2).
\end{align*}
Thus, $\sqrt{W}(\hat\theta_n - \theta_n)  \ \simd\  \mathcal{N}(0, \tSigma_n )  .$
Since $\hat\tau - \tau = (\hat\rho_1-\rho_1) - (\hat\rho_0-\rho_0)$, we have $(\tsigma_0^2 - 2\tsigma_{01} + \tsigma_1^2)^{-1/2} (\hat\tau - \tau) \xrightarrow{d} \mathcal{N}(0,1). $
\end{proof}

\subsection{Theorem \ref{prop-cons2}}

\subsubsection{Useful claims}
We first state some useful claims before proving Theorem \ref{prop-cons2}.
Claim \ref{clm:sample-var} follows from \cite{Finucan1974} and we omit the proof.

\begin{claim}
    \label{clm:sample-var}
    Consider a simple random sample $\{X_1,\ldots,X_n\}$ drawn without replacement from a finite population $\{x_1,\ldots,x_N\}$. Denote the centered moment statistics of the sample as $$ M_r = \frac{1}{n} \sum_{i=1}^n (X_i - \bar{X}_n)^r, \ r=1,2,\ldots ,$$
    where $\bar{X}_n = \sum_{i=1}^n X_i/n $. Denote the centered moments of the finite population as $$m_r =\frac{1}{N} \sum_{i=1}^N (x_i - \bar{x}_N)^r, \ r=1,2,\ldots,  $$
    where $\bar{x}_N = \sum_{i=1}^N x_i/N $.
    Then $$ \E M_2 = \frac{n-1}{n} \cdot \frac{N}{N-1} m_2 $$ and when $N>3$,
    \begin{align*}
        \var(M_2) = &\  \frac{(n-1)N (N-n)}{n^3 (N-1)^2 (N-2) (N-3)} \cdot \\
        & \ \left[ (nN-N-n-1)(N-1)m_4 - (nN^2-3N^2 + 6N - 3n-3) m_2^2 \right] .
    \end{align*}
\end{claim}

\begin{claim}
    \label{clm:sample-cov}
    Let $\{(X_1,Y_1),\ldots,(X_n, Y_n)\} $ be a simple random sample drawn without replacement from a finite population $\{(x_1,y_1),\ldots,(x_N,y_N)\}$ with $N\ge 3,\  1\le n\le N$. Denote the sample covariance as
    $$\sigma_{XY} = \frac{1}{n-1}\sum_{i=1}^n (X_i-\bar{X}_n)(Y_i-\bar{Y}_n), $$
    where $\bar{X}_n = \sum_{i=1}^n X_i / n$ and $\bar{Y}_n = \sum_{i=1}^n Y_i / n$ are sample means. Then 
    \begin{align*}
        \var(\sigma_{XY}) =&\  \frac{N(N-n)}{(N-1)^2(N-2)(N-3) n(n-1)} \cdot \\
        &\ \big[(nN-n-N-1)(N-1) \mu_{x^2y^2}+ (N-n-1)(N-1)\mu_{x^2}\mu_{y^2}  \\ &\ - (nN^2-nN-2n-2N^2+4N-2)\mu_{xy}^2   \big],
    \end{align*}
    where the moments are defined as 
    $$\mu_{x^2y^2} = \frac{1}{N}\sum_{i=1}^N (x_i-\bar{x})^2(y_i-\bar{y})^2, \ \  \mu_{xy} =\frac{1}{N}\sum_{i=1}^N (x_i-\bar{x})(y_i-\bar{y}), $$
    $$\mu_{x^2} = \frac{1}{N}\sum_{i=1}^N (x_i-\bar{x})^2,\  \text{ and } \  \mu_{y^2} = \frac{1}{N}\sum_{i=1}^N (y_i-\bar{y})^2.$$
\end{claim}
\begin{proof}[Proof of Claim \ref{clm:sample-cov}]
    The proof utilizes the method described by \cite{Finucan1974}. Suppose we are concerned with the variance of $$c_n = n\sum_{i=1}^n (X_i-\bar{X}_n)(Y_i-\bar{Y}_n) = n(n-1)\sigma_{XY}.$$
    Without loss of generality, assume $\bar{x}=\bar{y}=0$. Here are some useful identities:
    $$c_1 = 0 , \quad c_2 = (X_1-X_2)(Y_1-Y_2), \quad c_N = N^2\sigma_{xy},$$ and
    $$c_{N-1} = (N-1)\sum_{i=1}^{N-1}\left(X_i+\frac{X'}{N-1}\right) \left(Y_i+\frac{Y'}{N-1}\right) = (N-1)N\sigma_{xy} - NX'Y', $$
    where $\sigma_{xy} = \sum_{i=1}^n x_iy_i/n$ and $(X', Y')$ denotes the complement of the sample of size $N-1$. By expanding $c_n^2$ and $c_n$ as separate sums over individuals, pairs, triplets, and quartets and using symmetry in the expectations, it follows that the variance of $c_n$ is of the form $an+bn^2+cn^3+dn^4$. See \cite{Finucan1974} for more details. Then by $\var(c_1)=\var(c_N)=0,$ it simplifies to $n(n-1)(N-n)(A+Bn)$. Now we utilize the cases $n=N-1$ and $n=2$ to determine $A$ and $B$. When $n=N-1$,
    $$ \var(c_{N-1}) = N^2 \var(X'Y') = N^2\left( \mu_{x^2y^2} - \mu_{xy}^2\right), $$
    which gives us 
    \begin{equation}
        (N-1)(N-2)(A+BN-B) = N^2\left( \mu_{x^2y^2} - \mu_{xy}^2\right).
        \label{eq:AB1}
    \end{equation}
    When $n=2$,
    \begin{align*}
        \var(c_2) =&\  \E[(X_1-X_2)^2(Y_1-Y_2)^2] -\{\E[(X_1-X_2)(Y_1-Y_2)]\}^2\\
        =&\  \frac{1}{2\binom{N}{2}} \sum_{i,j} (x_i-x_j)^2(y_i-y_j)^2  - \left\{ \frac{1}{2\binom{N}{2}} \sum_{i,j}(x_i-x_j)(y_i-y_j)  \right\}^2,
    \end{align*}
    where by $\bar{x}=\bar{y}=0$,
    \begin{align*}
        \sum_{i,j} (x_i-x_j)^2(y_i-y_j)^2 = &\ 2 \sum_{i,j} x_i^2y_i^2 + 4 \sum_{i,j} x_i x_j y_i y_j + 2 \sum_{i,j} x_i^2 y_j^2 \\
        =&\  2N^2 \mu_{x^2y^2} + 4 N^2 \mu_{xy}^2 + 2 N^2 \mu_{x^2}\mu_{y^2}
    \end{align*}
    and 
    $$ \sum_{i,j} (x_i-x_j)(y_i-y_j) = 2N \sum_{i,j}x_i y_i = 2N^2 \mu_{xy}. $$
    Then
    $$\var(c_2) = \frac{2N}{N-1} \mu_{x^2y^2}  + \frac{2N}{N-1}\mu_{x^2}\mu_{y^2} - \frac{4N}{(N-1)^2} \mu_{xy}^2 ,$$
    which gives us
    \begin{equation}
    2(N-1)(N-2)(A+2B) = 2N \mu_{x^2y^2}  + 2N\mu_{x^2}\mu_{y^2} - \frac{4N}{(N-1)} \mu_{xy}^2 .   \label{eq:AB2}
    \end{equation}
    Solving $A$ from equation (\ref{eq:AB1}) and substitute into equation (\ref{eq:AB2}), we get
    $$B = \frac{N}{(N-1)^2(N-2)(N-3)}\left[(N-1)^2\mu_{x^2y^2} - (N-1)\mu_{x^2}\mu_{y^2}   - (N+1)(N-2)\mu_{xy}^2 \right].$$
    Then 
    $$A = \frac{N}{(N-1)(N-2)(N-3)}\left[-(N+1) \mu_{x^2y^2} + (N-1)\mu_{x^2}\mu_{y^2}   +2(N-1) \mu_{xy}^2 \right]. $$
    Finally, after some simple calculation we can derive the variance of $\sigma_{XY}$.
\end{proof}

\begin{claim}
    \label{clm:functcoef}
    Let $A_n$, $B_n$, and $C_n$, $n\ge 1$ be random variables and $a_n$, $b_n$, and $c_n$, $n\ge 1$ be constants. Suppose that $A_n-a_n\xrightarrow{p}0$, $B_n-b_n\xrightarrow{p}0$, and $C_n-c_n\xrightarrow{p}0$ as $n\to\infty$. Let $f_n(x) = A_n x^2 + B_n x + C_n$, $n\ge1$ be a sequence of random functions of $x\in K$, where $K$ is a compact set in $ \mathbb{R}$. Let $g_n(x) = a_n x^2 + b_n x + c_n$, $n\ge 1$ be a sequence of quadratic functions of $x\in K$. Then $$\sup_{x\in K} |f_n(x) - g_n(x)| \xrightarrow{p}0.$$
\end{claim}
\iffalse
\begin{proof}[Proof of Claim \ref{clm:functcoef}]
    Since $K$ is compact, suppose that $x^2\le M$ for any $x\in K$. For any $\epsilon>0$, 
    \begin{align*}
        &\ \P\left( \sup_{x\in K} |f_n(x) - g_n(x)| < \epsilon \right) = \P\left( \sup_{x\in K} |(A_n-a_n)x^2 + (B_n-b_n)x + C_n - c_n| > \epsilon \right)\\
        \le & \  \P\left(|A_n-a_n| > \frac{\epsilon}{3M} \right) + \P\left(|B_n-b_n| > \frac{\epsilon}{3\sqrt{M}} \right) + \P\left(|C_n-c_n| > \frac{\epsilon}{3} \right) \to 0.
    \end{align*}
\end{proof}
\fi

\subsubsection{The proof of Theorem \ref{prop-cons2}}

\begin{proof}[Proof of Theorem \ref{prop-cons2}]
    Denote the limit of $(\rho_1, \rho_0) $ as $(\rho_{1,\infty},\rho_{0,\infty}). $ There exists a compact set $K\subseteq \mathbb{R}^2$ such that $(\rho_1, \rho_0) \in K $ for all $\B\ge1$. We first show that for any $1\le\p\le\B$,
\begin{equation}
    \left\| \frac{\np^2 \nu_\p(r_1,r_0)}{\E \left[\spb\np^2 \nu_\p(\rho_1, \rho_0) \right]} - \frac{\E[\np^2 \nu_\p(r_1,r_0)]}{\E \left[\spb\np^2 \nu_\p(\rho_1, \rho_0) \right]} \right\|_\infty \xrightarrow{p} 0,\ \text{ as } \np\to\infty.
    \label{eq:tri-randfunc2}
\end{equation}
By calculation, for $z=0,1$, the coefficient of $r_z^2$ in $\nu_\p(r_1,r_0)$ is 
\begin{align*}
    \frac{1}{\npk} \frac{1}{\npk-1} \sumpindk \left( \wps - \bar{w}_{\p z}^{obs} \right)^2,
\end{align*}
and by Claim \ref{clm:sample-var}, it has an expected value of
\begin{align*}
    \frac{1}{\npk}  \frac{1}{\np-1} \sump (\wps - \bar{w}_\p)^2
\end{align*}
and its variance has an upper bound of
\begin{align*}
    &\ \frac{1}{\npk^2} \frac{\np(\np-\npk)}{\npk (\npk-1)(\np-1) (\np-2) (\np-3)} (\npk\np - \np - \npk -1) \frac{1}{\np} \sump  (\wps - \bar{w}_\p)^4 \\
    \le &\  \frac{1}{\np\npk^3}  \sump  (\wps - \bar{w}_\p)^4 \le c_1 \frac{1}{\np^{3}} \bwp^4
\end{align*}
for a constant $c_1$. The expectation
\begin{equation}
\begin{split}
    \E [\np^2\nu_\p (\rho_1, \rho_0)]& = \sum_{z=0,1} \frac{1}{\npk} \np^2 \frac{1}{\np-1} \sump \left[ \wps (\ypsk - \rho_z) - \frac{1}{\np} \sum_{i'\in\p} \wpst (\ypstk - \rho_z) \right]^2 \\
    &\ge c_2  \sump \left[ \wps (\ypsk - \rho_z) - \frac{1}{\np} \sum_{i'\in\p} \wpst (\ypstk - \rho_z) \right]^2
\end{split}
\label{eq:Enu-upbd}
\end{equation}
for some constant $c_2$. Hence, for any $\epsilon>0$,
\begin{align*}
    & \P\left( \frac{1}{\E \left[\spb \np^2 \nu_\p(\rho_1, \rho_0) \right]}\frac{\np^2}{\npk} \left| \frac{1}{\npk-1} \sumpindk \left( \wps - \bar{w}_{\p z}^{obs} \right)^2 -  \frac{1}{\np-1} \sump (\wps - \bar{w}_\p)^2 \right| > \epsilon \right) \\
    & \le  \frac{1}{\epsilon^2}\frac{c_1 \np \bwp^4}{c_2 \left(\spb \np \bwp^2 \right)^2}  \to 0  \quad \text{as} \quad n\to\infty.
\end{align*}
Thus, the coefficient of $r_z^2$ in $\np^2 \nu_\p(r_1,r_0) /\E [\spb\np^2 \nu_\p(\rho_1, \rho_0) ] $ subtracted by its expected value tends to zero in probability.
Similarly, the constant term in $\np^2 \nu_\p(r_1,r_0) /\E [\spb\np^2 \nu_\p(\rho_1, \rho_0) ] $ minus its expectation tends to zero in probability:
\begin{align*}
    &\frac{1}{\E \left[\spb\np^2 \nu_\p(\rho_1, \rho_0) \right]}\frac{\np^2}{\npk} \frac{1}{\npk-1} \sumpindk \left( \wps \ypsk - \frac{1}{\npk} \sum_{\substack{i'\in\p\\z_{\s'}=z}}  \wpst \ypstk \right)^2 \\
    &- \frac{1}{\E \left[\spb\np^2 \nu_\p(\rho_1, \rho_0) \right]}\frac{\np^2}{\npk} \frac{1}{\np-1} \sump  \left( \wps \ypsk - \frac{1}{\np} \sum_{i'\in\p}  \wpst \ypstk \right)^2 \quad \xrightarrow{p} \quad 0.
\end{align*}
The coefficient of $r_z$ in $\nu_\p(r_1,r_0)$ is
$$\frac{1}{\npk} \frac{1}{\npk-1}  \sumpindk \left( \bar{w}_{\p |z} -\wps \right) \left(\wps \ypsk - \frac{1}{\npk} \sum_{\substack{i'\in\p\\z_{\s'}=z}} \wpst \ypstk  \right), $$
and its variance is bounded from above by a constant time of
\begin{align*}
    \frac{1}{\npk^3}\left[ \frac{1}{\np} \sump  (\wps - \bar{w}_\p)^4 + \frac{1}{\np} \sump  \left( \wps \ypsk - \frac{1}{\np} \sum_{i'\in\p}  \wpst \ypstk \right)^4 \right] \le c_3 \frac{1}{\np^3} \bwp^4
\end{align*}
for a constant $c_3$ according to Claim \ref{clm:sample-cov}. Similarly, we have
\begin{align*}
    &\frac{1}{\E \left[\spb\np^2 \nu_\p(\rho_1, \rho_0) \right]}\frac{\np^2}{\npk} \frac{1}{\npk-1}  \sumpindk \left( \bar{w}_{\p |z} -\wps \right) \left(\wps \ypsk - \frac{1}{\npk} \sum_{\substack{i'\in\p\\z_{\s'}=z}} \wpst \ypstk  \right) \\
    &-\frac{1}{\E \left[\spb\np^2 \nu_\p(\rho_1, \rho_0) \right]}  \frac{\np^2}{\npk} \frac{1}{\np-1} \left[ \sump( \bar{w}_\p - \wps) \left(\wps\ypsk - \frac{1}{\np} \sum_{i'\in\p}  \wpst \ypstk \right) \right] \quad \xrightarrow{p} \quad 0 .
\end{align*}
Now, $ \left\{\np^2\nu_\p(r_1,r_0) - \E[\np^2 \nu_\p(r_1,r_0)]\right\}/\E \left[\spb\np^2 \nu_\p(\rho_1, \rho_0) \right] $ is a quadratic function of $(r_1,r_0)$ without a cross term $r_0r_1$ defined on a compact set with coefficients tending to zero in probability, thus, by Claim \ref{clm:functcoef} $$ \sup_{(r_1,r_0)\in K} \left| \frac{\np^2 \nu_\p(r_1,r_0) - \E[\np^2 \nu_\p(r_1,r_0)]}{\E \left[\spb\np^2 \nu_\p(\rho_1, \rho_0) \right]}\right| \xrightarrow{p} 0.  $$
Since $\B$ is fixed as $n$ increases, combining results for all stratum gives us
$$ \sup_{(r_1,r_0)\in K} \left| \frac{\spb\np^2  \nu_\p(r_1,r_0)}{\E \left[\spb\np^2 \nu_\p(\rho_1, \rho_0) \right]} - \frac{\E \left[\spb \np^2 \nu_\p(r_1,r_0) \right]}{\E \left[\spb\np^2 \nu_\p(\rho_1, \rho_0) \right]}   \right| \xrightarrow{p} 0. $$
Denote 
$$ \mu(r_1,r_0) = \E\left[\frac{1}{W}\spb \np^2 \nu_\p (r_1,r_0) \right], $$
then $\mu$ is a continuous function on $K$.
Since $(\rho_1, \rho_0) $ converges to $(\rho_{1,\infty},\rho_{0,\infty}) $, we have that $\mu(\rho_1, \rho_0) $ converges to $ \mu (\rho_{1,\infty},\rho_{0,\infty})$ and following Theorem \ref{thm1}, $ (\hat\rho_1,\hat\rho_0) \xrightarrow{p}(\rho_{1,\infty},\rho_{0,\infty})$. Then by continuous mapping theorem, $\mu(\hat\rho_1,\hat\rho_0) \xrightarrow{p} \mu(\rho_{1,\infty},\rho_{0,\infty}) . $
Thus,
$$ \left| \mu(\hat\rho_1,\hat\rho_0) - \mu(\rho_1, \rho_0)  \right| \le \left| \mu(\hat\rho_1,\hat\rho_0) - \mu(\rho_{1,\infty},\rho_{0,\infty})  \right| + \left| \mu(\rho_{1,\infty},\rho_{0,\infty}) - \mu(\rho_1, \rho_0)  \right| \xrightarrow{p} 0.$$
By equation \eqref{eq:Enu-upbd}, there exists a constant $c$ such that $\mu(\rho_1, \rho_0) \ge c\cdot \spb \np\bwp^2 / W$, so $ \mu(\hat\rho_1,\hat\rho_0)/ \mu(\rho_1, \rho_0) \xrightarrow{p}1. $ Now for any $\epsilon>0$, 
\begin{align*}
    &\ \P\left(\left| \frac{\spb \np^2 \nu_\p(\hat\rho_1,\hat\rho_0)}{\E \spb \np^2 \nu_\p(\rho_1, \rho_0) } - 1  \right|> \epsilon \right) \\
    \le &\ \P \left(\left| \frac{\spb \np^2 \nu_\p(\hat\rho_1,\hat\rho_0)}{\E \spb \np^2 \nu_\p(\rho_1, \rho_0) } - \frac{\mu(\hat\rho_1,\hat\rho_0)}{\mu(\rho_1, \rho_0)}  \right|> \frac{\epsilon}{2} \right) + \P \left(\left| \frac{\mu(\hat\rho_1,\hat\rho_0)}{\mu(\rho_1, \rho_0)} - 1  \right|> \frac{\epsilon}{2} \right) \\
    \le &\ \P \left(\sup_{(r_1,r_0)\in K} \left|\frac{\spb \np^2 \nu_\p(r_1,r_0)}{\E \spb \np^2 \nu_\p(\rho_1, \rho_0) } - \frac{\mu(r_1,r_0)}{\mu(\rho_1, \rho_0)} \right|> \frac{\epsilon}{2} \right) + \P\left( (\hat\rho_1,\hat\rho_0) \notin K \right) \\
    +&\  \P \left(\left| \frac{\mu(\hat\rho_1,\hat\rho_0)}{\mu(\rho_1, \rho_0)} - 1 \right|> \frac{\epsilon}{2} \right) \\
    \to &\ 0, \ \ \text{as} \ \np\to\infty,\ \text{for all } \p.
\end{align*}
Here, there exists $\delta>0$ such that $\P\left( (\hat\rho_1,\hat\rho_0) \notin K \right) \le\P\left( \|(\hat\rho_1,\hat\rho_0) - (\rho_{1,\infty},\rho_{0,\infty}) \| > \delta \right) \to 0 .$
Therefore, 
$$ \frac{\spb \np^2 \nu_\p(\hat\rho_1,\hat\rho_0)}{\E \spb \np^2 \nu_\p(\rho_1, \rho_0) }  \xrightarrow{p} 1. $$
When the treatment effect is constant for all units, according to the derivation of $\nu_{\p} $ and Remark \ref{rmk:varbias}'s conclusion on the bias of $\nu_{\p}=\nu_{\p}^{(l)}  $ (which we prove below), we have $$ \E \left[ \frac{1}{W} \spb \nu_\p  (\rho_1, \rho_0)  \right] = \frac{1}{W} \spb \np^2 \var\left(\bgamobst(\rho_1)-\bgamobsc(\rho_0) \right). $$
More generally, by the derivation of $\nu_{\p} $,
\begin{equation}
    \E \left[ \frac{1}{W} \spb \np^2 \nu_\p  (\rho_1, \rho_0)  \right] \ge \frac{1}{W} \spb \np^2 \var\left(\bgamobst(\rho_1)-\bgamobsc(\rho_0) \right)
    \label{eq:proof-prop2}
\end{equation}
and as $\np\to\infty$ for all $\p$,
$$ \frac{1}{W} \spb \np^2 \var\left(\bgamobst(\rho_1)-\bgamobsc(\rho_0) \right) = \var(\sqrt{W}\tilde\tau), $$
the variance of $\sqrt{W}\tilde\tau$. Letting $\xi_n = \E \left[ \spb \np^2 \nu_\p  (\rho_1, \rho_0) /W \right] $ satisfies the theorem. When the treatment effect is constant for all units, meaning that $\ypst-\ypsc = \tau$ for all $\s$, the variance estimator \eqref{eq:case2varest} is unbiased for $ (\sigma_0^2 - 2\sigma_{01} + \sigma_1^2)/W$, so equality holds in \eqref{eq:proof-prop2} and $\liminf_{\np\to\infty, \forall b}[\xi_n / \var(\sqrt{W}\tilde\tau) ] = 1$.

\iffalse
as a sequence of sets of random functions in $C(K)$ with
$$ \E \left\|\frac{B}{W} \nu_\p(r_1,r_0)\right\|_\infty< \infty, $$
because $ \nu_\p(r_1,r_0)$ is a quadratic function of $(r_1,r_0)$ on a compact set for any realization of treatment assignments. 
To apply Lemma \ref{lem:wllnrandfunc2}, we check the conditions. 
Note that the variance estimator $\nu_\p$ is essentially the same with $\nu^{(1)}_\p$ in the many-strata case. According to the proof of Proposition \ref{prop:cons1}, for any $(r_1,r_0)\in K$, there exists a constant $c_1$ such that $\var \left( \nu_\p(r_0, r_1) \right) \le c_1 \np ,$
so $$ \var \left( \frac{B}{W}\nu_\p(r_0, r_1) \right) \le \frac{B^2}{W^2}c_1  \np  \to 0 \ \text{ as } n\to \infty.$$
We have also shown that the coefficients of $r_z^2$ and $r_z$ in $\nu_\p$ are bounded by constant times of $\np$. Since $B\np/W$ is bounded as $\np$ grows, by Claim \ref{clm:eigenvalue} there exists $L>0$ such that 
$$ \frac{\B}{W}\left|\nu_\p (r_1,r_0) - \nu_\p(r_1', r_0')  \right| \le L \|(r_1,r_0) - (r_1', r_0') \|, \  \forall \ (r_1,r_0), (r_1', r_0') \in K.  $$
Now we check that for any $b=1,\ldots,B$, as $n\to\infty$,
$$\sup_{(r_1,r_0)\in K} \left|\frac{\B}{W} \nu_\p (r_1,r_0) - \E \left[ \frac{\B}{W} \nu_\p (r_1,r_0) \right]\right| \xrightarrow{p} 0.$$
\fi
\end{proof}

\section{Additional simulation details and results}

\subsection{Comparing SATE estimators under size-effect correlation}

\subsubsection{Detailed setup}

In this section we detail the setup of the simulation study described in Section \ref{sec:ATEsimu}. Let the sample be 10 or 50 strata of 2 clusters, or 2 strata of 10 or 50 clusters. Generate cluster weights from an i.i.d.\ Gamma distribution with shape $4$ and rate $4/30$ (referencing the simulation studies of \cite{Tong2022accounting}). %and the gamma distribution also fits with the empirical distribution of program sizes in the example introduction in Section \ref{sec:mov}. 
We consider two settings with regard to stratification: clusters are closely matched on size or clusters are not matched on size in strata. In the former setting, cluster weights are ordered and then assigned to strata. Potential outcomes are generated from the following model (with reference to the data generating process of \cite{pashley2021insights}):
$$\left(\begin{array}{c}
        \ypsc  \\
        y_{\s1}\\
    \end{array}\right) \sim MV \mathcal{N} \left( \left( \begin{array}{c}
        \Phi^{-1}(1-\frac{i}{1+n})  \\
        \Phi^{-1}(1-\frac{i}{1+n}) + t_i
    \end{array} \right),
    \left(\begin{array}{ccc}
        1 & 0  \\
        0 & 1  
    \end{array}\right) \right),\ i=1,\ldots,n,$$
where $\Phi^{-1}$ is the quantile function of standard normal distribution, and $t_i$ is the expected cluster treatment effect. We set $t_i=5$ to consider a constant treatment effect and $t_i=w_i/6$ (where $\E (w_i/6)=5$) to consider a heterogeneous effect that is correlated with cluster size. Thus, the simulation employs a $2\times2$ factorial design with a factor of cluster size variation within strata and a factor of treatment effect correlation with cluster size.

After cluster weights and potential outcomes are generated, they are kept fixed. In each stratum, half of the clusters are randomly selected to receive the treatment. Then the four estimators are calculated based on the treatment assignment and observed potential outcomes. As described in Section \ref{sec:mov}, IKN takes the stratumwise differences in weighted means of outcomes and then aggregates them based on stratum weights:
$$  \spb \frac{\sump \wps }{\ssn \wps} \cdot \left[ \frac{\sump \Indt \wps \ypst}{\sump\Indt \wps} - \frac{\sump \Indc \wps \ypsc}{\sump\Indc \wps}  \right]. $$  The estimator derived from the fixed effect model is the estimated treatment coefficient in a linear model with $\B-1$ stratum indicators and an intercept. Repeat the randomization and estimation process for 10,000 times.

\subsubsection{More results}

Figure \ref{fig:hajek-abcd} shows empirical distributions of the four SATE estimators under different designs with different sample sizes. Table \ref{tab:hajek} presents the biases, standard deviations (SD), and root mean squared errors (rMSEs) of these estimators in different cases. IKN and FE have larger biases than the other two when cluster sizes are not matched in strata and treatment effects are correlated with sizes, which is consistent with our findings in Section \ref{sec:inconsistency}.

\begin{figure}[!ht]
    \centering
    \includegraphics[width=\textwidth]{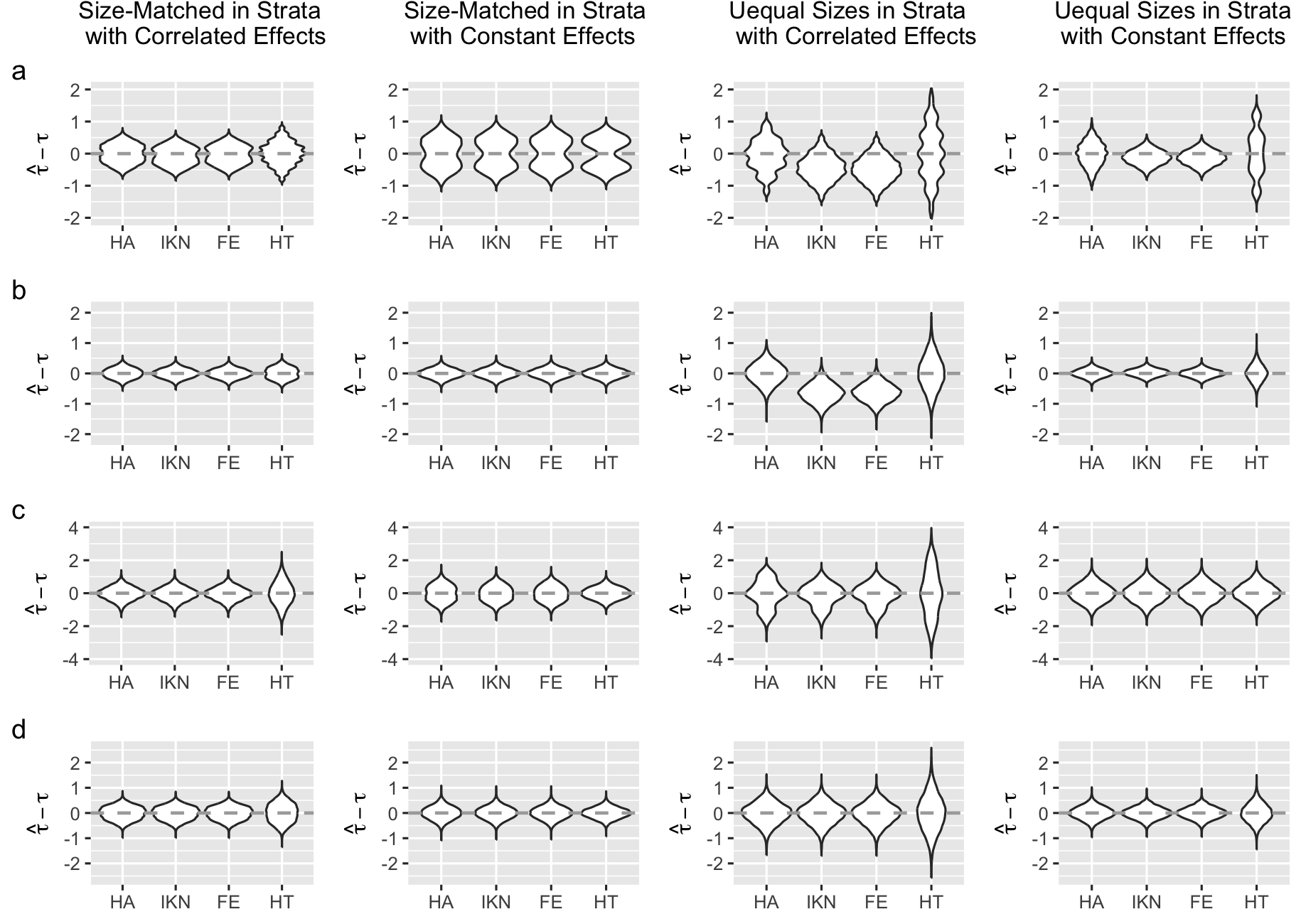}
    \caption{Empirical distributions of the H\'ajek effect estimator (HA), the IKN effect estimator (IKN), the fixed effects estimator (FE), and the Horvitz-Thompson estimator (HT) centered at the SATE. Each column represents a design, described at the top. Each row represents a sample size: a, 10 strata of 2 clusters; b, 50 strata of 2 clusters; c, 2 strata of 10 clusters; d, 2 strata of 50 clusters.}
    \label{fig:hajek-abcd}
\end{figure}

\begin{table}[!ht]
\caption{Biases, standard deviations (SDs), and root mean squared
  errors (rMSEs) of the H\'ajek effect estimator (HA), the
  \citeauthor{imai2009essential} estimator (IKN), the fixed effects estimator (FE), and the Horvitz-Thompson estimator (HT) under different designs and sample sizes.}
\centering
\resizebox{\columnwidth}{!}{
\begin{tabular}{ccccccccccccccc}
  \hline
   &&& \multicolumn{3}{c}{Size-matched in Strata}  &\multicolumn{3}{c}{Size-matched in Strata} & \multicolumn{3}{c}{Unequal Sizes in Strata} & \multicolumn{3}{c}{Unequal Sizes in Strata}\\ 
   &&& \multicolumn{3}{c}{with Correlated Effects}  &\multicolumn{3}{c}{with Constant Effects} & \multicolumn{3}{c}{with Correlated Effects} & \multicolumn{3}{c}{with Constant Effects}\\
$\np$ & $B$ & Estimator & Bias & SD & rMSE & Bias & SD & rMSE & Bias & SD & rMSE & Bias & SD & rMSE  \\ 
  \hline
  2 & 10 & HA & -0.0048 & 0.29 & \underline{0.30} & 0.0027 & 0.45 & \underline{0.45} & -0.0630 & 0.48 & \underline{0.54} & -0.0185 & 0.38 & \underline{0.40} \\
  2 & 10 & IKN & -0.0695 & 0.29 & 0.39 & 0.0069 & 0.46 & 0.46 & -0.4485 & 0.40 & 0.78 & -0.1185 & 0.23 & 0.41 \\ 
  2 & 10 & FE & -0.0332 & 0.29 & 0.34 & 0.0073 & 0.46 & 0.46 & -0.4858 & 0.39 & 0.80 & -0.1312 & 0.23 & 0.43 \\ 
  2 & 10 & HT & 0.0023 & 0.35 & 0.35 & 0.0005 & 0.46 & 0.46 & -0.0060 & 0.81 & 0.82 & -0.0002 & 0.71 & 0.71 \\ \hline
  2 & 50 & HA & 0.0036 & 0.17 & 0.18 & -0.0006 & 0.16 & \underline{0.16} & -0.0115 & 0.34 & \underline{0.36} & 0.0030 & 0.14 & \underline{0.15} \\ 
  2 & 50 & IKN & -0.0008 & 0.15 & \underline{0.16} & 0.0037 & 0.16 & 0.17 & -0.6465 & 0.32 & 0.87 & 0.0247 & 0.13 & 0.20 \\ 
  2 & 50 & FE & -0.0042 & 0.15 & 0.17 & 0.0036 & 0.16 & 0.17 & -0.6252 & 0.30 & 0.85 & -0.0049 & 0.13 & \underline{0.15} \\ 
  2 & 50 & HT & 0.0039 & 0.19 & 0.20 & -0.0011 & 0.16 & \underline{0.16} & 0.0058 & 0.55 & 0.56 & 0.0057 & 0.28 & 0.29 \\ \hline
  10 & 2 & HA & -0.0270 & 0.39 & \underline{0.42} & -0.0192 & 0.51 & 0.53 & -0.0723 & 0.83 & 0.87 & 0.0036 & 0.55 & \underline{0.55} \\ 
  10 & 2 & IKN & -0.0493 & 0.39 & 0.45 & -0.0284 & 0.48 & 0.51 & -0.1428 & 0.68 & 0.78 & 0.0086 & 0.56 & 0.57 \\ 
  10 & 2 & FE & -0.0365 & 0.39 & 0.44 & -0.0285 & 0.48 & 0.51 & -0.1426 & 0.66 & \underline{0.76} & 0.0079 & 0.56 & 0.57 \\ 
  10 & 2 & HT & -0.0086 & 0.72 & 0.72 & -0.0015 & 0.36 & \underline{0.36} & -0.0041 & 1.41 & 1.41 & 0.0004 & 0.55 & \underline{0.55} \\ \hline
  50 & 2 & HA & -0.0066 & 0.26 & \underline{0.28} & -0.0023 & 0.26 & 0.26 & -0.0191 & 0.41 & \underline{0.44} & -0.0022 & 0.25 & 0.25 \\ 
  50 & 2 & IKN & -0.0101 & 0.26 & \underline{0.28} & -0.0017 & 0.25 & 0.25 & -0.0341 & 0.42 & 0.46 & -0.0015 & 0.23 & \underline{0.23} \\ 
  50 & 2 & FE & -0.0102 & 0.26 & \underline{0.28} & -0.0017 & 0.25 & 0.25 & -0.0339 & 0.42 & 0.46 & -0.0016 & 0.23 & 0.24 \\ 
  50 & 2 & HT & 0.0008 & 0.38 & 0.38 & -0.0029 & 0.21 & \underline{0.22} & -0.0109 & 0.72 & 0.73 & -0.0096 & 0.36 & 0.38 \\ 
   \hline
\end{tabular}
}
    \label{tab:hajek}
\end{table}

\clearpage

\subsection{Variance estimators and confidence intervals for stratified cRCTs with many small strata}

\subsubsection{Detailed setup}
In this section we detail the setup of the simulation study described in Section \ref{secsimu-small}.
Consider the design with many small strata. Let the number of strata $\B$ be $10$ or $50$ and vary the sizes of each stratum $\np $ to be $2$ or $4$. The units in each stratum are randomly assigned to treatment and control. 
When stratum size is 4, the design is either perfectly balanced in the sense that $\np /2$ clusters are assigned to treatment, or unbalanced where one cluster is randomly selected to treatment in strata $\p=1,2,\ldots,\B/2$ and $\np/2 $ clusters are randomly selected to treatment in strata $\p=\B/2+1,\B/2+2,\ldots,\B$. 
Cluster sizes are still generated from the iid Gamma distribution with shape $4$ and rate $4/30$. The potential outcomes are generated from the model \eqref{eq:potential_gen} but with $t_{\s}  = 5+ \Phi^{-1}(1-\frac{\p}{1+\B})\cdot \alpha + \Phi^{-1}(1-\frac{\s_\p}{1+\np})\cdot \beta$, where $\s_\p$ is the index of unit $\s$ in stratum $\p$. Here, $\alpha$ controls for the treatment effect heterogeneity among strata, and $\beta$ controls for that within strata. When varying $\alpha$, keep $\beta$ fixed at 0, and vice versa. Take $\alpha$ and $\beta$ to be $0, 0.25, 0.5, 0.75, 1$. 

After generating the weights and potential outcomes, they are kept fixed.  The ``mixed" estimator applies $\nu^{(l)}_{\p}$ whenever possible.
The comparison with HC2 is possible because the H\'ajek effect estimator equals the coefficient of $z_i$ in a weighted least squares (WLS) model that regresses $y_{iz_i}$ on $z_i$ weighted by $\wps/\pi_{[i]z_i}$. The HC2 variance estimate for the coefficient of $\zps$ serve as a model-based variance estimation for the H\'ajek effect estimator.

\subsubsection{More results}

Figure \ref{fig:var_sd} shows the standard deviations (SDs) of all variance estimators under different sample sizes and values of $\alpha$ and $\beta$. There are marginal differences between the SDs of the two proposed estimators, as shown in Figure \ref{fig:var_sd}. The estimator $\hat{v}_{\text{small}}(\hat\tau)$ has a slightly higher SD than $\hat{v}_{\text{large}}(\hat\tau)$ in small samples, leading to reduced CI coverage even when their biases are comparable. HC2 has a larger SD compared with the proposed estimators. 

\begin{figure}[!htb]
    \centering
    %\begin{subfigure}{\textwidth}
    %\includegraphics[clip,width=\textwidth]{plots/var_est_bias.pdf}
    %\end{subfigure}
    \begin{subfigure}{\textwidth}
    \includegraphics[clip,width=\textwidth]{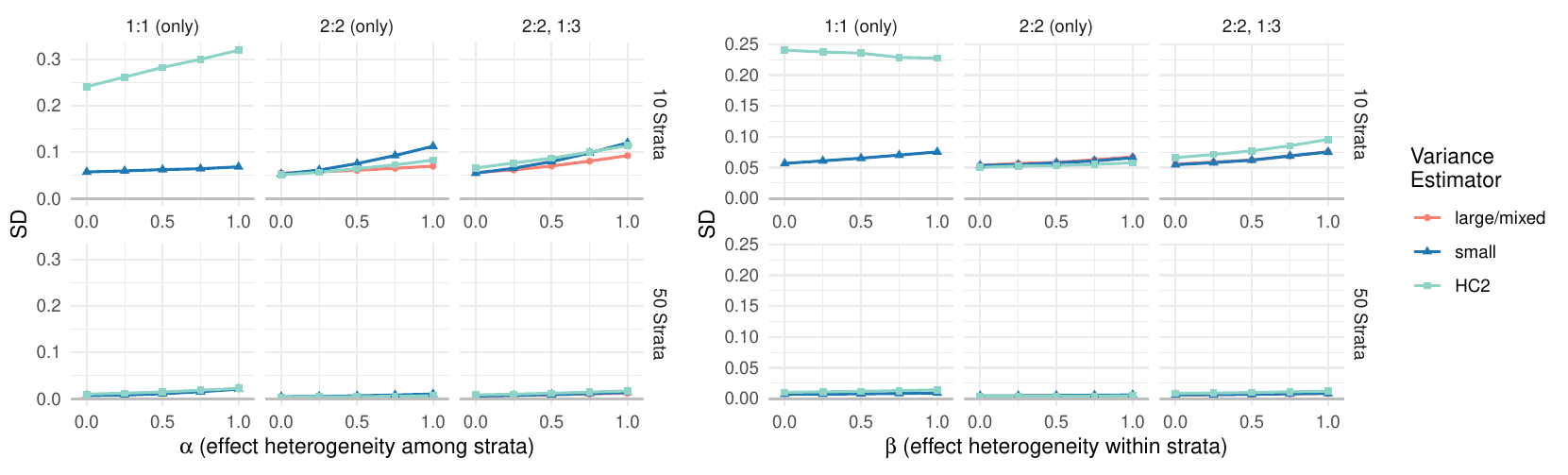}
    \end{subfigure}
    \caption{Standard deviations (SDs) of the variance estimators plotted against $\alpha$ and $\beta$ under different designs and sample sizes. Each column represents a stratum design and each row represents a total number of strata. }
    \label{fig:var_sd}
\end{figure}

Figure \ref{fig:var_cil} shows the mean lengths of the CIs. 
\begin{figure}[!htb]
    \centering
    \includegraphics[clip,width=\textwidth]{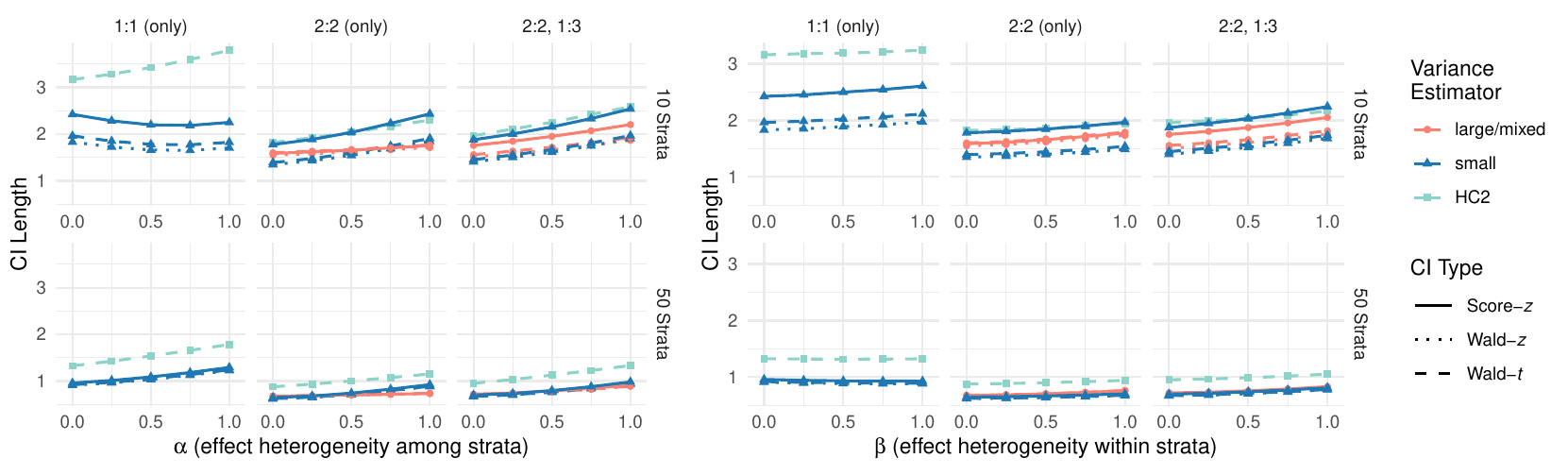}
    \caption{Mean lenghts of different types of 95\% confidence intervals (CIs) plotted against $\alpha$ and $\beta$ under different designs and sample sizes. Each column represents a stratum design and each row represents a total number of strata. }
    \label{fig:var_cil}
\end{figure}
The differences between lengths of these CIs are marginal with a large sample size. The HC2 CIs are generally longer than others, as HC2 is the most conservative when estimating the variance. On the other hand, the Score-based CIs reach the desired level in most cases but are not as long in length.

\subsection{Variance estimators and confidence intervals for stratified cRCTs with a few large strata}

In this section we present more results of the simulation study described in Section \ref{sec:simu-large}.
Figure \ref{fig:c2_sd} shows the relative biases and SDs of the proposed variance estimator $\hat v(\hat\tau)$ and HC2 when there are two large strata of sizes 10, 20, or 50 in the experiment. 
\begin{figure}[htb]
    \centering
    \begin{subfigure}{\textwidth}
    \includegraphics[clip,width=\textwidth]{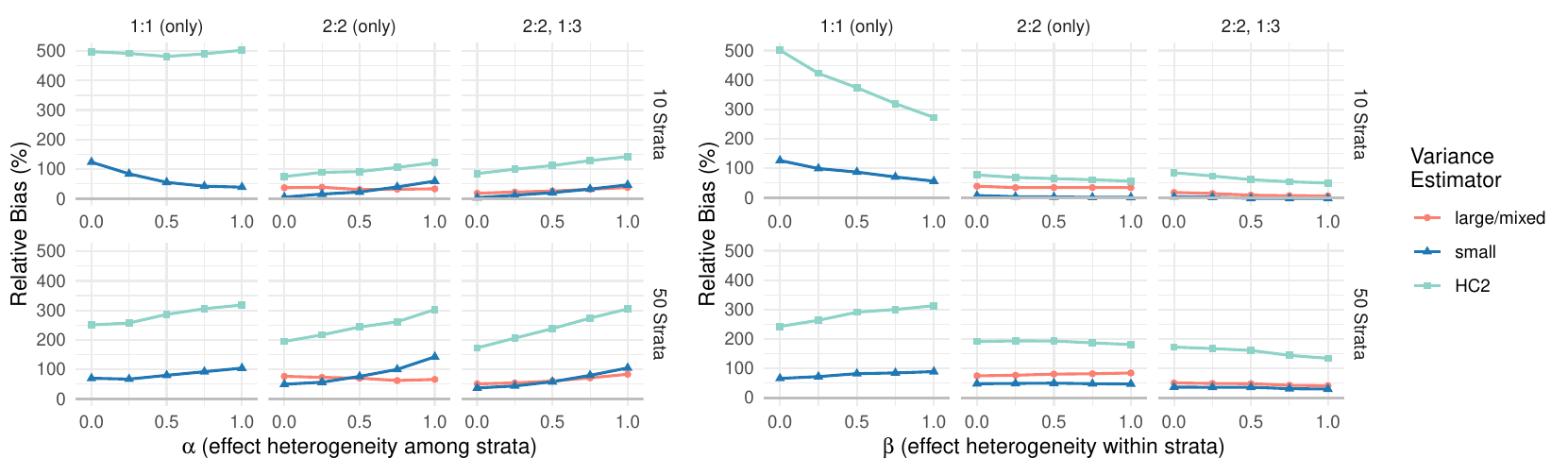}
    \end{subfigure}
    \begin{subfigure}{\textwidth}
    \includegraphics[clip,width=\textwidth]{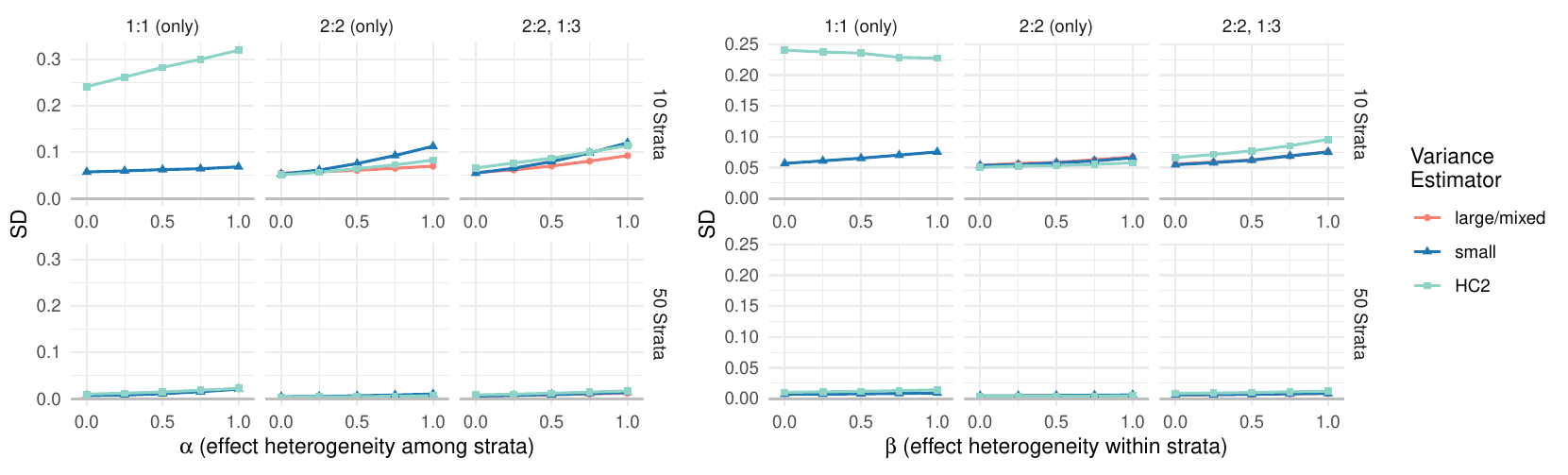}
    \end{subfigure}
    \caption{Relative biases and standard deviations (SDs) of the variance estimators as functions of $\alpha$ and $\beta$ under different designs and sample sizes. Each column represents a stratum size and each row represents a treatment assignment design. }
    \label{fig:c2_sd}
\end{figure}
In general, $\hat v(\hat\tau)$ has a smaller upward bias. Neither of the two estimators exhibit a smaller SD under all the experiment designs. When the design is unbalanced compared to balanced, the estimators' relative biases are similar but their SDs are larger.

Figure \ref{fig:c2_cil} shows the mean lengths of CIs. 
\begin{figure}[ht]
    \centering
    %\begin{subfigure}{\textwidth}
    %\includegraphics[clip,width=\textwidth]{plots/ci_cvg_case2.pdf}
    %\end{subfigure}
    %\begin{subfigure}{\textwidth}
    \includegraphics[clip,width=\textwidth]{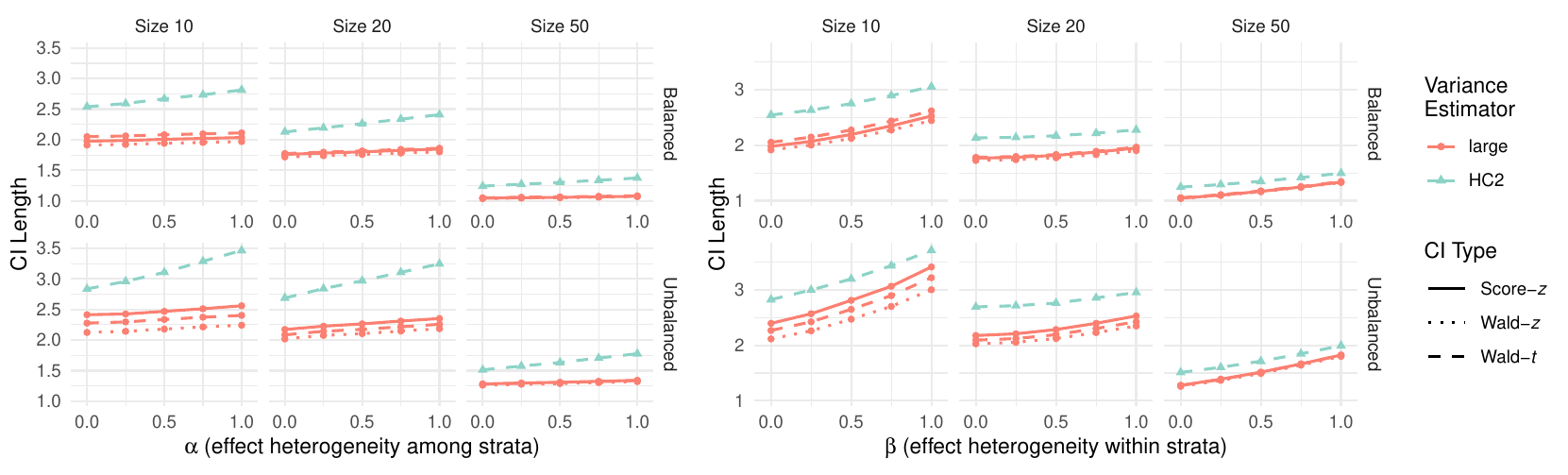}
    %\end{subfigure}
    \caption{Mean lengths of different types of 95\% confidence intervals plotted against $\alpha$ and $\beta$ under different designs and sample sizes. Each column represents a stratum size and each row represents a treatment assignment design. }
    \label{fig:c2_cil}
\end{figure}
HC2 produces longer Wald-type CIs, which is consistent with its upward bias in the variance estimate. When the stratum size is less than or equal to 20 and the treatment assignment is unbalanced, the Score-based test provides a better coverage rate without a largely increased length compared to other Wald-type CIs generated by $\hat v(\hat\tau)$. This makes it a valuable complement to the Wald-type CIs in such challenging cases.

\newpage
\section{Additional details and results for empirical examples}
\label{sec:sm-examples}
\subsection{Variance estimation calibrated to the OSNAP design}
\label{sec:SM-var-OSNAP}

Besides the potential outcome imputation scheme in the previous section, we also consider a linear model that better fits the data to impute the counterfactuals. Forward-selected via the Akaike information criterion (AIC), this model includes treatment, number of staff, percentage of boys, and pre-treatment site-total water consumption. We use model predictions to impute the counterfactuals. For each of the $2^{10}$ possible treatment assignments, we calculate the H\'ajek effect estimator, variance estimators, and 95\% CIs. 

Biases and SDs of variance estimators are shown in Table \ref{tab:water-var1}. Both variance estimators are overestimating the variance, but our proposed estimator has a smaller bias and SD in both cases.
\begin{table}[ht]
  \caption{Biases and SDs of our proposed variance estimator $\hat{v}(\hat\tau)$ and the HC2 estimator under different outcome imputations.}
  \label{tab:water-var1}
  \centering
  %------------ first half -----------------
  \begin{minipage}[t]{0.45\linewidth}
    \centering
    \text{(a) Constant site-total effect imputation}\\[2pt]
    \resizebox{\linewidth}{!}{%
      \begin{tabular}{ccc}\hline
        Estimator & Bias (Relative \%) & SD\\\hline
        $\hat{v}(\hat\tau)$ & 0.00010 (136) & 0.00005\\
        HC2                & 0.00021 (281) & 0.00007\\\hline
      \end{tabular}}
  \end{minipage}
  \hfill
  %------------ second half ----------------
  \begin{minipage}[t]{0.45\linewidth}
    \centering
    \text{(b) Linear model imputation}\\[2pt]
    \resizebox{\linewidth}{!}{%
      \begin{tabular}{ccc}\hline
        Estimator & Bias (Relative \%) & SD\\\hline
        $\hat{v}(\hat\tau)$ & 0.00007 (46) & 0.00009\\
        HC2                & 0.00028 (182) & 0.00011\\\hline
      \end{tabular}}
  \end{minipage}
\end{table}

The mean lengths and empirical coverage rates of 95\% CIs are shown in Table \ref{tab:water-ci1}. In both cases, the estimators are overestimating the variance, so all CIs have coverage rates above 95\%. However, our proposed CIs are less conservative.

\begin{table}[ht]
    \caption{Mean lengths and empirical coverage probabilities of 95\% confidence intervals under different outcome imputations. }
    \label{tab:water-ci1}
    \begin{subtable}{.5\linewidth}
        \caption{A constant site-total effect }
        \centering
        \resizebox{\columnwidth}{!}{
        \begin{tabular}{cccc} \hline
         Method  & Mean length & Coverage (\%) \\
        \hline 
         Wald-$z$ with $\hat{v}(\hat\tau)$  & 0.052 & 99.32 \\
         Wald-$t$ with $\hat{v}(\hat\tau)$   & 0.056 & 99.61 \\ 
         Score-$z$ with $\hat{v}_{H}(\hat\tau)$ & 0.090 & 100 \\
         Wald-$t$ with HC2  & 0.071 & 100 \\ \hline
        \end{tabular}
        }
    \end{subtable}
    \begin{subtable}{.5\linewidth}
        \caption{A linear model}
        \centering
        \resizebox{\columnwidth}{!}{
        \begin{tabular}{cccc} \hline
         Method  & Mean length & Coverage (\%) \\
        \hline 
         Wald-$z$ with $\hat{v}(\hat\tau)$  & 0.057 & 97.27 \\
         Wald-$t$ with $\hat{v}(\hat\tau)$   & 0.062 & 98.14 \\ 
         Score-$z$ with $\hat{v}_{H}(\hat\tau)$ & 0.092 & 100 \\
         Wald-$t$ with HC2  & 0.086 & 100 \\ \hline
        \end{tabular}
        }
    \end{subtable}
\end{table}

\subsection{A study of student spatial thinking}
\label{sec:sm-minecraft-eg}

In this section, we evaluate the covariate adjustment method with another empirical example.
Spatial thinking is crucial in STEM learning, yet sometimes neglected in traditional classrooms. \cite{SLATTERY2024} evaluates the impact of a 6-week Minecraft: Education Edition intervention on spatial thinking among 5th and 6th graders. 
The sample consists of 31 classes (22-32 students each), stratified by classroom gender (mixed or all girls), grade, and educational disadvantage status into five strata (2-9 classes each). Within each stratum, about half of the classes are randomly chosen for the intervention. The outcome is the change in Spatial Reasoning Instrument (SRI) scores from pre- to post-treatment. Student-level data is collected, and we adjust for a person- and a cluster-level covariate, namely pre-treatment SRI score and urban versus rural school location.

\subsubsection{Data analysis}
Applied to these data, the covariate-adjusted H\'ajek technique estimates the SATE as 0.38, with a standard error of 0.30. The score-based CI yields $(-0.21,1.05)$, and the Wald CI with $n-4=27$ degrees of freedom (accounting for four parameters: $\hat\rho_{adj,z}$ for $z=1,0$ and $\hat\beta$ for two covariates.) yields $(-0.24,1.01)$. 
Neither interval excludes zero, in accordance with the findings of \cite{SLATTERY2024}. Results are also summarized in Table \ref{tab:spatial-pointest}.
\begin{table}[ht]
    \caption{The covariate-adjusted H\'ajek effect estimator $\hat\tau_{\mathrm{adj}}$, standard error (SE) estimate, and 95\% confidence intervals (CIs) derived from a $t$-test with $\hat{v}(\hat\tau)$ and a $z$-test with $\hat{v}_{H}(\hat\tau)$. }
    \label{tab:spatial-pointest}
    \centering
    \begin{tabular}{cccc}
    \hline
        $\hat\tau_{\mathrm{adj}}$ & SE & $t$-test with $\hat{v}(\hat\tau_{\mathrm{adj}})$ & $z$-test with $\hat{v}_{H}(\hat\tau_{\mathrm{adj}})$ \\
        0.38 & 0.30 % 0.099
        & $(-0.24,1.00)$ & $(-0.12,1.07)$ \\        \hline
    \end{tabular}
\end{table}

\subsubsection{Simulation of the covariate adjustment method based on the data}
\label{sec:minecraft-simu1}
We also conduct a simulation study. Impute the counterfactuals with predictions from a linear model determined via AIC forward selection, including the treatment, pre-treatment SRI score, student age, and a binary class gender. We then repeatedly randomize treatment assignments and calculate the covariate-adjusted H\'ajek effect estimate $\hat\tau_{\mathrm{adj}}$ and its variance estimate $\hat v(\hat\tau_{\mathrm{adj}})$. We also include a model-based alternative, CR2, a cluster-robust variance estimator for the linear, independent disturbances model of \eqref{eq:covadj-wls}, and construct its CI using a $t$-distribution with adjusted degrees of freedom (23.6) following \cite{McCaffrey2002Bias} and \cite{pustejovsky2017small}. 

Table \ref{tab:spatial-var} presents the biases, relative biases and SDs of the variance estimators. Table \ref{tab:spatial-ci} shows the mean lengths and empirical coverage rates of the CIs. The variance estimator exhibits less upward bias and a smaller SD (34\% and 0.018) compared to CR2 (114\% and 0.025), as shown in Table \ref{tab:spatial-var}. Coverage rates exceed 95\% for all CIs (Table \ref{tab:spatial-ci}). The two $z$-test produces shorter CIs than those from the $t$-tests. These simulation results validate the usage of the proposed method for covariate adjustments.

\begin{table}[ht]
    \caption{Biases and standard deviations (SDs) of the variance estimators. }
    \label{tab:spatial-var}
    \centering
        \begin{tabular}{cccc} \hline
        Estimator & Bias (Relative \%) & SD \\
        \hline 
        $\hat{v}(\hat\tau_{adj})$ & 0.011 (34.1) & 0.018 \\ 
        CR2 & 0.036 (113.6) & 0.025  \\ \hline
        \end{tabular}
\end{table}
\begin{table}[ht]
    \caption{Mean lengths and empirical coverage probabilities of 95\% confidence intervals. }
    \label{tab:spatial-ci}
    \centering
        \begin{tabular}{cccc} \hline
         Method  & Mean length & Coverage (\%) \\
        \hline 
         Wald-$z$ with $\hat{v}(\hat\tau_{adj})$  & 0.79 & 96.23 \\
         Wald-$t$ with $\hat{v}(\hat\tau_{adj})$   & 0.83 & 97.14 \\
         Score-$z$ with $\hat{v}_{H}(\hat\tau_{adj})$ & 0.82 & 96.23 \\
         Wald-$t$ with CR2 & 1.06 & 99.86 \\ \hline
        \end{tabular}
\end{table}

\subsubsection{Simulation without covariate adjustment based on the data}
We also evaluate the proposed method without covariate adjustment using the Minecraft study. We impute the counterfactual potential outcomes with two models. The first model assumes a constant effect of 0.48 at student level. This number is chosen because the H\'ajek effect estimator is calculated to be 0.48 for this study. The second model is a linear model including the treatment indicator, student age, educational disadvantage status, and an intercept. Note that fitting a WLS model with student-level observations enable the calculation of CR2, a cluster robust variance estimator. We are able to compare our variance estimator with CR2 because the estimated slope of a WLS regression of student-level outcome on binary treatment indicator and an intercept weighted by inverse assignment probability is equivalent to the H\'ajek effect estimator. Thus, we can compare our variance estimate with the CR2 variance estimate of the treatment coefficient. The CI associated with CR2 is still calculated using a $t$-distribution with adjusted degrees of freedom (calculated to be 26.5) following \cite{McCaffrey2002Bias} and \cite{pustejovsky2017small}. The other two $t$-tests has $n-2=29$ degrees of freedom.

Results are shown in Tables \ref{tabS:spatial-var} and \ref{tabS:spatial-ci}. 
\begin{table}[ht]
    \caption{Biases and standard deviations (SDs) of our proposed variance estimator $\hat{v}(\hat\tau)$, HC2, and CR2 variance estimator under different outcome imputation schemes. }
    \begin{subtable}{.5\linewidth}
      \centering
      \resizebox{0.9\columnwidth}{!}{
        \begin{tabular}{cccc} \hline
         Estimator & Bias (Relative \%) & SD \\
         \hline 
         $\hat{v}(\hat\tau)$ & -0.0078 (-7.8) & 0.0160 \\
         HC2 & 0.0083 (8.2) & 0.0070 \\
        CR2 & 0.0089 (8.8) & 0.0071  \\ \hline
        \end{tabular}
        }
        \caption{Outcome imputed by a constant effect}
    \end{subtable}%
    \begin{subtable}{.5\linewidth}
      \centering
      \resizebox{0.9\columnwidth}{!}{
        \begin{tabular}{ccc} \hline
         Estimator & Bias (Relative \%) & SD \\
         \hline 
         $\hat{v}(\hat\tau)$ & 0.0141 (50.4) & 0.0195 \\
         HC2& 0.0315 (112.5) & 0.0270 \\
         CR2 & 0.0318 (113.6) & 0.0274 \\ \hline
        \end{tabular}
        }
        \caption{Outcome imputed by a linear model}
    \end{subtable} 
    \label{tabS:spatial-var}
\end{table}
\begin{table}[ht]
    \caption{Mean lengths and empirical coverage probabilities of 95\% confidence intervals under different outcome imputation schemes. }
    \begin{subtable}{.5\linewidth}
        \centering
        \resizebox{\columnwidth}{!}{
        \begin{tabular}{cccc} \hline
         Method  & Mean length & Coverage (\%) \\
        \hline 
         Wald-$z$ with $\hat{v}(\hat\tau)$  & 1.19 & 93.29 \\
         Wald-$t$ with $\hat{v}(\hat\tau)$   & 1.24 & 94.31 \\ 
         Score-$z$ with $\hat{v}_{H}(\hat\tau)$ & 1.20 & 94.52 \\
         Wald-$t$ with HC2  & 1.35 & 95.92 \\
         Wald-$t$ with CR2 & 1.35 & 95.94 \\ \hline
        \end{tabular}
        }
        \caption{Outcome imputed by a constant effect }
    \end{subtable}
    \begin{subtable}{.5\linewidth}
        \centering
        \resizebox{\columnwidth}{!}{
        \begin{tabular}{cccc} \hline
         Method  & Mean length & Coverage (\%) \\
        \hline 
         Wald-$z$ with $\hat{v}(\hat\tau)$  & 0.78 & 95.50 \\
         Wald-$t$ with $\hat{v}(\hat\tau)$   & 0.82 & 96.23 \\ 
         Score-$z$ with $\hat{v}_{H}(\hat\tau)$ & 0.79 & 97.28 \\
         Wald-$t$ with HC2  & 0.97 & 99.01 \\ 
         Wald-$t$ with CR2 & 0.97 & 99.01 \\ \hline
        \end{tabular}
        }
        \caption{Outcome imputed by a linear model}
    \end{subtable}
    \label{tabS:spatial-ci}
\end{table}
When the treatment effect is constant, both HC2 and CR2 are computed under a correctly specified model, and they achieve a coverage around 95\%. Our variance estimator is underestimating the variability of the SATE estimate, resulting in a slightly lower coverage rate of the CIs, but the Score-based test still has a coverage close to 95\%. On the other hand, when the outcomes are imputed by a linear model, our variance estimator has the smaller bias and SD, and its associated CIs also effectively cover the SATE. Overall, the Score-based test produces CIs with fairly good coverage. 

\section*{Data and Software Availability}
The Minecraft data example in this Supplementary Material is publicly available at \url{https://osf.io/buhnr/}. The proposed design-based variance estimators can be implemented using the publicly available \texttt{propertee} R software, available on GitHub at \url{https://github.com/benbhansen-stats/propertee} and on CRAN at \url{https://cran.r-project.org/package=propertee}.

\bibliographystyle{agsm}

\bibliography{causal}
\makeatletter\@input{xx-ms.tex}\makeatother
\makeatletter\@input{xx-sup.tex}\makeatother